\documentclass[11pt]{article}
\usepackage{fullpage}

\usepackage{dsfont}
\usepackage{amsmath}
\usepackage{amsthm}
\usepackage{thmtools} 
\usepackage{thm-restate}

\usepackage{amssymb}
\usepackage{graphicx}
\usepackage{url}
\usepackage{nicefrac}
\usepackage{enumitem}
\usepackage{authblk}
\usepackage{tikz} 
\usepackage{subcaption}
\usepackage{xspace}
\usepackage[compact]{titlesec}
\usetikzlibrary{arrows,decorations.markings}
\usepackage[normalem]{ulem}

\usepackage{xcolor}
\usepackage[linesnumbered,ruled,vlined]{algorithm2e}

\usepackage{verbatim}
\usepackage[english]{babel}
\usepackage[utf8]{inputenc}

\usepackage{enumitem}
\usepackage{bbm}
\usepackage[noend]{algpseudocode}
\usepackage{url}
\usepackage{color}
\usepackage{lipsum} 
\usepackage{amssymb,amsmath,epsfig}
\usepackage{mathtools}
\usepackage{graphicx}

\usetikzlibrary{calc}
\usetikzlibrary{shapes.geometric}
\usetikzlibrary{arrows}
\usetikzlibrary{decorations.pathreplacing, decorations.pathmorphing}
\usetikzlibrary{patterns}
\usetikzlibrary{backgrounds}
\usetikzlibrary{scopes}
\usetikzlibrary{arrows, automata, positioning}
\usetikzlibrary{decorations.markings}

\theoremstyle{plain}
\newtheorem{theorem}{Theorem}

\newtheorem{lemma}{Lemma}

\newtheorem{observation}{Observation}

\newtheorem{corollary}{Corollary}

\newcommand{\W}{W^*}
\newcommand{\I}{\mathcal{I}}

\newcommand{\bnu}{\bar{\nu}}

\usepackage{hyperref}

\usepackage[colorinlistoftodos,textsize=tiny,textwidth=2cm,color=red!25!white,obeyFinal]{todonotes}

\title{Finding Almost Tight Witness Trees}

\author[1]{Dylan Hyatt-Denesik}
\author[1]{Afrouz Jabal Ameli}
\author[2]{Laura Sanit\`a}
\affil[1]{Eindhoven University of Technology, Eindhoven, The Netherlands, \{d.v.p.hyatt-denesik,a.jabal.ameli\}@tue.nl}
\affil[2] {Bocconi University, Milan, Italy, laura.sanita@unibocconi.it}
\date{}

\begin{document}
%
%
%

%
\maketitle              
\begin{abstract}
    This paper addresses a graph optimization problem, called the Witness Tree problem, which seeks a spanning tree of a graph minimizing a certain non-linear objective function. This problem is of interest because it plays a crucial role in the analysis of the best approximation algorithms for two fundamental network design problems: Steiner Tree and  Node-Tree Augmentation. We will show how a wiser choice of witness trees leads to an improved approximation for Node-Tree Augmentation, and for Steiner Tree in special classes of graphs.
\end{abstract}
\section{ Introduction}
Network connectivity problems play a central role in combinatorial optimization. As a general goal, one would like to design a cheap network able to satisfy some connectivity requirements among its nodes. 
Two of the most fundamental problems in this area are \emph{Steiner Tree} and \emph{Connectivity Augmentation}. 

Given a network $G=(V,E)$ with edge costs, and a subset of terminals $R \subseteq V$, Steiner Tree asks to compute a minimum-cost tree $T$ of $G$ connecting the terminals in $R$. In Connectivity Augmentation, we are instead given a $k$-edge-connected graph $G=(V,E)$ and an additional set of edges  $L \subseteq V \times V$ (called \emph{links}). The goal is to add a minimum-cardinality subset of links to $G$ to make it $(k+1)$-edge-connected. 
It is well-known that the problem for odd $k$ reduces to $k=1$ (called Tree Augmentation), and for even $k$ reduces to $k=2$ (called Cactus Augmentation) (see \cite{dinitz1976structure}). All these problems are NP-hard, but admit a constant factor approximation. In the past 10 years, there have been several exciting breakthrough results in the approximation community on these fundamental problems (see~\cite{DBLP:journals/jacm/ByrkaGRS13} \cite{10.1145/2213977.2214081} \cite{DBLP:conf/stoc/Byrka0A20} \cite{nutov20202nodeconnectivity} \cite{DBLP:journals/corr/abs-2009-13257} \cite{DBLP:journals/corr/abs-2012-00086} \cite{traub2022better} \cite{DBLP:conf/stoc/0001KZ18} \cite{DBLP:journals/talg/Adjiashvili19} \cite{DBLP:journals/algorithmica/CheriyanG18} \cite{DBLP:journals/algorithmica/CheriyanG18a} \cite{DBLP:conf/soda/Fiorini0KS18} 
\cite{angelidakis2022node} 
\cite{https://doi.org/10.48550/arxiv.2209.07860}
\cite{traub2022local}). 

Several of these works highlight a deep relation between Steiner Tree and Connectivity Augmentation: the approximation techniques used for Steiner Tree have been proven to be useful for Connectivity Augmentation and vice versa. This fruitful exchange of tools and ideas has often lead to novel results and analyses. This paper continues bringing new ingredients in this active and evolving line of work. 

Specifically, we focus on a graph optimization problem which plays a crucial role in the analysis of some approximation results mentioned before. 
This problem, both in its edge- and node-variant, is centered around the concept of \emph{witness trees}. We now define this formally (see Figure~\ref{fig:witness} for an example).

\smallskip
\noindent
{\bf Edge Witness Tree (EWT) problem}. Given is a tree $T=(V,E)$ with edge costs $c: E \rightarrow \mathbb R_{\ge 0}$. We denote by $R$ the set of leaves of $T$.  The goal is to find a tree  \(W = (R,E_W)\), where $E_W \subseteq R \times R$, which minimizes the non-linear objective function  
\(\bar \nu_T(W)=\frac{1}{c(E)}\sum_{e\in E} c(e)H_{\bar w(e)}\), 
 where \(c(E) = \sum_{e \in E} c(e)\), the function \(\bar w: E \rightarrow \mathbb Z_{\geq 0} \)  is defined as
 \begin{equation*}
    \bar w(e) \coloneqq |\{pq \in E_W: e \textrm{ is an internal edge of the }p \textrm{-}q \textrm{ path in } T\}|
 \end{equation*}
and $H_{\ell}$ denotes the $\ell^{th}$ harmonic number ($H_\ell = 1 + \frac{1}{2} + \frac{1}{3} + \dots + \frac{1}{\ell}$). 
\smallskip

\smallskip
\noindent
{\bf Node Witness Tree (NWT) problem}. Given is a tree $T=(V,E)$. We denote by $R$ the set of leaves of $T$, and $S=V \setminus R$.  The goal is to find a tree  \(W = (R,E_W)\), where $E_W \subseteq R \times R$, which minimizes the non-linear objective function  \(\nu_T(W) = \frac{1}{|S|}\sum_{v\in S} H_{w(v)}\), 
 where \(w: S \rightarrow \mathbb Z_{\geq 0} \) is defined as
 \begin{equation*}
     w(v) \coloneqq |\{pq \in E_W: v \textrm{ is an internal node of the }p \textrm{-}q \textrm{ path in } T\}|
\end{equation*}
and again $H_\ell$ denotes the $\ell^{th}$ harmonic number.
\smallskip

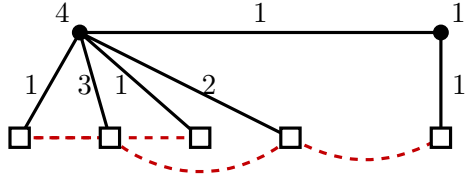
\begin{figure}[t]
    \centering
        \begin{tikzpicture}[scale=0.8]

   \coordinate (u) at  (0,1.75);
    \coordinate (v) at (6,1.75) ;
    \coordinate(uv) at (3,1.75);
    \coordinate (u1) at (-1,0);  
    \coordinate (u2) at (0.5 , 0);
    \coordinate  (u3) at (2, 0);
    \coordinate (u4) at (3.5,0);
    \coordinate  (u5) at (6,0);
    \coordinate(t3) at (2.05,0);
    \coordinate(blueEdge) at (0.08,1.75/2);
    \coordinate (ut4) at (2.15,1.77/2);

    \filldraw[very thick,color=black]
        (v)circle(3pt);


    \draw[very thick, color=black]
        (u) -- (v) node[midway] {}
        (u) -- (u1) node[midway,left] {$1$}
        (u) -- (u3) node[midway,left] {$1$}
       (u) -- (u4) node[midway,right] {}
        (v) -- (u5) node[midway,right] {$1$};
    \draw[very thick, color=black]
        (u) -- (u2) node[midway] {};
     \draw
         (u1) node[anchor= north]{}
         (u2) node[anchor= north]{}
         (t3) node[anchor= west]{}
         (u4) node[anchor= north]{}
         (u5) node[anchor= north]{}
        (uv)  node[anchor= south]{$1$}
        (blueEdge) node {$3$}
        (ut4) node {$2$}
        (u) node[anchor=south east]{$4$}
        (v) node[anchor=south west]{$1$};

     \draw[very thick, color={rgb,255: red,195; green,0; blue,3} , dashed]
        (u3) edge[]node{}  (u1)
        (u1) edge[] node{}  (u2)
        (u4) edge[bend left=40] node{} (u2)
        (u4) edge[bend right=30] node{} (u5);
    \draw[color=black,very thick]
        (u1) node[draw=black, fill=black!0]{}
        (u2) node[draw=black, fill=black!0]{}
        (u3) node[draw=black, fill=black!0]{}
        (u4) node[draw=black, fill=black!0]{}
        (u5) node[draw=black, fill=black!0]{};
    \filldraw[very thick, color=black]
        (u)circle(3pt);
        
\end{tikzpicture}
        \caption{In black, the tree \(T=(R\cup S, E)\). The dashed edges represent a witness tree $W$. The labels on edges of \(E\) and vertices of \(S\) indicate $\bar w(e)$ and  $w(v)$, respectively. We have \( \nu_T(W) = (H_4 + H_1)/2 = 1.541\bar 6\). Assuming unit cost on the edges of $E$, we have \( \bar \nu_T(W) = (4H_1+H_2+H_3)/6 = 1.\bar{2}\).}
    \label{fig:witness}
\end{figure}

\noindent
We refer to a feasible solution $W$ to either of the above problems as a \emph{witness tree}. 
We call $\bar w$ (resp. $w$) the \emph{vector imposed on} \(E\) (resp. \(S\)) by $W$.
We now explain how these problems relate to Steiner Tree and Connectivity Augmentation.

\paragraph*{ EWT and relation to Steiner Tree.} Currently, the best approximation factor for Steiner Tree is ($\ln(4) + \varepsilon$), which can be achieved by three different algorithms~\cite{10.1145/2213977.2214081} \cite{DBLP:journals/jacm/ByrkaGRS13} \cite{traub2022local}. These algorithms yield the same approximation  because in all three of them, the analysis at some point relies on constructing witness trees. 

More in detail, suppose we are given a Steiner Tree instance \((G=(V, E),R,c)\) where $c: E \rightarrow \mathbb R_{\ge 0}$ gives the edge costs. 
We can define the following:
\begin{equation*}
    \gamma_{(G,R,c)} \coloneqq \min_{\substack{T^* = (R \cup S^*, E^*):\; T^* \textrm{ is}\\\textrm{optimal Steiner tree of }(G,R,c)}} \;\;\;  \min_{\substack{W:\; W\textrm{ is a }\\\textrm{witness tree} \\ \textrm{of $T^*$} }} \bar \nu_{T^*}(W) 
\end{equation*}
We also define the following constant \(\gamma\): 
\begin{equation*}
    \gamma \coloneqq \sup \{\gamma_{(G,R,c)}: (G,R,c)\textrm{ is an instance of Steiner Tree}\}.
\end{equation*}

Byrka et al.~\cite{DBLP:journals/jacm/ByrkaGRS13} were the first to essentially prove the following.
\begin{theorem}\label{thm:gamma}
    For any $\varepsilon >0$, there is a $(\gamma + \varepsilon)$-approximation algorithm for Steiner Tree.
\end{theorem}
Furthermore, the authors in~\cite{DBLP:journals/jacm/ByrkaGRS13} showed that $\gamma \leq \ln(4)$, and hence they obtained the previously mentioned $(\ln(4) + \varepsilon)$-approximation for Steiner Tree.

\paragraph*{ NWT and relation to Connectivity Augmentation.} 
 Basavaraju et al~\cite{DBLP:conf/icalp/BasavarajuFGMRS14} introduced an approximation-preserving reduction from Cactus Augmentation (which is the hardest case of Connectivity Augmentation)\footnote{Tree Augmentation can be easily reduced to Cactus Augmentation by introducing a parallel copy of each initial edge.} to special instances of Node-Steiner Tree, named \emph{CA-Node-Steiner-Tree} instances in~\cite{angelidakis2022node}: the goal here is to connect a given set $R$ of terminals of a graph $G$ via a tree that minimizes the number of non-terminal nodes (Steiner nodes) in it. The special instances have the crucial property that each Steiner node is adjacent to at most 2 terminals. 

Byrka et al.~\cite{DBLP:conf/stoc/Byrka0A20} built upon this reduction to prove a 1.91-approximation for CA-Node-Steiner-Tree instances. This way, they were the first to obtain a better-than-2 approximation factor for Cactus Augmentation (and hence, for Connectivity Augmentation). 
Interestingly, Nutov~\cite{nutov20202nodeconnectivity} realized that a similar reduction also captures a fundamental node-connectivity augmentation problem: the \emph{Node-Tree Augmentation} (defined exactly like Tree Augmentation, but replacing edge-connectivity with node-connectivity). This way, he could improve over an easy 2-approximation for Node-Tree Augmentation that was also standing for 40 years~\cite{DBLP:journals/siamcomp/FredericksonJ81}. Angelidakis et al.~\cite{angelidakis2022node} subsequently explicitly formalized the problem at the heart of the approximation analysis: namely, the NWT problem.

More in detail, given a CA-Node-Steiner-Tree instance \((G=(V, E),R)\), we can define the following:
\begin{equation*}
    \psi_{(G,R)} \coloneqq \min_{\substack{T^* = (R \cup S^*, E^*):\; T^* \textrm{ is}\\\textrm{optimal Steiner tree of }(G,R)}} \;\;\; \min_{\substack{W:\; W\textrm{ is a }\\\textrm{witness tree} \\ \textrm{of $T^*$} }}  \nu_{T^*}(W),
\end{equation*}
We also define the  constant \(\psi\):
\begin{equation*}
    \psi \coloneqq \sup \{\psi_{(G,R)}: (G,R) \textrm{ is an instance of CA-Node-Steiner-Tree}\}.
\end{equation*}

Angelidakis et al.~\cite{angelidakis2022node} proved the following.
\begin{theorem}\label{thm:alpha}
    For any $\varepsilon >0$, there is a $(\psi + \varepsilon)$-approximation algorithm for CA-Node-Steiner Tree.
\end{theorem}
Furthermore, the authors of~\cite{angelidakis2022node} proved that $\psi < 1.892$, and hence obtained a $1.892$-approximation algorithm for Cactus Augmentation and Node-Tree Augmentation. This is currently the best approximation factor known for Node-Tree Augmentation (for Cactus Augmentation there is a better  algorithm~\cite{DBLP:journals/corr/abs-2012-00086}).

\paragraph*{ Our results and techniques.} Our main result is an improved upper bound on $\psi$. In particular, we are able to show $\psi < 1.8596$. Combining this with Theorem~\ref{thm:alpha}, we obtain a 1.8596-approximation algorithm for CA-Node-Steiner-Tree.
Hence, due to the above mentioned reduction, we improve the state-of-the-art approximation for Node-Tree Augmentation.

\begin{theorem}\label{thm:main_CA}
    There is a 1.8596-approximation algorithm for CA-Node-Steiner-Tree (and hence, for Node-Tree Augmentation). 
\end{theorem}


Our result is based on a better construction of witness trees for the NWT problem. At a very high level, the witness tree constructions used previously in the literature use a \emph{marking-and-contraction} approach, that can be summarized as follows. First, root the given tree $T$ at some internal Steiner node. Then, every Steiner node $v$ chooses (\emph{marks}) an edge which connects to one of its children: this identifies a path from $v$ to a terminal. Contracting the edges along this path yields a witness tree $W$. The way this marking choice is made varies: it is random in \cite{DBLP:journals/jacm/ByrkaGRS13}, it is biased depending on the nature of the children in \cite{DBLP:conf/stoc/Byrka0A20}, it is deterministic and taking into account the structure of $T$ in \cite{angelidakis2022node}. However, all such constructions share the fact that decisions can be thought of as being taken ``in one shot'', at the same time for all Steiner nodes. 
Instead, here we consider a bottom-up approach for the construction of our witness tree, where a node takes a marking decision only after the decisions of its children have been made. A sequential approach of this kind allows a node to have a more precise estimate on the impact of its own decision to the overall non-linear objective function cost, but it becomes more challenging to analyze. Overcoming this challenge is the main technical contribution of this work, and the insight behind our improved upper-bound on \(\psi\).

We complement this result with an almost-tight lower-bound on \(\psi\), which improves over a previous lower bound given in~\cite{angelidakis2022node}.
\begin{theorem}
    \label{thm:nodelowerbound}
    For any $\varepsilon > 0$, there exists a CA-Node-Steiner-Tree instance ($G_\varepsilon, R_\varepsilon$) such that $\psi_{(G_\varepsilon,R_\varepsilon)} > 1.841\bar{6} - \varepsilon$.
\end{theorem}
The above theorem implies that, in order to  significantly improve the approximation for Node-Tree Augmentation, very different techniques need to be used. 
To show our lower-bound we prove a structural property on optimal witness trees, called \emph{laminarity}, which in fact holds for optimal solutions of both the NWT problem and the EWT problem. 

As an additional result, we also improve the approximation bound for Steiner Tree in the special case of \emph{Steiner-claw free} instances. 
A Steiner-Claw Free instance is a Steiner-Tree instance where the subgraph \(G[V\setminus R]\) induced by the Steiner nodes is claw-free (i.e., every node has degree at most 2). These instances were introduced in~\cite{feldmann2016equivalence} 
in the context of studying the integrality gap of a famous LP relaxation for Steiner Tree, called the \emph{bidirected cut relaxation}, that is long-conjectured to have integrality gap strictly smaller than $2$. 

\begin{theorem}
\label{thm:clawupper}
    There is a    \((\frac{991}{732} + \varepsilon < 1.354)\)-approximation for  Steiner Tree on Steiner-claw free instances.
\end{theorem}   

We prove the theorem by showing that, for any Steiner-Claw Free instance \((G, R, c)\),  $\gamma_{(G,R,c)} \leq \frac{991}{732}$. 
The observation we use here is that an optimal Steiner Tree solution $T$ in this case is the union of components that are caterpillar graphs\footnote{{\color{black}A caterpillar graph is  defined as a tree in which every leaf is of distance $1$ from a central path. }}: this knowledge can be exploited to design ad-hoc witness trees. Interestingly, we can also show that this bound is tight: once again, the proof of this lower-bound result relies on showing laminarity for optimal witness trees.

\begin{theorem}
\label{thm:clawlowerbound}
    For any $\varepsilon > 0$, there exists Steiner-Claw Free instance \((G_\varepsilon, R_\varepsilon, c_\varepsilon)\) such that \\$\gamma_{(G_\varepsilon,R_\varepsilon,c_\varepsilon)} > \frac{991}{732} - \varepsilon$.

\end{theorem}

 As a corollary of our results, we also get an improved bound on the integrality gap of the bidirected cut relaxation for Steiner-Claw Free instances (this follows directly from combining our upper bound with the results in~\cite{feldmann2016equivalence}). 
    Though these instances are quite specialized, they serve the purpose of passing the message: exploiting the structure of optimal solutions helps in choosing better witnesses, hopefully arriving at tight (upper and lower) bounds on $\gamma$ and $\psi$.

\section{ Laminarity}
\label{sec:laminarity}
In this section, we prove some key structural properties of witness trees. We assume to be given a Node (Edge) Witness Tree instance $T = (V,E)$ with leaves $R$ (and edge costs $c : E \rightarrow \mathbb{R}_{\geq 0}$), where $R$ denotes the leaves of $T$, we will show that we can characterize witness trees minimizing \(\nu_T(W)\)  (\(\bnu_T(W)\)) using the following notion of \emph{laminarity}. 
Given a witness tree \(W = (R, E_W)\), we say edges  \(f_1f_2,f_3f_4\in E_W\) \emph{cross} if the \(f_1\textrm{-}f_2\) and \(f_3\textrm{-}f_4\) paths in \(T\) share an internal node but not an endpoint. We say that \(W\) is \emph{laminar} if it has no crossing edges. For nodes $u,v\in V$, we denote by $T_{uv}$ the path in $T$ between the nodes $u$ and $v$. Similarly, for $e\in E_W$, we denote by $T_{e}$ the path in $T$ between the endpoints of $e$.

The following Theorem shows that there is always a witness tree minimizing \(\nu_T(W)\) that is laminar.

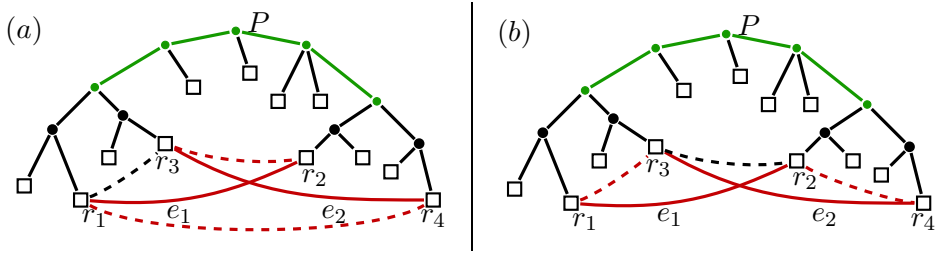
\begin{figure}[t]
    \begin{center}
        \begin{tabular}{c|c}
            \begin{tikzpicture}[scale=0.75]
                
                \tikzset{black dot/.style={draw=black, very thick, circle,minimum size=0pt, inner sep=1pt, outer sep=1pt,fill=black}}
                \tikzset{terminal/.style={draw=black,  thick,minimum size=0pt, inner sep=2.5pt, outer sep=1pt}}
                \tikzset{P node/.style={fill={rgb,255: red,20; green,154; blue,0}, draw={rgb,255: red,20; green,154; blue,0}, circle, minimum size=0pt,inner sep=1pt, outer sep=1pt}}
            
                \tikzstyle{witness edge}=[-, draw={rgb,255: red,195; green,0; blue,3}, very thick]
                \tikzstyle{T edges}=[-, very thick]
                \tikzstyle{new witness}=[-, draw={rgb,255: red,195; green,0; blue,3}, dashed, very thick]
                \tikzstyle{connected terminals}=[-, draw=black, dashed, very thick]
                \tikzstyle{P}=[-, draw={rgb,255: red,20; green,154; blue,0}, very thick]

    		\node [style=P node] (1) at (-9.5, 3.25) {};
    		\node [style=terminal] (2) at (-9, 2.5) {};
    		\node [style=P node] (3) at (-8.25, 3.5) {};
    		\node [style=terminal] (5) at (-7.5, 2.25) {};
    		\node [style=P node] (6) at (-7, 3.25) {};
    		\node [style=P node] (7) at (-5.75, 2.25) {};
    		\node [style=black dot] (8) at (-6.5, 1.75) {};
    		\node [style=black dot] (9) at (-5, 1.5) {};
    		\node [style=terminal] (10) at (-5.5, 1) {};
    		\node [style=terminal] (11) at (-4.75, 0.5) {};
    		\node [style=terminal] (12) at (-7, 1.25) {};
    		\node [style=P node] (13) at (-10.75, 2.5) {};
    		\node [style=black dot] (14) at (-10.25, 2) {};
    		\node [style=terminal] (15) at (-9.5, 1.5) {};
    		\node [style=black dot] (16) at (-11.5, 1.75) {};
    		\node [style=terminal] (17) at (-12, 0.75) {};
    		\node [style=terminal] (18) at (-11, 0.5) {};
    		\node [style=terminal] (19) at (-6, 1.25) {};
    		\node [style=terminal] (20) at (-10.5, 1.25) {};
    		\node [style=terminal] (21) at (-6.75, 2.25) {};
    		\node  (22) at (-9.25, 0.25) {$e_1$};
    		\node  (23) at (-6.5, 0.25) {$e_2$};
    		\node  (24) at (-10.75, 0.175) {$r_1$};
    		\node  (25) at (-6.85, 0.9) {$r_2$};
    		\node  (26) at (-9.45, 1.125) {$r_3$};
    		\node  (27) at (-4.75, 0.175) {$r_4$};
    		\node  (29) at (-7.85, 3.65) {{$P$}};
    		\node [style=terminal] (26) at (-8, 2.75) {};
      		\node (31) at (-12, 3.5) {$(a)$};
            
    		\draw [style=witness edge, bend right=15] (18) to (12);
    		\draw [style=witness edge, in=180, out=-30] (15) to (11);
    		\draw [style=P] (1) to (13);
    		\draw [style=T edges] (13) to (14);
    		\draw [style=T edges] (14) to (20);
    		\draw [style=T edges] (14) to (15);
    		\draw [style=T edges] (13) to (16);
    		\draw [style=T edges] (16) to (17);
    		\draw [style=T edges] (16) to (18);
    		\draw [style=T edges] (1) to (2);
    		\draw [style=P] (1) to (3);
    		\draw [style=P] (3) to (6);
    		\draw [style=T edges] (6) to (5);
    		\draw [style=T edges] (6) to (21);
    		\draw [style=P] (6) to (7);
    		\draw [style=T edges] (7) to (8);
    		\draw [style=T edges] (7) to (9);
    		\draw [style=T edges] (9) to (10);
    		\draw [style=T edges] (9) to (11);
    		\draw [style=T edges] (8) to (19);
    		\draw [style=T edges] (8) to (12);
    		\draw [style=new witness, bend right=15, looseness=0.75] (15) to (12);
    		\draw [style=connected terminals, bend right=15, looseness=0.50] (18) to (15);
    		\draw [style=new witness, bend right, looseness=0.50] (18) to (11);
                
		\draw [style=T edges] (3) to (26);
            \end{tikzpicture}
            &
            \begin{tikzpicture}[scale=0.75]
                
                \tikzset{black dot/.style={draw=black, very thick, circle,minimum size=0pt, inner sep=1pt, outer sep=1pt,fill=black}}
                \tikzset{terminal/.style={draw=black,  thick,minimum size=0pt, inner sep=2.5pt, outer sep=1pt}}
                \tikzset{P node/.style={fill={rgb,255: red,20; green,154; blue,0}, draw={rgb,255: red,20; green,154; blue,0}, circle, minimum size=0pt,inner sep=1pt, outer sep=1pt}}
            
                \tikzstyle{witness edge}=[-, draw={rgb,255: red,195; green,0; blue,3}, very thick]
                \tikzstyle{T edges}=[-, very thick]
                \tikzstyle{new witness}=[-, draw={rgb,255: red,195; green,0; blue,3}, dashed, very thick]
                \tikzstyle{connected terminals}=[-, draw=black, dashed, very thick]
                \tikzstyle{P}=[-, draw={rgb,255: red,20; green,154; blue,0}, very thick]
                
                \node [style=P node] (1) at (-9.5, 3.25) {};
                \node [style=terminal] (2) at (-9, 2.5) {};
                \node [style=P node] (3) at (-8.25, 3.5) {};
                \node [style=terminal] (5) at (-7.5, 2.25) {};
                \node [style=P node] (6) at (-7, 3.25) {};
                \node [style=P node] (7) at (-5.75, 2.25) {};
                \node [style=black dot] (8) at (-6.5, 1.75) {};
                \node [style=black dot] (9) at (-5, 1.5) {};
                \node [style=terminal] (10) at (-5.5, 1) {};
                \node [style=terminal] (11) at (-4.75, 0.5) {};
                \node [style=terminal] (12) at (-7, 1.25) {};
                \node [style=P node] (13) at (-10.75, 2.5) {};
                \node [style=black dot] (14) at (-10.25, 2) {};
                \node [style=terminal] (15) at (-9.5, 1.5) {};
                \node [style=black dot] (16) at (-11.5, 1.75) {};
                \node [style=terminal] (17) at (-12, 0.75) {};
                \node [style=terminal] (18) at (-11, 0.5) {};
                \node [style=terminal] (19) at (-6, 1.25) {};
                \node [style=terminal] (20) at (-10.5, 1.25) {};
                \node [style=terminal] (21) at (-6.75, 2.25) {};
                \node (22) at (-9.25, 0.25) {$e_1$};
                \node (23) at (-6.5, 0.25) {$e_2$};
                \node (24) at (-10.75, 0.175) {$r_1$};
                \node (25) at (-6.85, 0.9) {$r_2$};
                \node (26) at (-9.45, 1.125) {$r_3$};
                \node (27) at (-4.75, 0.175) {$r_4$};
                \node (29) at (-7.85, 3.65) {{$P$}};
    		\node [style=terminal] (26) at (-8, 2.75) {};
                \node (31) at (-12, 3.5) {$(b)$};
                
                \draw [style=witness edge, bend right=15] (18) to (12);
                \draw [style=witness edge, in=180, out=-30] (15) to (11);
                \draw [style=P] (1) to (13);
                \draw [style=T edges] (13) to (14);
                \draw [style=T edges] (14) to (20);
                \draw [style=T edges] (14) to (15);
                \draw [style=T edges] (13) to (16);
                \draw [style=T edges] (16) to (17);
                \draw [style=T edges] (16) to (18);
                \draw [style=T edges] (1) to (2);
                \draw [style=P] (1) to (3);
                \draw [style=P] (3) to (6);
                \draw [style=T edges] (6) to (5);
                \draw [style=T edges] (6) to (21);
                \draw [style=P] (6) to (7);
                \draw [style=T edges] (7) to (8);
                \draw [style=T edges] (7) to (9);
                \draw [style=T edges] (9) to (10);
                \draw [style=T edges] (9) to (11);
                \draw [style=T edges] (8) to (19);
                \draw [style=T edges] (8) to (12);
                \draw [style=connected terminals, bend right=15, looseness=0.75] (15) to (12);
                \draw [style=new witness, bend right=15, looseness=0.50] (18) to (15);
                \draw [style=new witness, bend right=15, looseness=0.50] (12) to (11);
                
		      \draw [style=T edges] (3) to (26);
            \end{tikzpicture}

        \end{tabular}
    \end{center}
    \caption{In both figures we have a tree, $T$, shown with black edges and green edges, with leaves, $R$, denoted by squares. Crossing edges $e_1$ and $e_2$ are shown with solid red edges. The green edges denote the path $P$.
    Figure $(a)$: In this case, $r_1$ and $r_3$ are in the same component of $W\backslash \{e_1, e_2\}$, represented by the dashed black edge. We can replace $e_1$ with $r_2r_3$ or replace $e_2$ with $r_1r_4$ (red dashed edges).
    Figure $(b)$: In this case, $r_3$ and $r_2$ are in the same component, denoted by the black dashed edge. We can replace $e_1$ and $e_2$ with $r_1r_3$ and $r_2r_4$ (red dashed edges).
    }
    \label{fig:laminarity}
\end{figure}

\begin{theorem}
\label{thm:nodelaminar} 
    Given an instance of the Node Witness Tree problem $T=(V,E)$, let \(\mathcal{W}\) be the family of all witness trees for $T$. Then there exists a laminar witness tree $W$ such that $\nu_T(W) = \min_{W'\in \mathcal{W}} \nu_T(W')$.
\end{theorem}
\begin{proof}
    We first show that there is a witness tree $W$ minimizing $\nu_T(W)$ such that the induced subgraph of $W$ on any maximal set of terminals that share a neighbour in $V\backslash R$ is a star. We assume for the sake of contradiction that there is a maximal set of terminals $S\subseteq R$ sharing a neighbour $v\in V\backslash R$, such that the induced subgraph of $W$ on $S$ is a set of  connected components $W_1,\dots, W_i$ for $i>1$. 
    Without loss of generality, suppose the shortest path between two components is from $W_1$ to $W_2$, and let $e$ denote the edge of this path incident to $W_2$. We define $W'\coloneqq W\cup \{f\} \backslash\{e\}$, where $f$ is an arbitrary edge between $W_1$ and $W_2$. Since $\{v\} = T_f\backslash R \subsetneq T_e \backslash R$, we have $\nu_T(W') < \nu_T(W)$, contradicting the minimality of $W$. Therefore, the induced subgraph on $S$ is connected. We can rearrange the edges of this subgraph to be a star as this will not affect $\nu_T(W)$, so we assume this holds on $W$ for any such $S$. 
    
    For a maximal set of terminals $S\subseteq R$ that share a neighbour, by a slight abuse of notation, we denote by $S$ the induced star subgraph of $W$ on $S$, and denote its center by $s\in S$. We will assume without loss of generality that edges of $W$ incident to $S$ have endpoint $s$. To see this, as $S$ is a connected subgraph of $W$, any pair of edges incident to $S$ cannot share an endpoint outside of $S$, otherwise we have found a cycle in $W$. 
    Furthermore, for any edge of $W$ incident to $S$ where $s$ is not an endpoint, we can change the endpoint in $S$ of that edge to be $s$ and maintain the connectivity of $W$ since $S$ is connected.
    Edges changed in this way will have the same interior nodes between their endpoints, so this does not increase $\nu_T(W)$.

    We assume for the sake of contradiction that the witness tree $W$ minimizing $\nu_T(W)$ is not a laminar witness tree. As $W$ is not laminar, there exist distinct leaves $r_1,r_2,r_3,r_4 \in R$  such that $e_1=r_1r_2, e_2=r_3r_4\in E_W$ are crossing. We denote the path $T_{e_1}\cap T_{e_2}$ by $P$. We denote  by $P_i$ the (potentially empty) set of internal nodes of the shortest path from $P$ to $r_i$ in $T$. 
    
    Since $e_1$ and $e_2$ are crossing edges, one of $T_{r_1r_3}$ or $T_{r_1r_4}$ contains exactly one node of $P$. The same is true for $r_2$. Without loss of generality, let us assume that the paths $T_{r_1r_3}$ and $T_{r_2r_4}$ contain exactly one node of $P$. We consider by cases which component of $W\backslash \{e_1, e_2\}$ contains two nodes among $r_1, r_2, r_3$ and $r_4$. See Figure~\ref{fig:laminarity} for an example.

    \begin{itemize}
        \item Case: $r_1$ and $r_3$ (or similarly, $r_2$ and $r_4$) are in the same component of $W  \backslash \{e_1,e_2\}$. If $P_1 = P_3 = \emptyset$, then $r_1$ and $r_3$ share a neighbour and thus, as shown above, $e_1$ and $e_2$ are assumed to share an endpoint, and are thus not crossing.
        
        Consider $W' \coloneqq  W\cup \{r_2r_3\}\backslash\{e_1\}$ and $W'' \coloneqq  W\cup \{r_1r_4\} \backslash \{e_2\}$. If $\nu_T(W) - \nu_T(W') > 0$, this contradicts the minimality of $\nu_T(W)$. Therefore, we can see
        \begin{align*}
            0 &\leq  |V\backslash R| (\nu_T(W') - \nu_T(W) ) = \sum_{u\in P_3} \frac{1}{w(u)+1} -  \sum_{u\in P_1} \frac{1}{w(u)} \\
            & < \sum_{u\in P_3} \frac{1}{w(u)} -  \sum_{u\in P_1} \frac{1}{w(u)+1} =|V\backslash R| ( \nu_T(W) - \nu_T(W'') )
        \end{align*}
        Clearly, we have $\nu_T(W'') < \nu_T(W)$, contradicting minimality of $\nu_T(W)$. 

        \item Case: $r_2$ and $r_3$ (or similarly, $r_1$ and $r_4$) are in the same component of $W \backslash\{e_1,e_2\}$. Without loss of generality we can assume that $|V(P)| > 1$, because if $|V(P)| =1$ then we can reduce to the previous case by relabelling the nodes $r_1, r_2, r_3$ and $r_4$.  In this case, consider $W' \coloneqq W\cup\{r_1r_3,r_2r_4\}\setminus\{e_1,e_2\}$. Therefore, we can see
        \begin{align*}
            |V\backslash R|\left(\nu_T(W') - \nu_T(W) \right) \le - \sum_{u\in P} \frac{1}{w(u)} < 0
        \end{align*}
    \end{itemize}
    Thus, we have $\nu_T(W') < \nu_T(W)$, contradicting the minimality of  $\nu_T(W)$.
\end{proof}
The following theorem, similar to Theorem~\ref{thm:nodelaminar}, shows that there are laminar witness trees that are optimal for the EWT problem. The proof is deferred to 
{\color{black}the full version of the paper.}
\begin{restatable}{theorem}{edgelaminar}
\label{thm:edgelaminar}
    Given an instance of the Edge Witness Tree problem  $T=(V,E)$ with edge costs $c$, let \(\mathcal{W}\) be the family of all witness trees for $T$. Then there exists a laminar witness tree $W$ such that $\bnu_T(W) = \min_{W'\in \mathcal{W}} \bnu_T(W')$.
\end{restatable}

We now show that laminar witness trees are precisely the set of trees that one could obtain with a marking-and-contraction approach.  
The proof of this Theorem can be found in 
{\color{black}the full version of the paper.}

\begin{theorem}
\label{thm:contractlam}
    Given a tree $T = (V,E)$ with leaves $R$, a witness tree $W = (R,E_W)$ for $T$ can be found by marking-and-contraction if and only if $W$ is laminar.
\end{theorem}

Incidentally, this has the following side implication. The authors of~\cite{10.1145/2213977.2214081} gave a dynamic program (that is also a bottom-up approach) to compute the best possible witness tree obtainable with a marking-and-contraction scheme. Our structural results imply that their dynamic program computes an optimal solution for the EWT problem (though for the purpose of the approximation analysis, being able to compute the best witness tree is not that relevant: being able to bound \(\psi\) and \(\gamma\) is what matters). 
\section{ Improved approximation for CA-Node-Steiner Tree}\label{sec:approximationCA-Node}
The goal of this section is to prove Theorem~\ref{thm:main_CA}. We will achieve this by showing
$\psi<1.8596$, and by using Theorem~\ref{thm:alpha}.
From now on, we assume we are given a tree \(T=(R\cup S^*, E^*)\), where each Steiner node is adjacent to at most two terminals. 

\subsection{ Preprocessing.} We first apply some preprocessing operations as in \cite{angelidakis2022node}, that allow us to simplify our witness tree construction. The first one is to remove the terminals from $T$, and then decompose $T$ into smaller components which will be held separately.
We start by defining a \emph{final} Steiner node as a Steiner node that is adjacent to at least one terminal. We let \(F\subseteq S^*\) denote the set of final Steiner nodes. Since we remove the terminals from $T$, we will construct a spanning tree $W$ on \(F\)  with edges in \(F \times F\). With a slight abuse of notation, we refer to $W$ as a witness tree: this is because \cite[Section 4.1]{angelidakis2022node} showed that one can easily map $W$ to a witness tree for our initial tree $T$ (with terminals put back),
and the following can be considered the vector imposed on \(S^*\) by $W$:
\begin{equation}\label{eq:w}
    w(v) \coloneqq \vert\{pq \in E_W: v \textrm{ belongs to the }p \textrm{-}q \textrm{ path in } T[S^*]\}\vert + \mathbbm{1}[v \in F]
\end{equation}
where $\mathbbm{1}[v \in F]$ denotes the indicator of the event ``$v \in F$'', and \(T[S^*]\) is the subtree of \(T\) induced by the Steiner nodes.  
See Figure~\ref{fig:ClawUpperBound}.
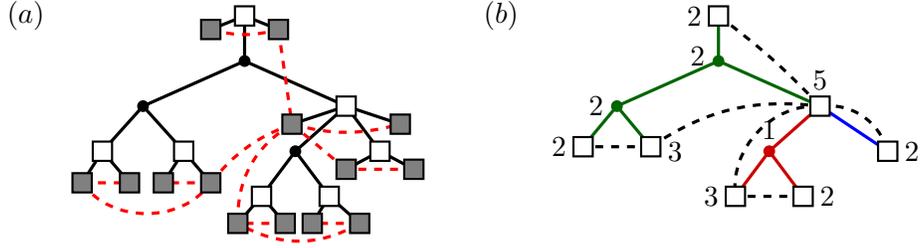
\begin{figure}[t]
    \centering
        \begin{tikzpicture}[scale=0.9]

    \coordinate (b) at (-1.8,2  /1.5);
    \coordinate (a) at (-8.8,2  /1.5);
    \coordinate (s1) at (1,2  /1.5);
    \coordinate (s1') at (0.9,2 /1.5);
    \coordinate (s11) at (1,1  /1.5);
    \coordinate (s11') at (0.95,1.1 /1.5);
    \coordinate (s111) at (-0.5,0 /1.5);
    \coordinate (s111') at (-0.55,0 /1.5);
    \coordinate (s112) at (2.5,0 /1.5);
    \coordinate (s112') at (2.5,0.15 /1.5);
    
    \coordinate (s1111) at (-1,-0.9  /1.5);
    \coordinate (s1111') at (-1.1,-0.9 /1.5);
    \coordinate (s1112) at (0,-0.9  /1.5);
    \coordinate (s1112') at (0.1,-1  /1.5);
    \coordinate (s1121) at (1.75,-1  /1.5);
    
    \coordinate (s1121') at (1.75,-0.95  /1.5);
    \coordinate (s1122) at (3.5,-1  /1.5);
    \coordinate (s1122') at (3.6,-1  /1.5);
    
    \coordinate (s11211) at (1.25,-2  /1.5); 
    \coordinate (s11211') at (1.15,-2  /1.5); 
    
    \coordinate (s11212) at (2.25,-2  /1.5);
    \coordinate (s11212') at (2.35,-2  /1.5); 
    
    \draw
        (s1') node[anchor= east]{$2$} 
        (s11') node[anchor= east]{$2$}
        (s111') node[anchor= east]{$2$}
        
        (s1121') node[anchor= south]{$1$}
        (s112') node[anchor= south]{$5$}
        (s1122') node[anchor= west]{$2$}
        (s1111') node[anchor= east]{$2$} 
        (s1112') node[anchor= west]{$3$} 
        (s11211') node[anchor= east]{$3$}
        (s11212') node[anchor= west]{$2$}
        (a) node[anchor= east]{$(a)$}
        (b) node[anchor= east]{$(b)$};

    \filldraw[color=green!40!black, very thick]
        (s11)circle(1.8pt)
        (s111)circle(1.8pt);
    \filldraw[color=red!80!black, very thick]
        (s1121)circle(1.8pt);
        
    \draw[color=green!40!black,very thick]
        (s1) -- (s11) node[midway,left] {} 
        (s11) -- (s111) node[midway,left] {} 
        (s11) -- (s112) node[midway,right] {} 
        (s111) -- (s1111) node[midway,left] {} 
        (s111) -- (s1112) node[midway,right] {};
    \draw[color=blue,very thick]
        (s112) -- (s1122) node[midway,right] {};

    \draw[very thick,dashed]
        (s1) edge[bend left=05] (s112)
        (s112) edge[bend left=40] (s1122)
        (s112) edge[bend right=50] (s11211)
        (s11212) edge (s11211)
        (s1112) edge[bend left=20] (s112)
        (s1112) edge[] (s1111);
        
    \draw[color=red!80!black,very thick]    
        (s112) -- (s1121) node[midway,left] {}
        (s1121) -- (s11211) node[midway,left] {} 
        (s1121) -- (s11212) node[midway,right] {};
    
    \draw[color=black, thick]
        (s1) node[draw=black, fill=black!0]{}
        (s1111)node[draw=black, fill=black!0]{}
        (s1112)node[draw=black, fill=black!0]{}
        (s112)node[draw=black, fill=black!0]{}
        (s1122)node[draw=black, fill=black!0]{}
       (s11211)node[draw=black, fill=black!0]{}
        (s11212)node[draw=black, fill=black!0]{};

    \filldraw[very thick,color=white]
        (a)circle(1.8pt)
        (b)circle(1.8pt);


    \coordinate (h1) at (-6,2  /1.5);
    \coordinate (t11) at (-6.5,1.7 /1.5);
    \coordinate (t12) at (-5.5,1.7 /1.5);
    \coordinate (h11) at (-6,1  /1.5);
    \coordinate (h111) at (-7.5,0 /1.5);
    
    \coordinate (h112) at (-4.5,0 /1.5);
    \coordinate (t1121) at (-5.3,-0.4 /1.5);
    \coordinate (t1122) at (-3.7,-0.4 /1.5);

    \coordinate (h1111) at (-8.1,-1  /1.5);
    \coordinate (t11111) at (-8.4,-1.7 /1.5);
    \coordinate (t11112) at (-7.7,-1.7 /1.5);

    \coordinate (h1112) at (-6.9,-1  /1.5);

    \coordinate (t11121) at (-6.5,-1.7 /1.5);
    \coordinate (t11122) at (-7.2,-1.7 /1.5);

    \coordinate (h1121) at (-5.25,-1  /1.5);
    
    \coordinate (h1122) at (-4,-1  /1.5);
    \coordinate (t11221) at (-3.5,-1.4 /1.5);
    \coordinate (t11222) at (-4.5,-1.4 /1.5);

    \coordinate (h11211) at (-5.75,-2  /1.5); 
    \coordinate (t112111) at (-5.4,-2.6  /1.5); 
    \coordinate (t112112) at (-6.1,-2.6  /1.5);

    \coordinate (h11212) at (-4.75,-2  /1.5);
    \coordinate (t112121) at (-4.3,-2.6  /1.5); 
    \coordinate (t112122) at (-5,-2.6  /1.5);

    \filldraw[color=black!40!black, very thick]
        (h11)circle(1.8pt)
        (h111)circle(1.8pt);
    \filldraw[color=black!80!black, very thick]
        (h1121)circle(1.8pt);
        
    \draw[color=black!40!black,very thick]
        (h1) -- (h11) node[midway,left] {} 
        (h1) -- (t11) node[midway,left] {}
        (h1) -- (t12) node[midway,left] {}
        (h1111) -- (t11111) node[midway,left] {}
        (h1111) -- (t11112) node[midway,left] {}
        (h1112) -- (t11121) node[midway,left] {}
        (h1112) -- (t11122) node[midway,left] {}

        (h112) -- (t1121) node[midway,left] {}
        (h112) -- (t1122) node[midway,left] {}

        (h1122) -- (t11221) node[midway,left] {}
        (h1122) -- (t11222) node[midway,left] {}

        (h11211) -- (t112111) node[midway,left] {}
        (h11211) -- (t112112) node[midway,left] {}

        (h11212) -- (t112121) node[midway,left] {}
        (h11212) -- (t112122) node[midway,left] {}

        (h11) -- (h111) node[midway,left] {} 
        (h11) -- (h112) node[midway,right] {} 
        (h111) -- (h1111) node[midway,left] {} 
        (h111) -- (h1112) node[midway,right] {};
    \draw[color=black,very thick]
        (h112) -- (h1122) node[midway,right] {};
        
    \draw[color=black!80!black,very thick]    
        (h112) -- (h1121) node[midway,left] {}
        (h1121) -- (h11211) node[midway,left] {} 
        (h1121) -- (h11212) node[midway,right] {};
     
     \draw[color=red, very thick, dashed]
        (t11) edge[bend right=15] (t12)
        (t1121) edge[bend right=30] (t112112)
        (t112121) edge[bend left=30] (t112112)
        (t1121) edge[bend right=15] (t1122)
        (t11111) edge[bend right=55] (t11121)
        (t11121) edge[bend left=15] (t1121)
        ;
     \draw[color=red, very thick, dashed ]
        (t12) -- (t1121)
        (t11111) -- (t11112)
        (t11121) -- (t11122)
        (t11221) -- (t11222)
        (t11222) -- (t1121)
        (t112111) -- (t112112)
        (t112121) -- (t112122);

    \draw[color=black, thick]
        (h1)      node[draw=black, fill=black!0]{}
        (t11)     node[draw=black, fill=black!50]{}
        (t12)     node[draw=black, fill=black!50]{}
        (t1121)   node[draw=black, fill=black!50]{}
        (t1122)   node[draw=black, fill=black!50]{}
        (t11111)  node[draw=black, fill=black!50]{}
        (t11112)  node[draw=black, fill=black!50]{}
        (t11121)  node[draw=black, fill=black!50]{}
        (t11122)  node[draw=black, fill=black!50]{}
        (t11221)  node[draw=black, fill=black!50]{}
        (t11222)  node[draw=black, fill=black!50]{}
        
        (h1111)   node[draw=black, fill=black!0]{}
        (h1112)   node[draw=black, fill=black!0]{}
        (h112)    node[draw=black, fill=black!0]{}
        (h1122)   node[draw=black, fill=black!0]{}
        (h11211)  node[draw=black, fill=black!0]{}
        
        (t112111) node[draw=black, fill=black!50]{}
        (t112112) node[draw=black, fill=black!50]{}
        
        (h11212)  node[draw=black, fill=black!0]{}
        
        (t112121) node[draw=black, fill=black!50]{}
        (t112122) node[draw=black, fill=black!50]{}
        ;

\end{tikzpicture}
        \caption{Figure (a): A tree \(T\) is shown by black edges. The terminals are shown by grey squares. The final Steiner nodes are shown by white squares, non-final Steiner nodes are shown by black dots.
        Figure (b): The tree \(T\) after the terminals have been removed. The color edges indicate the three components. A witness tree $W$ is shown by the black dashed lines. The numbers indicate the values of $w$ imposed on \(T\) computed according to~(\ref{eq:w}).  
        Red dashed lines in Figure (a) show how W can be mapped back.
        }
    \label{fig:ClawUpperBound}
\end{figure}

So, from now on, we consider $T=T[S^*]$. The next step is to root \(T\) at an arbitrary final node \(r\in F\). Following \cite{angelidakis2022node} we can decompose \(T\) into a collection of rooted components \(T_1, \dots T_{\tau}\), where a component is a subtree whose leaves are final nodes and non-leaves are non-final nodes. The decomposition will have the following properties:  each \(T_i\) is rooted at a final node \(r_i\) that has degree one in \(T_i\), 
\(r_1 \coloneqq r\) is the root of \(T_1\), \(\cup_{j<i}T_j\) is connected, and  \(T = \cup_{i=1}^{\tau} T_i\).
We will compute a witness tree \(W_i\) for each component \(T_i\), and then show that we can join these witness trees \(\{W_i\}_{i\ge 1}\) together to get a witness tree $W$ for $T$. 


\subsection{ Computing a witness tree \texorpdfstring{$W_i$}{} for a component \texorpdfstring{$T_i$}{}.}
Here we deal with a component \(T_i\) rooted at \(r_i\), and describe how to construct a witness tree $W_i$. 
If \(T_i\) is a single edge \(e = r_iv\), we simply let \(W_i = (\{r_i,v\}, \{r_iv\})\).

Now we assume that \(T_i\) is not a single edge. We will construct a witness tree with a bottom-up procedure. At a high level, each node $u \in T_i \backslash r_i$ looks at the subtree $Q_u$ of $T_i$ rooted at $u$, and constructs a portion of the witness tree: namely, a subtree \(\overline{W}^u\) spanning the leaves of $Q_u$ (note that, in case the degree of $u$ is 1 in $Q_u$, we  do not consider $u$ to be a leaf of $Q_u$ but just its root). Assume $u$ has children $u_1, \dots, u_k$. Because of the bottom-up procedure, each child $u_j$ has already constructed a subtree \(\overline{W}^{u_j}\). That is, $u$ has to decide how to join these subtrees to get \(\overline{W}^{u}\).

To describe how this is done formally, we first need to introduce some more notation. For every node \(u\in T_i\backslash F\), we select one of its children as the ``marked child'' of \(u\) (according to some rule that we will define later). 
In this way, for every \(u\in T_i\) there is a unique path along these marked children to a leaf. We denote this path by \(P(u)\), and we let \(\ell(u)\) denote the leaf descendent of this path. For final nodes \(u\in F\), we define \(\ell(u) \coloneqq u\) and \(P(u) \coloneqq u\).
For a subtree \(Q_u\) of \(T_i\) rooted at \(u\) and a witness tree \(\overline{W}^u\) over the leaves of \(Q_u\), let \(\overline{w}^u\) be the vector imposed on the nodes of $Q_u$ by \(\overline{W}^u\) according to (\ref{eq:w}). 
Next, we define the following quantity (which, roughly speaking, represents the cost-increase incurred 
after increasing \(\overline{w}^u(v)\) for each \(v\in P(u)\backslash \ell(u)\) for the \((j+1)^{th}\) time):

\[
    {C}^u_j \coloneqq \sum_{v\in P(u)\backslash \ell(u)} \big(H_{\overline{w}^u(v)+j+1} - H_{\overline{w}^u(v)+j}\big) = \sum_{v\in P(u)\backslash \ell(u)} \frac{1}{\overline{w}^u(v)+j+1}
\]

\begin{algorithm}[ht]
        {
            \(u\) has Steiner node children \(u_1, u_2, \ldots, u_k\), and \(\overline{W}^{u_j}\) have been defined\\
        }
        \eIf{ \(u_1, \ldots, u_k\) are all non-final,}
        {
            The \emph{marked} child is \(u_m\), minimizing \(C^{u_m}_{1}\)\\
        }
        {
            Assume \(\{u_1, \ldots,u_{k_1}\}\), \(1 \le k_1\le k\), are final node children of \(u\) \\
            
            \If{\(k_1=k\), or, for all \(j\in \{ k_1+1,\ldots, k\}\), \(C^{u_j}_1 \ge \phi - \delta - H_2\)}
            {
                The \emph{marked} child of \(u\) is \(u_m\) for \(1\le m \le k_1\) such that \(C^{u_m}_1 \) is minimized.
            }
            \If{There is a \(j\in \{ k_1+1,\dots, k\}\) such that \(C^{u_j}_1 < \phi - \delta - H_2\)}
            {
                The \emph{marked} child of \(u\) is \(u_m\) for \(k_1< m \le k\) such that \(C^{u_m}_1 \) is minimized.
            }
        }
        \(\overline{W}^u \leftarrow \left(\bigcup_{j=1}^k V[Q_{u_j}] ,  \bigcup_{j=1}^k \overline{W} ^{u_j} \bigcup_{j \neq m} \{\ell(u_m)\ell(u_j)\}\right)\)\\
    Return \(\overline{W}^u\)
    \caption{Computing the tree $\overline{W}^u$}
    \label{alg:computing_w}
\end{algorithm} 

We can now describe the construction of the witness tree more formally. We begin by considering the leaves of \(T_i\); for a final node (leaf) \(u\), we define a witness tree on the (single) leaf of \(Q_u\) as \(\overline{W}^u = (\{u\},\emptyset)\). 
For a non-final node \(u\), with children \(u_1,\dots,u_k\) and corresponding witness trees \(\overline{W}^{u_1},\dots, \overline{W}^{u_k}\), we select a marked child \(u_m\) for \(u\)   as outlined in Algorithm~\ref{alg:computing_w}, setting \(\phi = 1.86 - \frac{1}{2100}\) and \(\delta = \frac{97}{420}\). 
With this choice, we compute \( \overline{W}^u\) by joining the subtrees  \(\overline{W}^{u_1}, \dots, \overline{W}^{u_k} \) via the edges
\(\ell(u_m)\ell(u_j)\) for \(j\neq m\).
 Finally, let $v$ be the unique child of $r_i$. We let $W_i$ be equal to the tree $\overline{W}^{v}$ plus the extra edge $\ell(v)r_i$, to account for the fact that $r_i$ is also a final node.

\subsection{ Bounding the cost of \texorpdfstring{$W_i$}{}} 
It will be convenient
to introduce the following definitions. For a component $T_i$ and a node \(u \in T_i \setminus r_i\), we let \(W^u\) be the tree \(\overline{W}^u\) plus one extra edge $e^u$, defined as follows.  Let \(a(u)\) be the first ancestor node of \(u\) with \(\ell(a(u)) \neq \ell(u)\) (recall \(\ell(r_i) = r_i\)). We then let the edge \(e^u \coloneqq \ell(u)\ell(a(u))\). We denote by \(w^u\) the vector imposed on the nodes of \(Q_u\) by \(W^u \coloneqq \overline{W}^u + e^u\).  Note that, with this definition, \(W_i = W^v\) for $v$ being the unique child of $r_i$. 

We now state two useful lemmas. 
The first one relates the functions $w^u$ and $w^{u_j}$ for a child $u_j$ of $u$. The statements (a)-(c) below can be proved similarly to Lemma 4 of \cite{angelidakis2022node}. We defer its proof to 
{\color{black}the full version of the paper.}
\begin{restatable}{lemma}{lemIncrease}
\label{lem:increase}
    Let \(u \in T_i\setminus r_i\)  have children \(u_1, \dots, u_k\), and \(u_1\) be its marked child. Then:
    \begin{enumerate}[label=(\alph*)]
        \item \(w^u(u) = k\).
        \item For every \(j \in \{2, \ldots, k\}\) and every node \(v \in Q_{u_j}\), \(w^u(v) = w^{u_j}(v)\).
        \item For every \(v \in Q_{u_1} \setminus P(u_1)\), \(w^u(v) = w^{u_1}(v)\).
        \item  \(\sum_{v \in P(u_1)\setminus\ell(u_1)} H_{w^u(v)} = \sum_{v \in P(u_1)\setminus\ell(u_1)} H_{w^{u_1} (v)} + \sum_{j=1}^{k-1}C^{u_1}_{j}  \). 
    \end{enumerate}
\end{restatable}

Next lemma relates the ``increase'' of cost $C^u_j$ to the \emph{degree} of some nodes in $T_i$.
\begin{restatable}{lemma}{lemBounds}\label{lem:bounds}
Let \(u \in T_i \setminus r_i\)  have children \(u_1, \dots, u_k\), and \(u_1\) be its marked child. Then, \(C^u_1 = C^{u_1}_k + \frac{1}{k+1}\). Furthermore, if \(u_1\)  is non-final and has degree \(d\) in $T_i$, then: \\
    1) \(\sum_{j=1}^{k} (C^{u_1}_{j} - C^{u_j}_{1}) \leq \sum_{j=1}^{k-1} \left(\frac{1}{d+j} - \frac{1}{d}\right)\);
    2) \(H_{w^u(\ell(u_1))} - H_{w^{u_1}(\ell(u_1))} \leq \sum_{j=1}^{k-1}\frac{1}{d+j}\)
\end{restatable}
\begin{proof} 
    \begin{enumerate}

        \item First observe that since \(C_{1}^{u_1} = \min_{j\in [k]} C^{u_j}_1\), we have \(C_{j}^{u_1} - C_{1}^{u_j} \leq C_{j}^{u_1} - C_{1}^{u_1}\).
        Consider \(j\ge 1\), \(C^{u_1}_j - C_{1}^{u_1}\) is equal to
        \begin{align*}
            =& \sum_{v\in P(u_1)\backslash \ell(u)} \left( H_{{w}^{u_1}(v)+j} - H_{{w}^{u_1}(v)+j-1} - H_{{w}^{u_1}(v)+1} + H_{{w}^{u_1}(v)}\right) \\
            =& \sum_{v\in P(u_1)\backslash \ell(u)} \left( \frac{1}{w^{u_1}(v)+j} - \frac{1}{w^{u_1}(v)+1} \right)\leq \frac{1}{w^{u_1}(u_1)+j} - \frac{1}{w^{u_1}(u_1)+1}
        \end{align*}
        Where the inequality follows since every term in the sum is negative. We know that \(w^{u_1}(u_1)= d-1\) by Lemma~\ref{lem:increase}.(a), therefore,  \(C^{u_1}_j - C_{1}^{u_1} \le \frac{1}{d+j-1}- \frac{1}{d}\), and the claim is proven by summing over \(j=1,\dots, k\).
        
        \item To prove the second inequality, 
        first observe that \(w^u(\ell(u_1)) = w^{u_1}(\ell(u_1)) + k-1\). This follows by recalling that \(W^u\) is equal to \(\overline{W}^{u_1}, \dots, \overline{W}^{u_k}\)  plus the edges \(\ell(u_1)\ell(u_j)\) for \(j\neq 1\), and \(e^u\). Thus,
        \(
            H_{w^u(\ell(u_1))} - H_{w^{u_1}(\ell(u_1))} = H_{w^{u_1}(\ell(u_1)) + k-1} - H_{w^{u_1}(\ell(u_1))} = \sum_{i=1}^{k-1} \frac{1}{w^{u_1}(\ell(u_1)) + i} 
        \). 
        Recall \(u_1\) is not a final node, so \(w^{u_1}(\ell(u_1)) > d\). Therefore,
        \[
            \sum_{i=1}^{k-1} \frac{1}{w^{u_1}(\ell(u_1)) + i}\leq \sum_{i=1}^{k-1}\frac{1}{d+i}.
        \]
    \end{enumerate}
\end{proof}

\subsection{ Key Lemma}
To simplify our analysis, we define 
\(h_{W^u}(Q_u) \coloneqq \sum_{\ell \in Q_u} H_{w^u(\ell)}\), and we let \(|Q_u|\) be the number of nodes in \(Q_u\). The next lemma is the key ingredient to prove Theorem~\ref{thm:main_CA}. 
\begin{lemma}\label{lem:invariant}
    Let $\delta = \frac{97}{420}$ and \(\phi = 1.86 - \frac{1}{2100}\). Let $u \in T_i \setminus r_i$ and $k$ be the number of its children. Let \(\beta(k)\) be equal to \(0\) for \(k=0,\dots,8\) and \(\frac{1}{3}-\delta\) for \(k \ge 9\). Then 
    \begin{equation*}
        h_{W^u}(Q_u) + C^u_1 + \delta + \beta(k) \leq \phi \cdot \vert Q_u\vert
    \end{equation*}
\end{lemma}
\begin{proof}
The proof of Lemma~\ref{lem:invariant} will be by induction on \(\vert Q_u\vert\). 
The base case is when \(\vert Q_u \vert= 1\), and hence \(u\) is a leaf of \(T_i\). Therefore, \(W^u\) is just the edge \(e^u\), and so by definition of \(w^u\) we have \(w^u(u)=2\). We get \(h_{W^u}(Q_u) =1.5\), $C^u_1=0$, $\beta(k) =0$ 
and the claim is clear. 

For the induction step: suppose that \(u\) has children \(u_1,\dots,u_k\). We will distinguish 2 cases: (i) \(u\) has no children that are final nodes;
(ii) \(u\) has some child that is a final node (which is then again broken into subcases). 
We report here only the proof of case (i), and defer the proof of the other case to 
{\color{black}the full version of the paper }
as the reasoning follows similar arguments.

\paragraph*{ Case (i): No children of \(u\) are final.} According to Algorithm~\ref{alg:computing_w},
 we mark the child \(u_m\) of \(u\) that minimizes \(C_1^{u_j}\). Without loss of generality, let \(u_m =u_1\). Furthermore, let \(\ell \coloneqq \ell(u_1)\). 
We note the following.
\begin{align*}
    h_{W^u}(Q_u) 
    =\sum^k_{j=1}h_{W^{u}}(Q_{u_j}) + H_{w^u(u)} 
\end{align*}

By applying Lemma~\ref{lem:increase}.(a) we have \(H_{w^u(u)} = H_k\). By Lemma~\ref{lem:increase}.(b) we see \(h_{W^u}(Q_{u_j}) = h_{W^{u_j}}(Q_{u_j})\) for \(j\ge 2\). Using Lemma~\ref{lem:increase}.(c) and (d) we get 
\(
    h_{W^u}(Q_{u_1}) = h_{W^{u_1}}(Q_{u_1}) + \sum_{j=1}^{k-1}C^{u_1}_{j} + H_{w^u(\ell)} - H_{w^{u_1}(\ell)} 
\). Therefore:
\begin{align*}
    h_{W^u}(Q_u)
    =\sum^k_{j=1}h_{W^{u_j}}(Q_{u_j}) + \sum_{j=1}^{k-1}C^{u_1}_{j} + H_{k} + H_{w^u(\ell)} - H_{w^{u_1}(\ell)} 
\end{align*}
We apply our inductive hypothesis on \(Q_{u_1},\dots,Q_{u_k}\), and use \(\beta(j) \ge 0\) for all \(j \):
\begin{align*}
    h_{W^u}(Q_u) &\le \sum^k_{j=1}\left(\phi|Q_{u_j}| - \delta - C^{u_j}_{1} \right) + \sum_{j=1}^{k-1} C^{u_1}_{j} + H_{k} + H_{w^u(\ell)} - H_{w^{u_1}(\ell)} \\
    =& \phi(|Q_u|-1) - k\delta - C^{u_1}_{k} + \sum_{j=1}^{k}\left( C^{u_1}_{j} - C^{u_j}_{1}\right) + H_{k} + H_{w^u(\ell)} - H_{w^{u_1}(\ell)}
\end{align*}

    
Using Lemma~\ref{lem:bounds}, we get 
\begin{align*}
    \leq & \phi(|Q_u|-1) - k\delta - C^{u}_{1} + \sum_{j=1}^{k-1} \left(\frac{1}{d+j} - \frac{1}{d}\right) + H_{k+1} + \sum_{j=1}^{k-1}\frac{1}{d+j}\\
    \leq & \phi |Q_u|  - \delta  - C^{u}_1 -\beta(k)
\end{align*}
where the last inequality follows since one checks that for any $k\geq 1$ and $d \geq 2$ we have
    \(
      -\phi -(k-1)\delta  + \sum_{j=1}^{k-1} \left(\frac{1}{d+j} - \frac{1}{d}\right) + H_{k+1} + \sum_{j=1}^{k-1}\frac{1}{d+j} \le - \beta(k).
    \) 
    We show this inequality  
    {\color{black} the full version of the paper.}
\end{proof}

\subsection{ Merging and bounding the cost of \texorpdfstring{$W$}{}}
Once the \(\{W_i\}_{i\ge 1}\) are computed for each component \(T_i\), we let the final witness tree be simply the union $W = \cup_{i} W_i$. Our goal now is to prove the following.
\begin{lemma} \label{lem:tree_up_bound}
$\nu_{T}(W) 
\leq \phi= 1.86 - \frac{1}{2100}$.
\end{lemma}
\begin{proof}
Recall that we decomposed \(T\) into components \(\{T_i\}_{i=1}^\tau\), such that \(\cup_{j\leq i} T_j\) is connected for all \(i\in [\tau]\).
For a given $i$, define \(T' = \cup_{j<i}T_j\), \(W' = \cup_{j<i} W_i\), and let \(w'\) be the vector imposed on the nodes of \(T'\) by $W'$
(for \(i=1\), set  \(T' =\emptyset\), \(W'= \emptyset\), and \(w'=0\)). 
Finally, define \(W'' = W_i \cup W'\) and let \(w''\)
be the vector imposed on the nodes of \(T'' := T'\cup T_i\).  By induction on $i$, we will show that 
$\nu_{T''}(W'') \leq \phi $. The statement will then follow by taking $i=\tau$. Recall that, for any $i$, 
\(r_i\) is adjacent to a single node \(v\) in \(T_i\), and $W_i = W^v$.

First consider \(i=1\).  Hence, $W''=W_1=W^v$ and \(w''(r_1) =2\). By applying Lemma~\ref{lem:invariant} to the subtree \(Q_v\) we get
    \begin{equation*}
        \sum_{u\in T''} H_{w''(u)}= h_{W^{v}}(Q_{v}) + H_{w''(r_i)} \leq \phi(|Q_v|) + H_2 \leq \phi(|Q_v| + 1) \Rightarrow \nu_{T''}(W'') \leq \phi
    \end{equation*}
    
Now consider \(i>1\).  In this case, \(w''(r_i) = w'(r_i) +1 \geq 3\). Therefore:
    \begin{align*}
        &\sum_{u\in T''} H_{w''(u)} = \sum_{u\in T_i \setminus r_i}  H_{w^v(u)} + \sum_{u\in  T'} H_{w'(u)} - H_{w'(r_i)} + H_{w'(r_i)+1}\\
        =& \sum_{u\in T_i \setminus r_i}  H_{w^v(u)} + \sum_{u\in  T'} H_{w'(u)} + \frac{1}{w'(r_i)+1}
        \leq \sum_{u\in T_i \setminus r_i}  H_{w^v(u)} + \sum_{u\in  T'} H_{w'(u)} + \frac{1}{3}
    \end{align*}
    If \(v\) is a final node, then \(\sum_{u\in T_i \setminus r_i}  H_{w^v(u)} = H_{w^v(v)} = H_2\) and by induction 
    \begin{align*}
        \sum_{u\in T''} H_{w''(u)} \leq H_3 + \sum_{u\in  T'} H_{w'(u)} \leq \phi|T''| \Rightarrow \nu_{T''}(W'') \leq \phi
    \end{align*}
    If \(v\) is not a final node, then by induction on $T'$ and by applying Lemma~\ref{lem:invariant} to the subtree \(Q_v\), assuming that \(v\) has \(k\) children, we can see 
    \begin{align*}
        \sum_{u\in T''} H_{w''(u)} \leq \phi|T''| - C_1^v - \delta -\beta(k) + \frac{1}{3}  \leq \phi|T''| - \frac{1}{k+1} - \delta -\beta(k) + \frac{1}{3} 
    \end{align*}
    If $1\le k \le 8$, then \(\beta(k) = 0\), but we have \(\frac{1}{3} < 431/1260=\frac{1}{9}+\delta \le \frac{1}{k+1}+\delta\).
    If \(k\ge 9\), \(\beta(k) = \frac{1}{3}-\delta\) and \(\frac{1}{3} - \delta -\beta(k) =0\). In both cases, \(\nu_{T''}(W'') \leq \phi\). 

\end{proof}

Note that we did not make any assumption on $T$, other than being a CA-Node-Steiner-Tree. Hence, Lemma~\ref{lem:tree_up_bound} yields the following corollary.

\begin{corollary}\label{cor:psi}
$\psi \leq 1.86 - \frac{1}{2100} < 1.8596$.
\end{corollary}
Combining Corollary~\ref{cor:psi} with Theorem~\ref{thm:alpha} yields a proof of Theorem~\ref{thm:main_CA}.

\section{ Improved Lower Bound on \texorpdfstring{$\psi$}{}}
\label{sec:nodelower}
The goal of this section is to prove Theorem~\ref{thm:nodelowerbound}. For the sake of brevity, we will omit several details.
(see 
{\color{black}the full version of the paper}
for a completed proof).

\paragraph*{ Sketch of Proof of Theorem~\ref{thm:nodelowerbound}}
Consider a CA-Node-Steiner-Tree instance \((G, R)\), where \(G\) consists of a path of Steiner nodes \(s_1, \dots, s_q\) such that, for all \(i \in [q]\), \(s_i\) is adjacent to Steiner nodes \(t_{i1}, t_{i2}, t_{i3}\), and each \(t_{ij}\) is adjacent to two terminals $r_{ij}^1$ and $r_{ij}^2$. See Figure~\ref{fig:nodelowerbound}. 
We will refer to $B_i$ as the subgraph induced by $s_i,t_{ij},r_{ij}^1,r_{ij}^2$ ($j=1,2,3$).
Since \(G\) is a tree connecting the terminals, clearly the optimal Steiner tree for this instance is \(T=G\). 


\begin{figure}[t]
    \begin{center}
        \begin{tikzpicture}[scale=0.5]
        
            \tikzset{terminal/.style={draw=black,  thick,minimum size=0pt, inner sep=2.5pt, outer sep=1pt}}
            
            \begin{scope}[every node/.style={terminal}]
                \node (t111) at (0, 0){};
                \node (t112) at (1, 0){};
                
                \node (t121) at (2, 0){};
                \node (t122) at (3, 0){};
                
                \node (t131) at (4, 0){};
                \node (t132) at (5, 0){};

                \node (t211) at (6, 0){};
                \node (t212) at (7, 0){};
                
                \node (t221) at (8, 0){};
                \node (t222) at (9, 0){};
                
                \node (t231) at (10, 0){};
                \node (t232) at (11, 0){};

                \node (t311) at (12, 0){};
                \node (t312) at (13, 0){};

                \node (t321) at (14, 0){};
                \node (t322) at (15, 0){};
        
                \node (t331) at (16, 0){};
                \node (t332) at (17, 0){};

                \node (t411) at (18, 0){};
                \node (t412) at (19, 0){};
                
                \node (t421) at (20, 0){};
                \node (t422) at (21, 0){};
                
                \node (t431) at (22, 0){};
                \node (t432) at (23, 0){};
            \end{scope}
            
            \tikzset{vertex/.style={draw=black, very thick, circle,minimum size=0pt, inner sep=1pt, outer sep=1pt,fill=black}}
            \begin{scope}[every node/.style={vertex}]
                \node (r11) at (0.5, 1){};
                \node (r12) at (2.5, 1){};
                \node (r13) at (4.5, 1){};
                
                \node (r21) at (6.5, 1){};
                \node (r22) at (8.5, 1){};
                \node (r23) at (10.5,1){};
                
                \node (r31) at (12.5,1){};
                \node (r32) at (14.5,1){};
                \node (r33) at (16.5,1){};
                
                \node (r41) at (18.5,1){};
                \node (r42) at (20.5,1){};
                \node (r43) at (22.5,1){};
            \end{scope}

            
            \begin{scope}[every node/.style={vertex}]

                \node (s1) at (2.5, 2.2){};
                \node (s2) at (8.5, 2.2){};
                \node (s3) at (14.5,2.2){};
                \node (s4) at (20.5,2.2){};
                
                \draw[very thick] (s1) to (s2) to (s3) to (s4);

                \draw[very thick] (s1) to (r13);
                \draw[very thick] (r12) to (s1) to (r11);
                \draw[very thick] (t111) to (r11) to (t112);
                \draw[very thick] (t121) to (r12) to (t122);
                \draw[very thick] (t131) to (r13) to (t132);
                
                \draw[very thick] (r22) to (s2) to (r21);
                \draw[very thick] (s2) to (r23);
                \draw[very thick] (t211) to (r21) to (t212);
                \draw[very thick] (t221) to (r22) to (t222);
                \draw[very thick] (t231) to (r23) to (t232);

                \draw[very thick] (r32) to (s3) to (r31);
                \draw[very thick] (s3) to (r33);
                \draw[very thick] (t311) to (r31) to (t312);
                \draw[very thick] (t321) to (r32) to (t322);
                \draw[very thick] (t331) to (r33) to (t332);                
                
                \draw[very thick] (r42) to (s4) to (r41);
                \draw[very thick] (s4) to (r43);
                \draw[very thick] (t411) to (r41) to (t412);
                \draw[very thick] (t421) to (r42) to (t422);
                \draw[very thick] (t431) to (r43) to (t432);

                \draw [draw={rgb,255: red,195; green,0; blue,3}, thick] (t111) to (t112);
                \draw [draw={rgb,255: red,195; green,0; blue,3}, thick] (t121) to (t122);
                \draw [draw={rgb,255: red,195; green,0; blue,3}, thick] (t131) to (t132);

                \draw [draw={rgb,255: red,195; green,0; blue,3}, thick] (t211) to (t212);
                \draw [draw={rgb,255: red,195; green,0; blue,3}, thick] (t221) to (t222);
                \draw [draw={rgb,255: red,195; green,0; blue,3}, thick] (t231) to (t232);

                \draw [draw={rgb,255: red,195; green,0; blue,3}, thick] (t311) to (t312);
                \draw [draw={rgb,255: red,195; green,0; blue,3}, thick] (t321) to (t322);
                \draw [draw={rgb,255: red,195; green,0; blue,3}, thick] (t331) to (t332);

                \draw [draw={rgb,255: red,195; green,0; blue,3}, thick] (t411) to (t412);
                \draw [draw={rgb,255: red,195; green,0; blue,3}, thick] (t421) to (t422);
                \draw [draw={rgb,255: red,195; green,0; blue,3}, thick] (t431) to (t432);
                
                \draw [draw={rgb,255: red,195; green,0; blue,3}, thick, bend right=25] (t111) to (t121);
                \draw [draw={rgb,255: red,195; green,0; blue,3}, thick, bend right=25] (t121) to (t131);
                \draw [draw={rgb,255: red,195; green,0; blue,3}, thick, bend right=25] (t121) to (t221);
                
                \draw [draw={rgb,255: red,195; green,0; blue,3}, thick, bend right=25] (t211) to (t221);
                \draw [draw={rgb,255: red,195; green,0; blue,3}, thick, bend right=25] (t221) to (t231);
                \draw [draw={rgb,255: red,195; green,0; blue,3}, thick, bend right=25] (t221) to (t321);
                
                \draw [draw={rgb,255: red,195; green,0; blue,3}, thick, bend right=25] (t311) to (t321);
                \draw [draw={rgb,255: red,195; green,0; blue,3}, thick, bend right=25] (t321) to (t331);
                \draw [draw={rgb,255: red,195; green,0; blue,3}, thick, bend right=25] (t321) to (t421);
                
                \draw [draw={rgb,255: red,195; green,0; blue,3}, thick, bend right=25] (t412) to (t421);
                \draw [draw={rgb,255: red,195; green,0; blue,3}, thick, bend right=25] (t421) to (t431);

            \end{scope}
            \node (s0) at (0  , 2.2){};
            \node (s5) at (23  ,2.2){};
            \node (t0) at (0,-0.7) {};
            \node (t5) at (23,-0.5) {};

            \draw[very thick, dotted]  (s0) -- (s1);
            \draw[very thick, dotted]  (s4) -- (s5);
            
            \draw [draw={rgb,255: red,195; green,0; blue,3}, thick, bend left=20, dotted] (t121) to (t0);
            \draw [draw={rgb,255: red,195; green,0; blue,3}, thick, bend right=20, dotted] (t421) to (t5);

            \node[above=1pt] () at (s1) {};
            \node[above=1pt] () at (s2) {};
            \node[above=1pt] () at (s3) {};
            \node[above=1pt] () at (s4) {};
        \end{tikzpicture}
    \end{center}
    \caption{
    Lower bound instance shown in black. The white squares are terminals and black circles are Steiner nodes. Red edges form the laminar witness tree $W^*$.
    }
    \label{fig:nodelowerbound}
\end{figure}
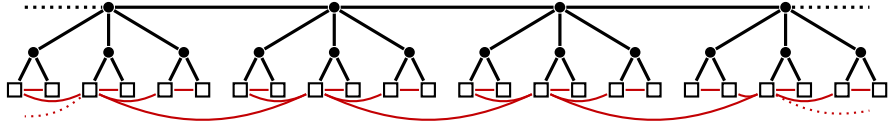

Let $W^*$ be a witness tree that minimizes $\nu_T(\W)$. Recall that we can assume $W^*$ to be laminar by Theorem~\ref{thm:nodelaminar}. We arrive at an explicit characterization of $W^*$ in three steps. First, we observe that, without loss of generality, we can assume that  every pair of terminals \(r_{ij}^1\) and \(r_{ij}^2\) are adjacent in \(\W\) and that \(r_{ij}^2\) is a leaf of $W^*$.
Second, using the latter of these observations and laminarity, we show that for all $i$, the subgraph of \(W\) induced by  \(r_{i1}^1, r_{i2}^1, r_{i3}^1\) can only be either (a) a star, 
or (b) three singletons, adjacent to a unique  terminal \(f \notin B_i \). We say that $B_i$ is a \emph{center} in $W^*$ if (a) holds.  
Finally, we get rid of case (b), and essentially arrive at the next lemma, whose proof can be found in 
{\color{black}the full version of the paper.}
\begin{lemma}\label{lem:CentersOflengthOne}
    Let \(\mathcal{W}\) be the family of all laminar witness trees over \(T\), and let \(\W\) be a laminar witness tree such that for every \(i\in[q]\), \(B_i\) is a center in \(\W\). 
    Then \(\nu_T(\W) = \min_{W\in \mathcal{W}} \nu_T(W)\).
\end{lemma}
    Once we impose the condition that all $B_i$ are centers, one notes that the tree $W^*$ essentially must look like the one shown in Figure~\ref{fig:nodelowerbound}.
     So it only remains to compute \(\nu_T(\W)\). For every \(B_i\), we can compute \(\sum_{v\in B_i} H_{w^*(v)}\), where \(w^*\) is the vector imposed on the set \(S\) of Steiner nodes by \(\W\). For \(i\in \{2,\dots, q-1\}\), 
    one notes that \(\frac{1}{4}\sum_{v\in B_i} H_{w^*(v)} = \frac{1}{4}(2H_2 + H_4 + H_5) = 221/120 = 1.841\bar{6} \). Similarly, 
    for $i=1$ and $q$ we have \(\frac{1}{4}\sum_{v\in B_1} H_{w^*(v)} = \frac{1}{4}\sum_{v\in B_q} H_{w^*(v)} = \frac{1}{4}(2H_2 + H_3 + H_4) = \frac{83}{48} =  1.7291\bar{6}\). 
    Therefore, we can see that \(\nu_T(\W) = \sum_{v\in S} \frac{H_{w^*(v)}}{|S|} = \frac{1.841\bar{6}q - 2(1.841\bar{6}-1.7291\bar{6})}{q}\). Thus, for \(q>\frac{1}{\varepsilon}\) we have \(\nu_T(\W) > 1.841\bar{6}-\frac{1}{q}\).

\section{ Tight bound for Steiner-Claw Free Instances}
\label{sec:clawupper}
We here prove Theorem~\ref{thm:clawupper}. 
Our goal is to show that for any Steiner-Claw Free instance \((G, R, c)\),  $\gamma_{(G,R,c)} \leq \frac{991}{732}$, improving over the known $\ln(4)$ bound that holds in general. 
From now on, we assume that we are given an optimal solution \(T = (R\cup S^*, E^*)\) to \((G, R, c)\). 

\paragraph*{ Simplifying Assumptions.} As standard, note that \(T\) can be decomposed into components \(T_1,\dots, T_\tau\), where each component is a maximal subtree of \(T\) whose leaves are terminals and internal nodes are Steiner nodes. Since components do not share edges of $T$, it is not difficult to see that one can compute a witness tree \(W_i\) for each component \(T_i\) separately, and then take the union of the \(\{W_i\}_{i\ge 1}\) to get a witness tree $W$ whose objective function $\bar \nu_T(W)$ will be bounded by the maximum among $\bar \nu_{T_i}(W_i)$. Hence, from now on we assume that $T$ is made by one single component.
Since $T$ is a solution to a Steiner-claw free instance, each Steiner node is adjacent to at most \(2\) Steiner nodes. In particular,
the Steiner nodes induce a path in \(T\), which we enumerate as \(s_1, \dots, s_q\).
%
%
We will assume without loss of generality that each \(s_j\) is adjacent to exactly one terminal \(r_j\in R\): this can be achieved by replacing a Steiner node incident to $p$ terminals, with a path of length $p$ made of 0-cost edges, if $p>1$, and with an edge of appropriate cost connecting its 2 Steiner neighbors, if $p=0$. 
We will also assume that \(q>4\). For \(q\le 4\), it is not hard to compute that \(\gamma_{(G,R,c)} \le \frac{991}{732}\). (For sake of completeness we explain this in 
{\color{black}the full version of the paper)}
\begin{figure}[t]
    \centering
        \begin{tikzpicture}[scale=1.15]
    \coordinate (s1) at (0,0.75);
    \coordinate (s2) at (1,0.75);
    \coordinate (s3) at (2,0.75);
    \coordinate (s4) at (3,0.75);
    \coordinate (s5) at (4,0.75);
    \coordinate (s6) at (5,0.75);
    \coordinate (s7) at (6,0.75);
    \coordinate (s8) at (7,0.75);
    \coordinate (s9) at (8,0.75);
    \coordinate (s10) at (9,0.75);
    \coordinate (s11) at (10,0.75);
    
    \coordinate (r1) at (0,0);
    \coordinate (r2) at (1,0);
    \coordinate (r3) at (2,0);
    \coordinate (r4) at (3,0);
    \coordinate (r5) at (4,0);
    \coordinate (r6) at (5,0);
    \coordinate (r7) at (6,0);
    \coordinate (r8) at (7,0);
    \coordinate (r9) at (8,0);
    \coordinate (r10) at (9,0);
    \coordinate (r11) at (10,0);
    
    \filldraw[color=black, very thick]
        (s1)circle(1.5pt)
        (s2)circle(1.5pt)
        (s3)circle(1.5pt)
        (s4)circle(1.5pt)
        (s5)circle(1.5pt)
        (s6)circle(1.5pt)
        (s7)circle(1.5pt)
        (s8)circle(1.5pt)
        (s9)circle(1.5pt)
        (s10)circle(1.5pt)
        (s11)circle(1.5pt);
    
    \draw[color=black,very thick]
        (s1) -- (s2) 
        (s2) -- (s3) 
        (s3) -- (s4) 
        (s4) -- (s5) 
        (s5) -- (s6) 
        (s6) -- (s7) 
        (s7) -- (s8)
        (s8) -- (s9)
        (s9) -- (s10)
        (s10) -- (s11)

        (r1) -- (s1) 
        (r2) -- (s2) 
        (r3) -- (s3) 
        (r4) -- (s4) 
        (r5) -- (s5) 
        (r6) -- (s6) 
        (s7) -- (r7)
        (s8) -- (r8)
        (s9) -- (r9)
        (s10) -- (r10)
        (s11) -- (r11);
        
    \draw [color={rgb,255: red,195; green,0; blue,3},very thick] 
    (r1) edge (r2)
    (r1) edge[bend right=35](r5)
    
    (r3) edge[bend right=35](r5)
    (r4) edge (r5)
    (r5) edge (r6)
    (r5) edge[bend right=35](r7)
    
    (r5) edge[bend right=30] (r10)
    
    (r8) edge[bend right=35] (r10)
    (r9) edge (r10)
    (r11) edge (r10);

    \draw[color=black,very thick]
        (r1) node[draw=black, fill=blue!40]{$r_1$}
        (r2) node[draw=black, fill=black!0]{$r_2$}
        (r3) node[draw=black, fill=black!0]{$r_3$}
        (r4) node[draw=black, fill=black!0]{$r_4$}
        (r5) node[draw=black, fill=blue!40]{$r_5$}
        (r6) node[draw=black, fill=black!0]{$r_6$}
        (r7) node[draw=black, fill=black!0]{$r_7$}
        (r8) node[draw=black, fill=black!0]{$r_8$}
        (r9) node[draw=black, fill=black!0]{$r_9$}
        (r10) node[draw=black, fill=blue!40]{$r_{10}$}
        (r11) node[draw=black, fill=black!0]{$r_{11}$};
\end{tikzpicture}
        \caption{Edges of $T$ are shown in black. Red edges show $W$. 
        Here, $q=11$, $t_\alpha=5$ and $\sigma=5$. Initially $r_5$ and \(r_{10}\) are picked as the centers of stars in \(W\).
        Since $\sigma>\lceil\frac{t_\alpha}{2}\rceil$, \(r_1\) is also the center of a star.
        Since $\sigma + t_\alpha \lfloor\frac{q-\sigma}{t_\alpha}\rfloor > q-\lceil\frac{t_\alpha}{2}\rceil$, $r_{q}$ is not the center of a star.} 
    \label{fig:clawMarking}
\end{figure}
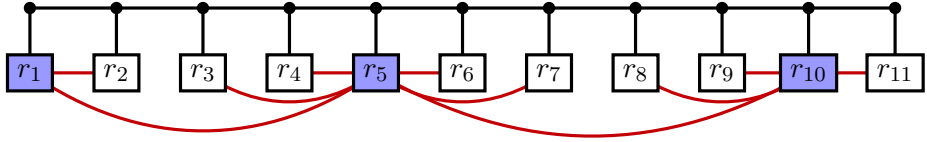
\paragraph*{ Witness tree computation and analysis.}
We denote by \(L\subseteq E^*\) the edges of $T$ incident to a terminal, and by \(O = E^*\setminus L\) the edges of the path \(s_1, \dots, s_q\). 
Let \(\alpha \coloneqq c(O) / c(L)\). For a fixed value of \(\alpha \ge 0\), we will fix a constant $t_\alpha$ as follows: If \(\alpha\in [0,32/90]\), then \(t_\alpha = 5\), if \(\alpha\in (32/90, 1)\), then \(t_\alpha =3\), and if \(\alpha \ge 1\), then \(t_\alpha = 1\).
Given \(\alpha\) (and thus \(t_\alpha\)), we construct \(W\) using the randomized process outlined in Algorithm~\ref{alg:random_w}.  
At a high level, starting from a random offset, Algorithm~\ref{alg:random_w} adds sequential stars of $t_\alpha$ terminals to $W$, connecting the centers of these stars together in this sequence. See Figure~\ref{fig:clawMarking} for an example. 

\begin{algorithm}[ht]
    Initialize \(W = (R, E_W =\emptyset)\)\\
    Sample uniformly at random  \(\sigma\) from \(\{1, \dots, t_\alpha\}\).\\
    \(E_W \leftarrow  \{r_\sigma  r_{\sigma+k} | 1\le | k| \le \left\lfloor \frac{t_\alpha}{2}\right \rfloor, 1\le \sigma +k \le q \}\)\\
    Initialize $j$=1\\
    \While{\(j \leq \frac{q-\sigma}{t_\alpha}\)}
    {
        \(\ell \coloneqq \sigma + t_\alpha j\)\\
        \(E_W \leftarrow E_W \cup \{r_\ell r_{\ell+k} | 1\le | k| \le \left\lfloor \frac{t_\alpha}{2}\right \rfloor, 1\le \ell +k \le q \}\)\\
        \(E_W \leftarrow E_W  \cup\{r_{\sigma + t_\alpha(j-1)}r_{\sigma + t_\alpha j}\} \)\\ 
        \(j \leftarrow j + t_\alpha\)
    }
    \If{\(\sigma > \lceil \frac{t_\alpha}{2} \rceil \)}
    {
        \(E_W \leftarrow E_W \cup \left\{ r_1 r_{k} |   2 \leq k \leq \sigma - \lceil \frac{t_\alpha}{2} \rceil \right\}\cup \{r_1r_\sigma\}\)\\
    }
    \(j \leftarrow \lfloor \frac{q-\sigma}{t_\alpha}\rfloor\)\\
    \If{\( \sigma + t_\alpha j \le q - \lceil \frac{t_\alpha}{2} \rceil\)}
    {
        \(E_W \leftarrow E_W \cup  \{r_{k} r_q |  \sigma+ t_\alpha j + \lceil \frac{t_\alpha}{2} \rceil \leq k \leq q-1 \} \cup \{r_{\sigma +t_\alpha j}r_q\}\)\\
    }
    Return \(W\)
    \caption{Computing the witness tree $W$}
    \label{alg:random_w}
\end{algorithm}

Under this random scheme, we define \(\lambda_{L} (t_\alpha) \coloneqq \max_{e\in L} \mathbb{E}[H_{\bar w(e)}]\), and \(\lambda_{O} (t_\alpha) \coloneqq \max_{e\in O} \mathbb{E}[H_{\bar w(e)}]\). 
\begin{lemma} 
    \label{lem:lambdabound} 
    For any \(\alpha \ge 0\),  
   \( \lambda_{L}(t_{\alpha}) \leq \frac{1}{t_{\alpha}} H_{t_{\alpha}+1} + \frac{t_{\alpha}-1}{t_{\alpha}}  \), and  \( \lambda_{O}(t_{\alpha}) \leq \frac{1}{t_{\alpha}}  + \frac{2}{t_{\alpha}} \sum_{i=2}^{\lceil \frac{t_{\alpha}}{2} \rceil}H_{i}.\) 
\end{lemma}

\begin{proof}
Let \(W = (R,E_W)\) be a witness tree returned from running Algorithm~\ref{alg:random_w} with \(\alpha\) and \(t \coloneqq t_\alpha\), and let \(w\) be the vector imposed on \(E^*\) by \(W\). 
If Algorithm~\ref{alg:random_w} samples \(\sigma \in \{1,\dots, t\}\), then we say that the terminals \(r_{\sigma + t j}\) are \emph{marked} by the algorithm. Moreover, if \(\sigma > \lceil \frac{t_\alpha}{2}\rceil\) (resp. \(\sigma + t_\alpha \lfloor\frac{q-\sigma}{t_\alpha}\rfloor \le q- \lceil \frac{t_\alpha}{2}\rceil\)) then \(r_1\) (resp. \(r_q\)) is also considered marked.
    \noindent\begin{enumerate}
        \item Consider edge \(e=s_js_{j+1} \in O\), with \(j\in \{\lceil\frac{t}{2}\rceil, \dots,  q-\lceil{\frac{t}{2}}\rceil\} \). Let \(m\in \{ j-\lfloor \frac{t}{2} \rfloor, \dots, j + \lfloor \frac{t}{2} \rfloor\}\), such that $\sigma\mod{t}=m\mod{t}$. Observe that in this case $r_m$ is marked.
        If $m=j-x$ for $x\in \{0,\ldots,\lfloor \frac{t}{2} \rfloor\}$, 
        then $w(s_js_{j+1})= \lceil\frac{t}{2}\rceil-x$.
        Similarly if $m=j+x$ for $x\in \{1,\ldots,\lfloor \frac{t}{2} \rfloor\}$, then $w(s_js_{j+1})= \lceil\frac{t}{2}\rceil-x+1$. Since $m\mod{t} = \sigma\mod{t}$ with probability $\frac{1}{t}$, we have  \( \mathbb{E}[H_{w(s_{j}s_{j+1})}] = \frac{1}{t} + \frac{2}{t} \sum_{k=2}^{ \lceil \frac{t}{2} \rceil}H_{k}\). 
        
        Now assume \(j < \lceil\frac{t}{2}\rceil\) (the case  $j > q-\lceil{\frac{t}{2}}\rceil$ can be handled similarly). {\color{black}{Recalling that since $t$ is odd it is not hard to determine the value of \(w(s_js_{j+1})\) by cases, depending on the value of $\sigma$.}}
        \begin{enumerate}
            \item $1\le \sigma \le j$: Then \(w(s_js_{j+1})=\lceil\frac{t}{2}\rceil +\sigma - j\).
            \item  \(j+1\le \sigma \le \lceil \frac{t}{2} \rceil\): Then \(w(s_js_{j+1})=j\).
            \item  \(\lceil \frac{t}{2} \rceil+1 \le \sigma \le j+ \lfloor\frac{t}{2}\rfloor\):  Then \(w(s_js_{j+1})=\lceil\frac{t}{2}\rceil - \sigma + j +1 \).
            \item  $j+\lceil\frac{t}{2}\rceil \le  \sigma\le t$: Then \(w(s_js_{j+1})=\sigma-j -\lceil\frac{t}{2}\rceil +1\).
        \end{enumerate}
        \begin{align*}
            & \mathbb{E}[H_{w(s_js_{j+1})}] =
             \\&= \frac{1}{t}\left(\sum_{\sigma =1}^j H_{\lceil\frac{t}{2}\rceil +\sigma - j }
            + \sum_{\sigma = j+1}^{\lceil\frac{t}{2}\rceil} H_{j}
            +\sum_{\sigma = \lceil\frac{t}{2}\rceil +1}^{j+\lfloor\frac{t}{2}\rfloor} H_{\lceil\frac{t}{2}\rceil - \sigma + j+1 }
            + \sum_{\sigma = j + \lceil\frac{t}{2}\rceil }^{t} H_{\sigma - j - \lceil\frac{t}{2}\rceil + 1} \right )\\
            & = \frac{1}{t}\left( \sum_{i=\lceil\frac{t}{2}\rceil-j+1}^{\lceil\frac{t}{2}\rceil} H_i 
            + \left( \left\lceil \frac{t}{2} \right\rceil-j \right) H_j
            + \sum_{i=2}^{j}H_i
            + \sum_{i=1}^{\lceil\frac{t}{2}\rceil -j }H_i\right)
        \end{align*}
            
        \begin{align*}
             &= \frac{1}{t}\left( \sum_{i=1}^{\lceil\frac{t}{2}\rceil} H_i 
            + \left( \left\lceil \frac{t}{2} \right\rceil-j \right) H_j
            + \sum_{i=2}^{j}H_i \right)
            < \frac{1}{t}\left( 1+2\sum_{i=2}^{\lceil\frac{t}{2}\rceil} H_i \right).
        \end{align*}        
        \item Consider edge \(e=s_jr_j\in L\). We first show the bound for \(j \in \{1,\dots, q\}\). Algorithm~\ref{alg:random_w} marks terminal \(r_i\) with probability \(\frac{1}{t}\). If \(r_i\) {\color{black}is marked,} then \(w(e) \le t\). If \(r_i\) is not marked, then \(w(e)=1\). Therefore, \( \mathbb{E}[H_{w(e)}] \le \frac{1}{t} H_{t+1} + \frac{t-1}{t}\)
        
        Now consider edge \(e=s_1r_1\) (the case \(e=s_qr_q\) can be handled similarly). We consider specific values of \(\sigma \in \{1, \dots, t\}\) sampled by Algorithm~\ref{alg:random_w}. With probability \(\frac{1}{t}\), we have \(\sigma=1\), so \(r_1\) is marked initially and \(w(e) = \lceil t/2 \rceil\). For \(\sigma = 2,\dots,  \lceil t/2 \rceil \), \(r_1\) is unmarked and \(w(e) = 1\). If \(\sigma > \lceil t/2 \rceil\), then \(r_1\) is marked by the algorithm and \(w(e) = \sigma - \lceil t/2 \rceil\). Therefore, we can see 
        \begin{align*}
            \mathbb{E}[H_{w(r_1s_1)}] =  \frac{1}{t}\left( H_{\lceil t/2 \rceil} + \left\lfloor \frac{t}{2} \right\rfloor + \sum_{k=1}^{t-\lceil t/2 \rceil}H_k \right)
        \end{align*}
        We let \(g(t)\) be equal to the equality above. It remains to show that \(g(t) \le \frac{1}{t} H_{t+1} + \frac{t-1}{t} \coloneqq f(t)\) for \(t\in \{1,3,5\}\).  
        \begin{align*}
            g(1) &=  H_1 = 1 < H_{2} =  f(1)\\
            g(3) &= \frac{1}{3}\left( H_{2} + 1 + H_1 \right) = 1.1\Bar{6} < 1.36\Bar{1} = \frac{1}{3}(H_{4} + 2) = f(3) \\
            g(5) &= \frac{1}{5}\left( H_{3} + 2 + H_1 + H_2 \right) = 1.2\Bar{6} < 1.29 = \frac{1}{5}( H_{6} + 4) = f(5)
        \end{align*}
    \end{enumerate}
    Combining these two facts gives us the bound on \(\lambda_{L_i}(t)\), for \(t\in \{1,3,5\}\). 
\end{proof}

The following Lemma is proven in 
{\color{black}the full version of the paper.}
\begin{lemma}
\label{lem:upperbound}
    For any \(\alpha \ge 0\), the following bounds holds:
    \begin{align*}
         \frac{1}{\alpha + 1} \Bigg(\frac{1}{t_\alpha} H_{t_\alpha+1} + 
         \frac{t_\alpha-1}{t_\alpha} + \alpha \Bigg(\frac{1}{t_\alpha}  + 
         \frac{2}{t_\alpha} \sum_{i=2}^{\lceil \frac{t_\alpha}{2} 
         \rceil}H_{i} \Bigg) \Bigg) \le \frac{991}{732}
    \end{align*}
\end{lemma}
We are now ready to prove the following:
\begin{lemma}
    \label{lem:gamma_claw}
 \(\mathbb{E}[\bnu_{T}(W)]  \leq \frac{991}{732}\).
\end{lemma}
\begin{proof}
One observes:
    \begin{align*}
    \sum_{e\in L\cup O}c(e)\mathbb{E}[H_{\bar w(e)}] 
        \leq  \sum_{e\in L}c(e) \lambda_{L}(t_{\alpha}) +  \sum_{e\in O}c(e) \lambda_{O}(t_{\alpha}) 
        = (\lambda_{L}(t_{\alpha}) + \alpha\lambda_{O}(t_{\alpha}) )\sum_{e\in L}c(e)
    \end{align*}
Therefore \(\mathbb{E}[\nu_{T}(W)]\) is bounded by:
\begin{align*}
   \frac{\sum_{e\in L\cup O}c(e) \mathbb{E}[H_{\bar w(e)}] }{\sum_{e\in L\cup O}c(e)} 
        \leq \frac{(\lambda_{L}(t_{\alpha}) + \alpha\lambda_{O}(t_{\alpha})) \sum_{e\in L}c(e)}{ (\alpha+1)\sum_{e\in L}c(e)  } 
        = \frac{\lambda_{L}(t_{\alpha}) + \alpha\lambda_{O  }(t_{\alpha}) }{\alpha + 1} \leq \frac{991}{732}.
 \end{align*}
where the last inequality follows 
using 
Lemma~\ref{lem:lambdabound} and~\ref{lem:upperbound}.
\end{proof}
Now Theorem~\ref{thm:clawupper} follows by combining Lemma~\ref{lem:gamma_claw} with Theorem~\ref{thm:gamma} in which $\gamma$ is replaced by the supremum taken over all Steiner-claw free instances (rather than over all Steiner Tree instances).

\paragraph*{ Tightness of the bound.}
\begin{figure}[t]
    \centering
    \begin{tikzpicture}
                   
        \tikzset{black dot/.style={draw=black, very thick, circle,minimum size=0pt, inner sep=1pt, outer sep=1pt,fill=black}}
        \tikzset{terminal/.style={draw=black,  thick,minimum size=0pt, inner sep=2.5pt, outer sep=1pt}}
        \tikzset{P node/.style={fill={rgb,255: red,20; green,154; blue,0}, draw={rgb,255: red,20; green,154; blue,0}, circle, minimum size=0pt,inner sep=1pt, outer sep=1pt}}
    
        \tikzstyle{witness edge}=[-, draw={rgb,255: red,195; green,0; blue,3}, very thick]
        \tikzstyle{T edges}=[-, very thick]
        \tikzstyle{new witness}=[-, draw={rgb,255: red,195; green,0; blue,3}, dashed]
        \tikzstyle{connected terminals}=[-, draw=black, dashed, very thick]
        \tikzstyle{P}=[-, draw={rgb,255: red,20; green,154; blue,0}, very thick]
        \tikzstyle{Marked Edges}=[-, draw={rgb,255: red,255; green,207; blue,14}, very thick]
        
		\node [style=black dot] (37) at (-14, 2.25) {};
		\node [style=black dot] (38) at (-13, 2.25) {};
		\node [style=black dot] (39) at (-12, 2.25) {};
		\node [style=black dot] (40) at (-11, 2.25) {};
		\node [style=black dot] (41) at (-10, 2.25) {};
		\node [style=black dot] (42) at (-9, 2.25) {};
		\node [style=black dot] (43) at (-8, 2.25) {};
		\node [style=black dot] (44) at (-7, 2.25) {};
		\node [style=black dot] (45) at (-6, 2.25) {};
		\node [style=black dot] (46) at (-5, 2.25) {};
		\node [style=black dot] (47) at (-4, 2.25) {};
		\node [style=black dot] (48) at (-3, 2.25) {};
		\node [style=black dot] (49) at (-2, 2.25) {};
		\node [style=black dot] (50) at (-1, 2.25) {};
  
		\node [style=terminal] (51) at (-14, 1.5) {};
		\node [style=terminal] (52) at (-13, 1.5) {};
		\node [style=terminal] (53) at (-12, 1.5) {};
		\node [style=terminal] (54) at (-11, 1.5) {};
		\node [style=terminal] (55) at (-10, 1.5) {};
		\node [style=terminal] (56) at (-9, 1.5) {};
		\node [style=terminal] (57) at (-8, 1.5) {};
		\node [style=terminal] (58) at (-7, 1.5) {};
		\node [style=terminal] (59) at (-6, 1.5) {};
		\node [style=terminal] (60) at (-5, 1.5) {};
		\node [style=terminal] (61) at (-4, 1.5) {};
		\node [style=terminal] (62) at (-3, 1.5) {};
		\node [style=terminal] (63) at (-2, 1.5) {};
		\node [style=terminal] (64) at (-1, 1.5) {};
  
		\node  (71) at (-10.5, 2.5) {$1$};
		\node  (72) at (-11.5, 2.5) {$2$};
		\node  (73) at (-12.5, 2.5) {$2$};
		\node  (74) at (-12.75, 1.85) {$1$};
		\node  (75) at (-11.75, 1.85) {$4$};
		\node  (76) at (-10.75, 1.85) {$1$};
		\node  (78) at (-7.5, 2.5) {$1$};
		\node  (79) at (-8.5, 2.5) {$2$};
		\node  (80) at (-9.5, 2.5) {$2$};
		\node  (81) at (-9.75, 1.85) {$1$};
		\node  (82) at (-8.75, 1.85) {$4$};
		\node  (83) at (-7.75, 1.85) {$1$};
		\node  (85) at (-4.5, 2.5) {$1$};
		\node  (86) at (-5.5, 2.5) {$2$};
		\node  (87) at (-6.5, 2.5) {$2$};
		\node  (88) at (-6.75, 1.85) {$1$};
		\node  (89) at (-5.75, 1.85) {$4$};
		\node  (90) at (-4.75, 1.85) {$1$};
		\node  (92) at (-1.5, 2.5) {$1$};
		\node  (93) at (-2.5, 2.5) {$2$};
		\node  (94) at (-3.5, 2.5) {$2$};
		\node  (95) at (-3.75, 1.85) {$1$};
		\node  (96) at (-2.75, 1.85) {$4$};
		\node  (97) at (-1.75, 1.85) {$1$};
		\node  (98) at (-0.75, 1.85) {$1$};
		\node  (99) at (-13.75, 1.85) {$1$};
		\node  (100) at (-13.5, 2.5) {$1$};

        \draw [style=T edges] (37) to (50);
		\draw [style=T edges] (50) to (64);
		\draw [style=T edges] (37) to (51);
		\draw [style=T edges] (38) to (52);
		\draw [style=T edges] (39) to (53);
		\draw [style=T edges] (40) to (54);
		\draw [style=T edges] (41) to (55);
		\draw [style=T edges] (42) to (56);
		\draw [style=T edges] (43) to (57);
		\draw [style=T edges] (44) to (58);
		\draw [style=T edges] (45) to (59);
		\draw [style=T edges] (46) to (60);
		\draw [style=T edges] (47) to (61);
		\draw [style=T edges] (48) to (62);
		\draw [style=T edges] (49) to (63);
		\draw [style=witness edge] (52) to (53);
		\draw [style=witness edge] (53) to (54);
		\draw [style=witness edge] (55) to (56);
		\draw [style=witness edge] (56) to (57);
		\draw [style=witness edge] (58) to (59);
		\draw [style=witness edge] (59) to (60);
		\draw [style=witness edge] (61) to (62);
		\draw [style=witness edge] (62) to (63);
		\draw [style=witness edge, bend right=15, looseness=1.25] (51) to (53);
		\draw [style=witness edge, bend right=15] (53) to (56);
		\draw [style=witness edge, bend right=15] (56) to (59);
		\draw [style=witness edge, bend right=15] (59) to (62);
		\draw [style=witness edge, bend right=15, looseness=1.25] (62) to (64);
    \end{tikzpicture}
    \caption{Lower bound instance shown in black with $c(e)=1$ for all the edges in $L$ and $c(e)=\alpha$ for all the edges in $O$, for $\alpha=\frac{32}{90}$. The white squares are terminals and black circles are Steiner nodes. Red edges form the laminar witness tree \(\W\), with the numbers next to each edge the value of $w$ imposed on $T$. 
    }
    \label{fig:clawwitnesstree}
\end{figure}
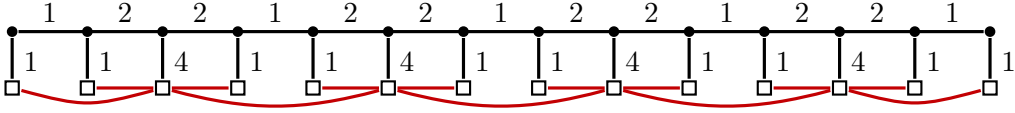

We conclude this section by spending a few words on Theorem~\ref{thm:clawlowerbound}. Our lower-bound instance is obtained by taking a tree $T$ on $q$ Steiner nodes, each adjacent to one terminal, with $c(e)=1$ for all the edges in $L$ and $c(e)=\alpha$ for all the edges in $O$, for $\alpha=\frac{32}{90}$. 
Similar to Section 3, a crucial ingredient for our analysis is in utilizing Theorem~\ref{thm:edgelaminar} stating that there is an optimal laminar witness tree. See Figure~\ref{fig:clawwitnesstree}.
We use this to show that there is an optimal witness tree for our tree \(T\), whose objective value is at least \(\frac{991}{732}-\varepsilon\). Details can be found in
{\color{black}the full version of the paper.}
\section*{Acknowledgements}
The authors are very grateful to Haris Angelidakis for many important discussions on this work.

\bibliographystyle{alpha}
\bibliography{refs}

\newcommand{\etalchar}[1]{$^{#1}$}
\begin{thebibliography}{BFG{\etalchar{+}}14}

\bibitem[Adj19]{DBLP:journals/talg/Adjiashvili19}
David Adjiashvili.
\newblock Beating approximation factor two for weighted tree augmentation with
  bounded costs.
\newblock {\em {ACM} Trans. Algorithms}, 15(2):19:1--19:26, 2019.

\bibitem[AHS22]{angelidakis2022node}
Haris Angelidakis, Dylan {Hyatt-Denesik}, and Laura Sanit{\`a}.
\newblock Node connectivity augmentation via iterative randomized rounding.
\newblock {\em Mathematical Programming}, pages 1--37, 2022.

\bibitem[BFG{\etalchar{+}}14]{DBLP:conf/icalp/BasavarajuFGMRS14}
Manu Basavaraju, Fedor~V. Fomin, Petr~A. Golovach, Pranabendu Misra, M.~S.
  Ramanujan, and Saket Saurabh.
\newblock Parameterized algorithms to preserve connectivity.
\newblock In {\em Proceedings of the 41st International Colloquium on Automata,
  Languages, and Programming ({ICALP})}, pages 800--811, 2014.

\bibitem[BGJ20]{DBLP:conf/stoc/Byrka0A20}
Jaroslaw Byrka, Fabrizio Grandoni, and Afrouz {Jabal Ameli}.
\newblock Breaching the 2-approximation barrier for connectivity augmentation:
  a reduction to {S}teiner tree.
\newblock In {\em Proceedings of the 52nd Annual {ACM} {SIGACT} Symposium on
  Theory of Computing ({STOC})}, pages 815--825, 2020.

\bibitem[BGRS13]{DBLP:journals/jacm/ByrkaGRS13}
Jaroslaw Byrka, Fabrizio Grandoni, Thomas Rothvo{\ss}, and Laura Sanit{\`{a}}.
\newblock Steiner tree approximation via iterative randomized rounding.
\newblock {\em J. {ACM}}, 60(1):6:1--6:33, 2013.

\bibitem[CG18a]{DBLP:journals/algorithmica/CheriyanG18}
Joseph Cheriyan and Zhihan Gao.
\newblock Approximating (unweighted) tree augmentation via lift-and-project,
  part {I:} stemless {TAP}.
\newblock {\em Algorithmica}, 80(2):530--559, 2018.

\bibitem[CG18b]{DBLP:journals/algorithmica/CheriyanG18a}
Joseph Cheriyan and Zhihan Gao.
\newblock Approximating (unweighted) tree augmentation via lift-and-project,
  part {II}.
\newblock {\em Algorithmica}, 80(2):608--651, 2018.

\bibitem[CTZ21]{DBLP:journals/corr/abs-2012-00086}
Federica Cecchetto, Vera Traub, and Rico Zenklusen.
\newblock Bridging the gap between tree and connectivity augmentation: unified
  and stronger approaches.
\newblock In {\em Proceedings of the 53rd Annual {ACM} {SIGACT} Symposium on
  Theory of Computing ({STOC})}, pages 370--383. {ACM}, 2021.

\bibitem[DKL76]{dinitz1976structure}
Efim~A Dinitz, Alexander~V Karzanov, and Michael~V Lomonosov.
\newblock On the structure of the system of minimum edge cuts in a graph.
\newblock {\em Issledovaniya po Diskretnoi Optimizatsii}, pages 290--306, 1976.

\bibitem[FGKS18]{DBLP:conf/soda/Fiorini0KS18}
Samuel Fiorini, Martin Gro{\ss}, Jochen K{\"{o}}nemann, and Laura Sanit{\`{a}}.
\newblock Approximating weighted tree augmentation via chv{\'{a}}tal-gomory
  cuts.
\newblock In {\em Proceedings of the 29th Annual {ACM-SIAM} Symposium on
  Discrete Algorithms ({SODA})}, pages 817--831. {SIAM}, 2018.

\bibitem[FJ81]{DBLP:journals/siamcomp/FredericksonJ81}
Greg~N. Frederickson and Joseph J{\'{a}}J{\'{a}}.
\newblock Approximation algorithms for several graph augmentation problems.
\newblock {\em {SIAM} J. Comput.}, 10(2):270--283, 1981.

\bibitem[FKOS16]{feldmann2016equivalence}
Andreas~Emil Feldmann, Jochen K{\"o}nemann, Neil Olver, and Laura Sanit{\`a}.
\newblock On the equivalence of the bidirected and hypergraphic relaxations for
  steiner tree.
\newblock {\em Mathematical programming}, 160(1):379--406, 2016.

\bibitem[GKZ18]{DBLP:conf/stoc/0001KZ18}
Fabrizio Grandoni, Christos Kalaitzis, and Rico Zenklusen.
\newblock Improved approximation for tree augmentation: saving by rewiring.
\newblock In {\em Proceedings of the 50th Annual {ACM} {SIGACT} Symposium on
  Theory of Computing {(STOC)}}, pages 632--645. {ACM}, 2018.

\bibitem[GORZ12]{10.1145/2213977.2214081}
Michel~X. Goemans, Neil Olver, Thomas Rothvo\ss{}, and Rico Zenklusen.
\newblock Matroids and integrality gaps for hypergraphic steiner tree
  relaxations.
\newblock In {\em Proceedings of the Forty-Fourth Annual ACM Symposium on
  Theory of Computing}, STOC '12, page 1161–1176, New York, NY, USA, 2012.
  Association for Computing Machinery.

\bibitem[Nut20]{nutov20202nodeconnectivity}
Zeev Nutov.
\newblock 2-node-connectivity network design.
\newblock In {\em Proceedings of the 18th International Workshop on
  Approximation and Online Algorithms ({WAOA})}, volume 12806 of {\em Lecture
  Notes in Computer Science}, pages 220--235. Springer, 2020.

\bibitem[Nut21]{DBLP:journals/corr/abs-2009-13257}
Zeev Nutov.
\newblock Approximation algorithms for connectivity augmentation problems.
\newblock In {\em Proceedings of the 16th International Computer Science
  Symposium in Russia ({CSR})}, volume 12730, pages 321--338. Springer, 2021.

\bibitem[TZ22a]{https://doi.org/10.48550/arxiv.2209.07860}
Vera Traub and Rico Zenklusen.
\newblock A $(1.5+\varepsilon)$-approximation algorithm for weighted
  connectivity augmentation, 2022.

\bibitem[TZ22b]{traub2022better}
Vera Traub and Rico Zenklusen.
\newblock A better-than-2 approximation for weighted tree augmentation.
\newblock In {\em 2021 IEEE 62nd Annual Symposium on Foundations of Computer
  Science (FOCS)}, pages 1--12. IEEE, 2022.

\bibitem[TZ22c]{traub2022local}
Vera Traub and Rico Zenklusen.
\newblock Local search for weighted tree augmentation and steiner tree.
\newblock In {\em Proceedings of the 2022 Annual ACM-SIAM Symposium on Discrete
  Algorithms (SODA)}, pages 3253--3272. SIAM, 2022.

\end{thebibliography}

\begin{appendix}
    \section{ Proofs of Section~\ref{sec:laminarity}}
{\color{black} In this section we discuss the proofs missing in Section~\ref{sec:laminarity}. Section~\ref{sec:edgelamproof} includes the proof of Theorem~\ref{thm:edgelaminar} and Section~\ref{sec:laminarityproof} includes the proof of Theorem~\ref{thm:contractlam}.}
\subsection{ Proof of Theorem~\ref{thm:edgelaminar}}
\label{sec:edgelamproof}

\edgelaminar*
\begin{proof}
    We first show that there is a witness tree $W$ minimizing $\bnu_T(W)$ such that the subgraph of $W$ induced on the terminals of any maximally connected region in $T$ of zero cost edges is a star. We assume for the sake of contradiction that such a maximally connected region $F\subseteq E$ exists where the subgraph of $W$ induced on the terminals of $F$ is a set of connected components $W_1, \dots, W_i$, for $i >1$. {\color{black}First, if there is an edge $e=uv\in E_W$ such that $V[T_e] \cap V[F] \neq \emptyset$, and $u,v \notin V[F]$, then the solution can be improved by replacing $uv$ with some an edge having one endpoint in $F$. To see this, first note that since $T$ is a tree, $u$ and $v$ are in separate components of $G\backslash\{uv\}$. Fix terminal $r\in R\cap V[F]$, and without loss of generality $r$ is in the same component as $u$. So we can replace $uv$ with $rv$ to find a solution with no greater cost.}
    
    So we can assume that any edge $e\in E_W$ such that $E[T_e] \cap F \neq \emptyset$ has an endpoint in $F$. Without loss of generality, suppose the shortest path between two components is from $W_1$ to $W_2$, and let $e$ denote the edge of this path incident to $W_2$. For an arbitrary but fixed edge $f$ between $W_1$ and $W_2$ we define $W'\coloneqq W\cup \{f\} \backslash \{e\}$. Clearly, we can see $\sum_{e'\in E[T_f] } c(e') = 0 < \sum_{e'\in E[T_e]} c(e')$, so we have $\bnu_T(W') < \bnu_T(W)$, contradicting the minimality of $W$. We can rearrange the edges between the terminals of $F$ to be a star, as this will not affect $\bnu_T(W)$. So we assume that this holds on $W$ for any such zero cost region of $T$.

    For a maximally connected region of zero cost edges $F\subseteq E$, by an abuse of notation, we will simply refer to the induced star subgraph of $W$ as $F$, and denote its center by $s$. We assume without loss of generality that edges of $W$ incident to $F$ have endpoint at $s$. To see this, first note that $F$ is a connected subgraph of $W$, so any edges incident to $F$ cannot share an endpoint outside of $F$, otherwise we have a cycle in $W$. Furthermore, for any edge of $W$ incident to $F$ with endpoint not equal to $s$, we can change the endpoint of that edge to be $s$ and maintain the connectivity of $W$ since $F$ is connected. Edges changed in this way will have the same edges between their endpoints except for those in the region $F$, which is zero cost, so this does not increase $\bnu_T(W)$. 

    We assume for the sake of contradiction that the witness tree $W$ minimizing $\nu_T(W)$ is not a laminar witness tree, and that it has the minimum number of pairs of crossing edges. That is, there exist distinct leaves $r_1,r_2,r_3,r_4 \in R$  such that $e_1=r_1r_2, e_2=r_3r_4\in E_W$ are crossing.  We denote the path $T_{e_1}\cap T_{e_2}$ by $P$. We denote the shortest path from $P$ to $r_i$ by $P_i$.  

    Since $e_1$ and $e_2$ are crossing edges, one of $T_{r_1r_3}$ or $T_{r_1r_4}$ contains exactly one node of $P$. The same is true for $r_2$. Without loss of generality, let us that the paths $T_{r_1r_3}$ and $T_{r_2r_4}$ contain exactly one node of $P$. We consider by cases which component of $W\backslash \{e_1, e_2\}$ contains two of  $r_1, r_2, r_3$ and $r_4$.
    \begin{itemize}
        \item Case: $r_1$ and $r_3$ (or similarly, $r_2$ and $r_4$) are in the same component $W\backslash \{e_1, e_2\}$.  Note, that if $|V(P)| =1$, then we can assume that we are in this case without loss of generality. If $\sum_{e\in E[P_1]} c(e) + \sum_{e\in E[P_3]}c(e)=0$, then as we have shown above, $e_1$ and $e_2$ are assumed to share an endpoint, and are thus not crossing. So we have that $\sum_{e\in E[P_1]} c(e) + \sum_{e\in E[P_3]} c(e)>0$.
        Consider $W' \coloneqq  W\cup \{r_2r_3\}\backslash\{e_1\}$ and $W'' \coloneqq  W\cup \{r_1r_4\} \backslash \{e_2\}$. If $\bnu_T(W) - \bnu_T(W') > 0$, this contradicts the minimality of $\bnu_T(W)$. Therefore, we can see
        \begin{align*}
            0 &\leq  c(e) (\bnu_T(W') - \bnu_T(W) ) = \sum_{e\in E[P_3]} \frac{c(e)}{w(e)+1} -  \sum_{e\in E[P_1]} \frac{c(e)}{w(e)} \\
            & < \sum_{e\in E[P_3]} \frac{c(e)}{w(e)} -  \sum_{e\in E[P_1]} \frac{c(e)}{w(e)+1} = c(e) ( \bnu_T(W) - \bnu_T(W'') )
        \end{align*}
        Clearly, we have $\bnu_T(W'') < \bnu_T(W)$, contradicting the minimality of $\bnu_T(W)$. 
        
        \item Case: $r_2$ and $r_3$ (or similarly, $r_1$ and $r_4$) are in the same component of $W\backslash \{e_1, e_2\}$. In this case, consider $W' \coloneqq \{r_1r_3,r_2r_4\} \backslash \{e_1,e_2\}$. If $\sum_{e\in E[P]} c(e) = 0$, then clearly $\bnu_T(W') = \bnu_T(W)$, but $W'$ has one fewer crossing pair, contradicting the assumption that $W$ minimizes the number of such pairs, thus $\sum_{e\in E[P]} c(e) > 0$. Without loss of generality, we can assume that $|V(P)| >1$, because if $|V(P)|=1$ then we can reduce to the previous case by relabelling the nodes $r_1,r_2,r_3$, and $r_4$. Therefore, we can see the following
        \begin{align*}
            c(e)\left(\bnu_T(W') - \bnu_T(W) \right) \leq - \sum_{e\in E[P]} \frac{c(e)}{w(e)} < 0
        \end{align*}
    \end{itemize}
    Clearly, we have $\bnu_T(W') < \bnu_T(W)$, contradicting the minimality of $\bnu_T(W)$. 
\end{proof}

\subsection{ Proof of Theorem~\ref{thm:contractlam}}
\label{sec:laminarityproof}
\begin{proof}
    $\Rightarrow )$ Consider a witness tree $W = (R, E_W)$ of $T$ found by marking and contraction. Assume for the sake of contradiction that $W$ is not laminar. That is, assume there are distinct edges $e_1, e_2 \in E_W$ that are crossing. By the method of marking and contraction, for $i=1,2$, we know that the nodes of $T_{e_i}$ are contained precisely in two separate connected regions of marked edges, denoted $M_{i,1}$ and $M_{i,2}$.
    
    Therefore, the endpoints of $e_i$ are the unique leaves that belong to the connected regions containing $M_{i,1}$ and $M_{i,2}$ respectively. Therefore, if $T_{e_1}\cap T_{e_2} \neq \emptyset$, then $e_1$ and $e_2$ share an endpoint, contradicting our assumption.

    $\Leftarrow )$ Let $W = (R, E_W)$ be a laminar witness tree. Our goal is to find $W$ by marking and contraction. {\color{black}As a simplifying step we contract any node of $T$ that has degree $2$, as any two edges in $E_W$ that share a degree $2$ node between their endpoints must share another node. }
    For edge $f \in E$, we mark $f$ if there are distinct edges $e, e'\in E_W$ such that $f\in E[T_{e}\cap T_{e'}]$.

    First, we want to show that for any edge $e \in E_W$,  there is at most one edge of $T_e$ that is unmarked. Assume for contradiction that distinct edges $f_1, f_2 \in T_{e}$ are both unmarked. Of the three connected components of $T \backslash \{ f_1, f_2\}$, consider the component that is incident to both $f_1$ and $f_2$, which we denote $T'$. We assumed $T$ has no degree $2$ vertices, so  there is at least one leaf $r \in T'\cap R$, therefore, there is a minimal path in $W$ from $r$ to an endpoint of $e$. So there is an edge not equal to $e$ with at least one of $f_1$ or $f_2$ between its endpoints, and thus that edge is marked, which is a contradiction. 
    See Figure~\ref{fig:optlaminar}.(a) for an example.

    It remains to show that for every edge $e = rr' \in E_W$, there is at least one edge on the $T_e$ path that is unmarked. We assume for contradiction that that every edge in $E_W$ on the path $T_e$ is marked, and we enumerate the nodes of $T_e$ as $r=v_0, v_1, \dots, v_k = r'$ in order. The edge $v_0v_1$ is marked, so there is an edge with endpoint at $v_0$ not equal to $e$. Pick $e_i$ to be edge that maximizes $\{v_0,v_1, \dots, v_i\} \subset T_{e_i}$, denote its endpoints as $r$ and $r_i$. Since the edge $v_iv_{i+1}$ is marked, there is an edge $e' \in E_W$ such that $v_iv_{i+1} \in E[T_{e'}]$, where $e' \neq e, e_i$. Since $W$ is laminar, we know that $e'$ must share an endpoint with $e$ and with $e_i$. We picked $e_i$ to be the edge in $E_W$ that maximizes the set $\{v_0,v_1, \dots, v_i\} \subset T_{e_i}$, so $e'$ cannot have $r$ as an endpoint, and therefore must have $r'$ as an endpoint. Similarly, the second endpoint of $e'$ must be $r_i$ so $e'$ does not cross with $e_i$. Therefore, the edges $e,e_i,e' \in E_W$ form a cycle in $W$, contradicting the assumption that $W$ is a tree. See Figure~\ref{fig:optlaminar}.(b) for an example.
    
    Furthermore, every edge $e\in E_W$ has a one to one correspondence to an edge of $E$, which is the unique edge $T_e$ that is unmarked. It remains to show that the connected regions of marked edges each contain exactly one leaf, by contracting these marked regions, the resulting tree on the unmarked edges will be exactly $W$. First, consider the connected regions of marked edges of $T$. By the one to one correspondence between unmarked edges and $E_W$, we have $|R|-1$ unmarked edges, and thus $|R|$ connected regions of marked edges. 
    \begin{figure}[t]
    \begin{center}
        \begin{tabular}{c|c}
            \begin{tikzpicture}[scale=0.75]
                
                \tikzset{black dot/.style={draw=black, very thick, circle,minimum size=0pt, inner sep=1pt, outer sep=1pt,fill=black}}
                \tikzset{terminal/.style={draw=black,  thick,minimum size=0pt, inner sep=2.5pt, outer sep=1pt}}
                \tikzset{P node/.style={fill={rgb,255: red,20; green,154; blue,0}, draw={rgb,255: red,20; green,154; blue,0}, circle, minimum size=0pt,inner sep=1pt, outer sep=1pt}}
            
                \tikzstyle{witness edge}=[-, draw={rgb,255: red,195; green,0; blue,3}, very thick]
                \tikzstyle{T edges}=[-, very thick]
                \tikzstyle{new witness}=[-, draw={rgb,255: red,195; green,0; blue,3}, dashed]
                \tikzstyle{connected terminals}=[-, draw=black, dashed, very thick]
                \tikzstyle{P}=[-, draw={rgb,255: red,20; green,154; blue,0}, very thick]
                \tikzstyle{Marked Edges}=[-, draw={rgb,255: red,255; green,207; blue,14}, very thick]
                    
    		\node [style=black dot] (1) at (-9.5, 3.25) {};
    		\node [style=terminal] (2) at (-9, 2.5) {};
    		\node [style=black dot] (3) at (-8.25, 3.5) {};
    		\node [style=terminal] (5) at (-7.5, 2.25) {};
    		\node [style=black dot] (6) at (-7, 3.25) {};
    		\node [style=black dot] (7) at (-5.75, 2.25) {};
    		\node [style=black dot] (8) at (-6.5, 1.75) {};
    		\node [style=black dot] (9) at (-5, 1.5) {};
    		\node [style=terminal] (10) at (-5.5, 1) {};
    		\node [style=terminal] (11) at (-4.75, 0.5) {};
    		\node [style=terminal] (12) at (-7, 1.25) {};
    		\node [style=black dot] (13) at (-10.75, 2.5) {};
    		\node [style=black dot] (14) at (-10.25, 2) {};
    		\node [style=terminal] (15) at (-9.5, 1.5) {};
    		\node [style=black dot] (16) at (-11.5, 1.75) {};
    		\node [style=terminal] (17) at (-12, 0.75) {};
    		\node [style=terminal] (18) at (-11, 0.5) {};
    		\node [style=terminal] (19) at (-6, 1.25) {};
    		\node [style=terminal] (20) at (-10.5, 1.25) {};
    		\node [style=terminal] (21) at (-6.75, 2.25) {};
    		\node  (22) at (-10.3, 3.05) {$f_1$};
    		\node  (23) at (-6.1, 2.925) {$f_2$};
    		\node  (24) at (-8, 0.25) {$e$};
    		\node  (25) at (-8.5, 2.25) {$r$};
    		\node [style=terminal] (26) at (-8, 2.75) {};
                \node (31) at (-12, 3.5) {$(a)$};
                
    		\draw [style=T edges] (1) to (13);
    		\draw [style=T edges] (13) to (14);
    		\draw [style=T edges] (14) to (20);
    		\draw [style=T edges] (14) to (15);
    		\draw [style=P] (13) to (16);
    		\draw [style=T edges] (16) to (17);
    		\draw [style=P] (16) to (18);
    		\draw [style=T edges] (1) to (2);
    		\draw [style=P] (1) to (3);
    		\draw [style=P] (3) to (6);
    		\draw [style=T edges] (6) to (5);
    		\draw [style=T edges] (6) to (21);
    		\draw [style=T edges] (6) to (7);
    		\draw [style=T edges] (7) to (8);
    		\draw [style=P] (7) to (9);
    		\draw [style=T edges] (9) to (10);
    		\draw [style=P] (9) to (11);
    		\draw [style=T edges] (8) to (19);
    		\draw [style=T edges] (8) to (12);
    		\draw [style=witness edge, bend right, looseness=0.50] (18) to (11);
    		\draw [style=witness edge, in=15, out=-75, looseness=1.50] (2) to (18);
		      \draw [style=T edges] (3) to (26);

            \end{tikzpicture}
            &
            \begin{tikzpicture}[scale=0.75]
                
                \tikzset{black dot/.style={draw=black, very thick, circle,minimum size=0pt, inner sep=1pt, outer sep=1pt,fill=black}}
                \tikzset{terminal/.style={draw=black,  thick,minimum size=0pt, inner sep=2.5pt, outer sep=1pt}}
                \tikzset{P node/.style={fill={rgb,255: red,20; green,154; blue,0}, draw={rgb,255: red,20; green,154; blue,0}, circle, minimum size=0pt,inner sep=1pt, outer sep=1pt}}
            
                \tikzstyle{witness edge}=[-, draw={rgb,255: red,195; green,0; blue,3}, very thick]
                \tikzstyle{T edges}=[-, very thick]
                \tikzstyle{new witness}=[-, draw={rgb,255: red,195; green,0; blue,3}, dashed, very thick]
                \tikzstyle{connected terminals}=[-, draw=black, dashed, very thick]
                \tikzstyle{P}=[-, draw={rgb,255: red,20; green,154; blue,0}, very thick]
                
    		\node [style=black dot] (1) at (-9.5, 3.25) {};
    		\node [style=terminal] (2) at (-9, 2.5) {};
    		\node [style=black dot] (3) at (-8.25, 3.5) {};
    		\node [style=terminal] (5) at (-7.5, 2.25) {};
    		\node [style=black dot] (6) at (-7, 3.25) {};
    		\node [style=black dot] (7) at (-5.75, 2.25) {};
    		\node [style=black dot] (8) at (-6.5, 1.75) {};
    		\node [style=black dot] (9) at (-5, 1.5) {};
    		\node [style=terminal] (10) at (-5.5, 1) {};
    		\node [style=terminal] (11) at (-4.75, 0.5) {};
    		\node [style=terminal] (12) at (-7, 1.25) {};
    		\node [style=black dot] (13) at (-10.75, 2.5) {};
    		\node [style=black dot] (14) at (-10.25, 2) {};
    		\node [style=terminal] (15) at (-9.5, 1.5) {};
    		\node [style=black dot] (16) at (-11.5, 1.75) {};
    		\node [style=terminal] (17) at (-12, 0.75) {};
    		\node [style=terminal] (18) at (-11, 0.5) {};
    		\node [style=terminal] (19) at (-6, 1.25) {};
    		\node [style=terminal] (20) at (-10.5, 1.25) {};
    		\node [style=terminal] (21) at (-6.75, 2.25) {};
    		\node  (24) at (-8, 0.25) {$e$};
    		\node  (25) at (-8.575, 2.475) {$r_i$};
    		\node [style=terminal] (26) at (-8, 2.75) {};
    		\node  (27) at (-9.5, 3.5) {$v_i$};
    		\node  (28) at (-8.95, 0.95) {$e_i$};
    		\node  (29) at (-7.9, 0.975) {$e'$};
    		\node  (30) at (-11.25, 0.125) {$r$};
    		\node  (31) at (-4.75, 0.15) {$r'$};
      		\node  (32) at (-8.425, 3.75) {$v_{i+1}$};
                \node (31) at (-12, 3.5) {$(b)$};
    
                \draw [style=P] (1) to (13);
    		\draw [style=T edges] (13) to (14);
    		\draw [style=T edges] (14) to (20);
    		\draw [style=T edges] (14) to (15);
    		\draw [style=P] (13) to (16);
    		\draw [style=T edges] (16) to (17);
    		\draw [style=P] (16) to (18);
    		\draw [style=T edges] (1) to (2);
    		\draw [style=P] (1) to (3);
    		\draw [style=P] (3) to (6);
    		\draw [style=T edges] (6) to (5);
    		\draw [style=T edges] (6) to (21);
    		\draw [style=P] (6) to (7);
    		\draw [style=T edges] (7) to (8);
    		\draw [style=P, in=135, out=-45] (7) to (9);
    		\draw [style=T edges] (9) to (10);
    		\draw [style=P, in=105, out=-75] (9) to (11);
    		\draw [style=T edges] (8) to (19);
    		\draw [style=T edges] (8) to (12);
    		\draw [style=witness edge, bend right, looseness=0.50] (18) to (11);
    		\draw [style=witness edge, in=15, out=-75, looseness=1.50] (2) to (18);
    		\draw [style=T edges] (3) to (26);
    		\draw [style=witness edge, bend right] (2) to (11);
            \end{tikzpicture}
    
        \end{tabular}
    \end{center}
    \caption{ Returning to the tree $T$ from Figure~\ref{fig:laminarity} the green edges denote the edges marked on $T_e$, and the red edges denote edges of $W$. 
    Figure $(a)$: $f_1, f_2\in T_e$ are unmarked. The component of $T\backslash\{f_1, f_2\}$ has a terminal $r$ and there must be a path from $r$ to an endpoint of $e$ in $W$. So $f_1$ must be marked by definition.
    Figure $(b)$: Every edge of $T_e$ is marked. So taking $e_i$ the edge with endpoint $r$ that maximally intersects $T_e$, has endpoint at $r_i$. The edge $v_{i+1}$ is marked, so $e'$ must have endpoints at $r_i$ and $r'$ by laminarity since it shares $v_i$ with $e_i$ and $e$.
    }
    \label{fig:optlaminar}
\end{figure}
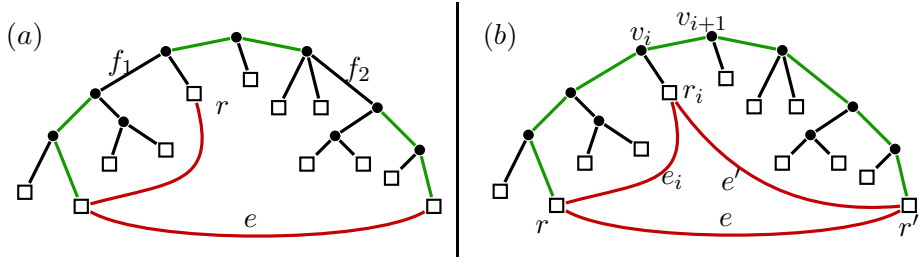

    We assume for contradiction that there is a maximal connected region of marked edges that does not contain a leaf, which we denote by $C$, noting that $C$ is itself a tree. Consider a leaf $v$ of $C$. Clearly, $v$ is not a leaf of $T$, and is incident to at least two unmarked edges $f_1, f_2\in E$, as $T$ is assumed to have no nodes of degree $2$. By the way we find these unmarked edges there are unique and distinct edges $e_1, e_2\in E_W$ such that $f_1 \in E[T_{e_1}]$ and $f_2\in E[T_{e_2}]$. By the laminarity of $W$, since $v \in T_{e_1} \cap T_{e_2}$, $e_1$ and $e_2$ must share an endpoint, and so there is a path from $v$ to a leaf $r\in R$ of marked edges. Since $C$ is a maximal region of marked edges we see that $r\in C$, contradicting the assumption that $C$ contains no leaves.
\end{proof}
    \section{ Proofs for Section~\ref{sec:approximationCA-Node}}
{\color{black}
In this section we provide the complete proofs required for Section~\ref{sec:approximationCA-Node}. In particular, we provide a complete proof of Lemma~\ref{lem:increase} in Section~\ref{sec:bounds}. A proof of Lemma~\ref{lem:case1} is found in Section~\ref{sec:case1}. And finally, we complete the proof of Lemma~\ref{lem:invariant} in Section~\ref{sec:othercases}.
}

The following observation will be a useful tool throughout this section.
\begin{observation}
\label{lem:child}
        Let \(u \in T_i\backslash r_i\) and \(k\ge 1\) be the number of its children, enumerated \(u_1, \dots, u_k\). Let \(W^{u_1},\dots,W^{u_k}\) be the witness trees of \(Q_{u_1}, \dots ,Q_{u_k}\). If \(u_m\) is the marked child of \(u\), then for all \(v\in P(u_m)\), \(w^u(v) = w^{u_m}(v) + k-1\)
\end{observation}
To see this, recall that if \(u_m\) is the marked child of \(u\), then \(W^u\) is equal to \(\overline{W}^{u_1}, \dots, \overline{W}^{u_k}\)  plus the edges \(\ell(u_m)\ell(u_j)\) \(j\neq m\), and \(e^u\). 
\subsection{  Proof of Lemma~\ref{lem:increase}}
\label{sec:bounds}
Recall that \(u\) has no final node children, and \(C_{1}^{u_j} = \min_{j\in [k]} C^{u_j}_1\) so \(u_1\) is the marked child of \(u\). We restate Lemma~\ref{lem:increase}  here.

\lemIncrease*
\begin{proof}
    Recall \(W^u\) is the union of \(\overline{W}^{u_1}, \dots, \overline{W}^{u_k}\)  plus the edges \(e^u\), and \(\ell(u_1)\ell(u_j)\) for $j=2, \dots, k$. 
    \begin{enumerate}[label=(\alph*)]
        \item None of the trees \(\overline{W}^{u_1}, \dots, \overline{W}^{u_k}\) contain \(u\), thus the edges \(\ell(u_1)\ell(u_2), \dots, \ell(u_1)\ell(u_k)\), and \(e^u\) are the only edges that contribute to \(w^u(u)\). Thus \(w^u(u) = k\).
        \item The only edges in \(W^u\) with endpoints in \(Q_{u_j}\) for \(j\in \{2,\dots, k\}\), are the  edges of \(W^{u_j}\). 
        \item This is shown with a similar argument to (b).
        \item Finally, we can see
        \begin{align*}
            &\sum_{v \in P(u_1)\setminus\ell(u_1)} H_{w^{u_1} (v)} + \sum_{j=1}^{k-1}C^{u_1}_{j} \\
            &= \sum_{v \in P(u_1)\setminus\ell(u_1)} H_{w^{u_1} (v)} + \sum_{j=1}^{k-1} \sum_{v\in P(u_1)\backslash \ell(u_1)} \big(H_{w^{u_1}(v)+j} - H_{w^{u_1}(v)+j-1}\big)\\
            &= \sum_{v \in P(u_1)\setminus\ell(u_1)} \left( H_{w^{u_1} (v)}  + \sum_{j=1}^{k-1}     \big(H_{w^{u_1}(v)+j} - H_{w^{u_1}(v)+j-1}\big)\right)\\
            &= \sum_{v \in P(u_1)\setminus\ell(u_1)} \left( H_{w^{u_1} (v)}  + H_{w^{u_1}(v)+k-1} - H_{w^{u_1}(v)} \right) \\
            &= \sum_{v \in P(u_1)\setminus\ell(u_1)} H_{w^{u_1}(v)+k-1} = \sum_{v \in P(u_1)\setminus\ell(u_1)} H_{w^u(v)} 
        \end{align*}
    \end{enumerate}
    Where the last equality above follows from Observation~\ref{lem:child}.            
\end{proof}

\subsection{ Proof of Lemma~\ref{lem:case1}}
\label{sec:case1}
{\color{black}
Before proving Lemma~\ref{lem:case1}, we will need the two following useful lemmas.

\begin{lemma}\label{lem:inequality2}
    Let $d\in\mathbb{Z}_{>0}$; then the following inequalities hold:
    \begin{enumerate}
        \item \(\frac{2}{d+1} + \frac{2}{d+2} - \frac{2}{d} \le \frac{2}{4} + \frac{2}{5} - \frac{2}{3} = \frac{7}{30}\)
        \item \(\frac{2}{d+1} + \frac{2}{d+2} + \frac{2}{d+3} - \frac{3}{d}\le  \frac{2}{5} + \frac{2}{6} + \frac{2}{7} - \frac{3}{4} = \frac{113}{420}\) 
    \end{enumerate}
\end{lemma}
\begin{proof}
    \begin{enumerate}
        \item Let \(f(d) \coloneqq \frac{2}{d+1}+\frac{2}{d+2}-\frac{2}{d}\). Then 
        \begin{align*}
            f(d+1)-f(d)=\frac{2}{d+3}-\frac{4}{d+1}+\frac{2}{d}
        \end{align*}
        One can easily compute that for \(d=1,2\), we have \(f(d+1) - f(d) > 0\), and for \(d \geq 3\), we have \(f(d+1) - f(d) \le 0\). Therefore \(f(d) \le f(3) = \frac{7}{30}\).
        
        \item Let \(g(d) \coloneqq \frac{2}{d+1}+\frac{2}{d+2}+\frac{2}{d+3}-\frac{3}{d}\), then we have
        \begin{align*}
            g(d+1)-g(d)=\frac{2}{d+4} + \frac{3}{d} - \frac{5}{d+1}
        \end{align*}
        One can easily compute that for \(d<4\), we have \(g(d+1)-g(d) > 0\), and for \(d\geq 4\) we have \(g(d+1)-g(d) \le 0\). Therefore \(g(d)\le g(4) = \frac{113}{420}\). 
    \end{enumerate}
\end{proof}
\begin{lemma}\label{lem:inequality4}
    \(d, k\in \mathbb{Z}_{> 0}\). Then $\frac{2}{d+k}-\frac{1}{d}<\frac{1}{5k}$.
\end{lemma}
\begin{proof}
    One has:
    \begin{align*}
        &\frac{2}{d+k}-\frac{1}{d}=\frac{d-k}{d(d+k)}
    \end{align*}
    To complete the proof it suffice to show that $5k(d-k)<d(d+k)$. Observe that: 
    \begin{align*}
        &d(d+k)-5k(d-k)=d^2+(2k)^2-4dk+k^2=(d-2k)^2+k^2>0
    \end{align*}
\end{proof}

\begin{lemma}
\label{lem:case1}
    Let \(k_1\le k \in \mathbb{Z}_{>0}\) and \(d \ge 2\). Let \(\delta = \frac{97}{420}\), and \(\phi 
    = 1.86 - \frac{1}{2100}\). Let \(\beta(k)\) be equal to \(0\) for \(k=0,\dots,8\), and \(\frac{1}{3}-\delta\) for \(k\ge 9\).
    Then the following inequality holds:
    \begin{align*}
        -(k-1)\delta + \sum_{j=1}^{k-1}\left(\frac{2}{d+j} - \frac{1}{d}\right) + H_{k+1} \le \phi - \beta(k)
    \end{align*}
\end{lemma}
\begin{proof}
    We can define the terms of the desired inequality to be equal to
    \[
        f(k) = \sum_{i=1}^{k-1} \left( \frac{2}{d+i} - \frac{1}{d} \right) + H_{k+1} - (k-1)\delta -\phi +\beta(k).
    \] Thus, if we show that \(f(k) \le 0\) we have proven the claim. Observe that 
    \[
        f(k+1) - f(k) = \frac{2}{d+k} - \frac{1}{d} + \frac{1}{k+2} - \delta + \beta(k+1)-\beta(k).
    \]
    Observe that using Lemma~\ref{lem:inequality4} \(f(k+1) - f(k)< \frac{1}{5k} + \frac{1}{k+2} - \delta+\beta(k+1)-\beta(k)\). Therefore for $k\ge4$, $f(k+1)-f(k)<0$.

    Furthermore observe that if $k\in \{1,2\}$, then \(f(k+1) - f(k)=\frac{2}{d+k} - \frac{1}{d} + \frac{1}{k+2} - \delta + \ge \frac{1}{k+2}-\delta>0 \)
    
    Therefore, it suffices to prove $f(k)\le 0$ only for $k\in \{3,4\}$.
    
    \begin{itemize}
        \item if \(k=3\), then by Lemma~\ref{lem:inequality2} we have \(f(3)\) is equal to
        \begin{align*}
            \frac{2}{d+1} + \frac{2}{d+2} - \frac{2}{d} + H_4 - 2\delta \le \frac{7}{30} + H_4 - 2\delta <1.855  < \phi
        \end{align*}
        
        \item if \(k=4\), then by Lemma~\ref{lem:inequality2} we have \(f(4)\) is equal to 
        \begin{align*}
            \frac{2}{d+1} + \frac{2}{d+2} + \frac{2}{d+3} - \frac{3}{d} + H_5 - 3\delta \leq \frac{113}{420} + H_5 - 3\delta = \phi
        \end{align*} 
    \end{itemize}
\end{proof}
}

\subsection{ Remaining Cases for Proof of Lemma~\ref{lem:invariant}}
\label{sec:othercases}
This section includes the remaining cases for the proof of Lemma~\ref{lem:invariant}. {\color{black}First, we will need the following lemma, which will be helpful in proving important inequalities for the remaining cases.
\begin{lemma}
\label{lem:case2}
    Let \(k_1\le k \in \mathbb{Z}_{>0}\). Let \(\delta = \frac{97}{420}\), and \(\phi 
    = 1.86 - \frac{1}{2100}\). Let \(\beta(k)\) be equal to \(0\) for \(k=0,\dots,8\), and \(\frac{1}{3}-\delta\) for \(k\ge 9\).
    Then the following inequalities hold:
    \begin{enumerate}[label=(\alph*)]
        \item \((k-1)H_2 + 2H_{k+1} \le (k+1)\phi - \beta(k) - \delta\); 
        \item  \( - (k-k_1-1)\delta + k_1H_2+H_{k+1}+\sum_{j=1}^{k-1} \left(\frac{2}{8+j} - \frac{1}{8}\right) + \frac{k_1}{8}  < (k_1+1)\phi-\beta(k)\);
        \item \(- (k-1)\delta +  H_{k+1} + k_1\phi + \sum_{j=1}^{k-1} \frac{1}{16+j}  < (k_1+1)\phi -\beta(k)\).
    \end{enumerate}
\end{lemma}
\begin{proof}
    \begin{enumerate}[label=(\alph*)]
        \item We reorganize the terms of the inequality so all are on the left side, and define \(f(k) \coloneqq (k-1)H_2 + 2H_{k+1} - (k + 1)\phi + \beta(k) + \delta \). We will show that \(f(k) \le 0\). First, note that \(f(k+1) - f(k) = (H_2 - \phi) + \frac{2}{k+2}+\beta(k+1)-\beta(k)\) from which it is clear that \(f(k+1) - f(k) >0\) if and only if \(k<4\). Therefore \(f(k) \le f(4)=0\), and the claim is proven. 

        \item Observe that 
        \begin{align*}
            &k_1H_2+H_{k+1}+\sum_{i=1}^{k-1} \left( \frac{2}{8+i} - \frac{1}{8} \right) + \frac{k_1}{8}- (k-k_1-1)\delta - (k_1+1)\phi + \beta(k)\\
            =&k_1(H_2 + \frac{1}{8} + \delta - \phi) + H_{k+1} + \sum_{i=1}^{k-1} \left(\frac{2}{8+i} - \frac{1}{8}\right) - (k-1)\delta - \phi + \beta(k)
        \end{align*}
        Since \(H_2 + \frac{1}{8} + \delta < \phi\) we have:
        \begin{align*}
            <& H_{k+1} + \sum_{i=1}^{k-1} \left(\frac{2}{8+i} - \frac{1}{8} \right) - (k-1)\delta - \phi+\beta(k)
        \end{align*}
        We show that \( \sum_{i=1}^{k-1} \left(\frac{2}{8+i} - \frac{1}{8} \right) + H_{k+1}  - (k-1)\delta +\beta(k) -\phi<0\). 
        
        Now let \(f(k) \coloneqq  \sum_{i=1}^{k-1} \left(\frac{2}{8+i} - \frac{1}{8} \right) + H_{k+1} - (k-1)\delta - \phi +\beta(k)\), and consider \( f(k+1) - f(k) = \frac{1}{k+2} + \frac{2}{8+k} -\frac{1}{8}- \delta + \beta(k+1) - \beta(k)\). Observe that \( f(k+1) - f(k)\) is positive if and only if  \(k < 4\). Therefore 
        \[
            f(k) \le  f(4) = -\frac{1109}{27720} < 0
        \] 
        \item 
        We can see
        \begin{align*}
            & - (k-1)\delta +  H_{k+1} + k_1\phi + \sum_{j=1}^{k-1} \frac{1}{16+j} - (k_1+1)\phi + \beta(k)\\
            & = - (k-1)\delta +H_{k+1} + \sum_{j=1}^{k-1}\frac{1}{16+j} - \phi+\beta(k)
        \end{align*}
        
        We show that \(f(k) = - (k-1)\delta +H_{k+1} + \sum_{j=1}^{k-1} \frac{1}{16+j} - \phi + \beta(k) < 0\). Note that \(f(k+1) - f(k) = \frac{1}{k+2} + \frac{1}{16+k} - \delta +\beta(k+1)-\beta(k)\) is negative if and only if $k\ge4$ and $k\neq8$. Therefore $f(k)$ is upper-bounded by
        \begin{align*}
             \max\{&f(4),f(9)\} \\
             &= \max\{H_5 + \sum_{j=1}^{3} \frac{1}{16+j}      - 3\delta - \phi, 
             H_{10} + \sum_{j=1}^{8} \frac{1}{16+j} - 8\delta - \phi+\beta(9)\}\\
            &\approx -0.102036 <0.
        \end{align*} 
    \end{enumerate}
\end{proof}
}

\paragraph{ Case(ii): \(u\) has a final child.} We note the following.
\begin{align*}
    h_{W^u}(Q_u) 
    =\sum^k_{j=1}h_{W^{u}}(Q_{u_j}) + H_{w^u(u)} 
\end{align*}

Let \(\ell \coloneqq \ell(u_m)\). 
By Lemma~\ref{lem:increase}.(a) we can see \(H_{w^u(u)} = H_k\),  and \(H_{w^{u_j}(u_j)} = H_2\) for \(j=1, \dots, k_1\).  By Lemma~\ref{lem:increase}.(b), we can see \(h_{W^u}(Q_{u_j}) = h_{W^{u_j}}(Q_{u_j})\) for \(j \neq m\), and by Lemma~\ref{lem:increase}.(c) and (d) we can see \(h_{W^{u}}(Q_{u_m}) = h_{W^{u_m}}(Q_{u_m})  + \sum_{j=1}^{k-1}C_j^{u_m} + H_{w^u(\ell)} - H_{w^{u_m}(\ell)}\). Therefore:
\begin{align*}
    h_{W^u}(Q_u) 
    = \sum_{j>k_1} h_{W^{u_j}} (Q_{u_j}) + k_1H_2 + \sum_{j=1}^{k-1}C_j^{u_m} + H_k + H_{w^u(\ell)} - H_{w^{u_m}(\ell)}
\end{align*}
By Algorithm~\ref{alg:computing_w} we mark a final child \(u_m\) of \(u\) depending on the value of \(\min_{j\in\{k_1+1,\dots, k\}} C_1^{u_j}\). We consider these cases.

\paragraph{ Case (ii).(a)} If \(k_1=k\) or if \(\min_{j\in\{k_1+1,\dots, k\}} C_1^{u_j} \geq \phi - \delta - H_2\), final node \(u_1 = u_m\) is the marked child of \(u\) according to Algorithm~\ref{alg:computing_w}. 

Since \(u_1\) is final, \(C_j^{u_1} = 0\)
and, \(h_{W^{u_1}}(Q_{u_1}) = H_2\). Finally, applying Observation~\ref{lem:child} to \(Q_{u_1}\), we see \(h_{W^u}(Q_{u_1}) = H_{k+1}\). Therefore:
\begin{align*}
     h_{W^u}(Q_u)=& \sum_{j>k_1} h_{W^{u_j}} (Q_{u_j}) + (k_1-1) H_2 + H_k + H_{k+1}
\end{align*}
We apply our inductive hypothesis on \(Q_{u_{k_1+1}}, \dots, Q_{u_k}\), and use \(\beta(j) \ge 0\) for all \(j\ge 0\):
\begin{align*}
    h_{W^u}(Q_u)&\leq \sum_{j > k_1}\left( \phi|Q_{u_j}| - C^{u_j}_1 - \delta \right) +  (k_1 - 1)H_2  + H_k + H_{k+1} \\
    =& \phi\left( |Q_u| - k_1 - 1 \right)   - \sum_{j>k_1} C^{u_j}_1 - (k - k_1)\delta  + (k_1 -1)H_2 + H_k + H_{k+1} 
\end{align*}
Applying the assumption that \(\min_{j\in\{k_1+1,\dots, k\}} C_1^{u_j} \geq \phi - \delta - H_2\):
\begin{align*}
    &\le  \phi\sum_{j>k_1} |Q_{u_j}| - (k-k_1)(\phi - \delta - H_2)  - (k - k_1)\delta + (k_1 -1)H_2 + H_k + H_{k+1} \\
    & = \phi(|Q_u| -k-1)  + (k-1)H_2 + 2H_{k+1} - \frac{1}{k+1}
\end{align*}

Using Lemma~\ref{lem:case2}.(a), and the fact that \(C_1^u = \frac{1}{k+1}\), we have
\begin{align*}
    \leq &  \phi\left( |Q_u| - k - 1 \right) - \delta + (k+1)\phi -\beta(k) - \frac{1}{k+1} = \phi|Q_u| - C_1^u - \delta -\beta(k).
\end{align*}

\paragraph{ Case (ii).(b)}  In this case we assume \(\min_{j\in\{k_1+1, \dots, k\}} C_1^{u_j} <  \phi-H_2-\delta\) and, by  Algorithm~\ref{alg:computing_w}, we mark some child \(u_m\) for \(k_1+1 \le m \le k\). Without loss of generality we will assume that \(m=k\). We let \(d_x\) denote the degree of a non-final node \(x\) in \(T_i\). Let \(d_{u_k} \coloneqq d\).

We now consider by cases if the marked child of \(u_k\), denoted \(v\), is a final node.
\paragraph{ Case (ii).(b).i: \(v\) is a final node.}
Since \(v\) is final, we have \(\ell = v\).
By Lemma~\ref{lem:increase}.(a), \(H_{w^{u_k}(u_k)} = H_{d-1}\). By Observation~\ref{lem:child} we have \(H_{w^u(u_k)} = H_{k+d-2}\).
Since \(v\) is a final node we know \(C_j^{u_k} = \frac{1}{w^{u_k}(u_k) +j} = \frac{1}{d+j-1}\). Therefore,
\begin{align*}
    h_{W^u}(Q_u) =  \sum_{j>k_1} h_{W^{u_j}}(Q_{u_j}) + k_1H_2  + \sum_{j=1}^{k-1} \frac{1}{d+j-1} + H_k + H_{w^u(v)} - H_{w^{u_k}(v)}
\end{align*}
Since \(v\) is final, \(H_{w^v(v)} = H_2\). By Observation~\ref{lem:child} we have, {\color{black}\(H_{w^{u_k}(v)} = H_{d}\)}, and \(H_{w^u(v)} = H_{d+ k - 1}\).
 Therefore, 
\begin{align*}
    &h_{W^u}(Q_u) = \sum_{j>k_1} h_{W^{u_j}}(Q_{u_j}) + k_1H_2 + H_k + \sum_{j=1}^{k-1}  \frac{1}{d+j-1} + \sum_{j=1}^{k-1}\frac{1}{d+j} \\
    & = \sum_{j>k_1} h_{W^{u_j}}(Q_{u_j}) + k_1H_2 + H_k + \sum_{j=1}^{k-1} \frac{2}{d+j} - \frac{1}{d+k-1} +\frac{1}{d}
\end{align*}
Observe that since \(C^{u_k}_1 = \frac{1}{d} < \phi - H_2 - \delta < 0.1286 < \frac{1}{7}\), we have \(d \geq 8\). 
\begin{align*}
    h_{W^u}(Q_u) \leq \sum_{j>k_1} h_{W^{u_j}}(Q_{u_j}) + k_1H_2 + H_k  + \sum_{j=1}^{k-1} \frac{2}{8+j} - \frac{1}{d+k-1} +\frac{1}{8}
\end{align*}
We apply our inductive hypothesis on \(Q_{u_{k_1}+1}, \dots, Q_{u_k}\), and use \(\beta(j) \ge 0\) for all \(j \ge 0\).
\begin{align*}
    \leq \sum_{j>k_1}(\phi |Q_{u_j}| - \delta - C^{u_j}_1 ) + k_1H_2 + H_k + \sum_{j=1}^{k-1} \frac{2}{8+j} - \frac{1}{d+k-1} +\frac{1}{8}
\end{align*}
We assumed that \(C_1^{u_k} = \min_{j\in \{k_1+1, \dots, k\} } C_1^{u_j}\). Since the marked child of \(u_k\) is a final node we know \(C_1^{u_k} = \frac{1}{w^{u_k}(u_k) +1} = \frac{1}{d}\). Therefore:
\begin{align*}
    \le& \sum_{j>k_1}(\phi|Q_{u_j}|-\delta -\frac{1}{8}) + k_1H_2 + H_k  + \sum_{j=1}^{k-1} \frac{2}{8+j} - \frac{1}{d+k-1} +\frac{1}{8}\\
    =& \sum_{j>k_1}\phi|Q_{u_j}| - (k-k_1)\delta  + k_1H_2 + H_k + \sum_{j=1}^{k-1}\left(\frac{2}{8+j} - \frac{1}{8} \right) - \frac{1}{d+k-1} +\frac{k_1}{8}
\end{align*}
Using Lemma~\ref{lem:case2}.(b), we see
\begin{align*}
    &<\sum_{j>k_1}\phi|Q_{u_j}|-\frac{1}{k+1} - \frac{1}{d+k-1}+(k_1+1)\phi -\delta -\beta(k) \\
    &=\phi|Q_u|  -\frac{1}{k+1} - \frac{1}{d+k-1} -\delta - \beta(k)\\ 
    &= \phi|Q_u| - \frac{1}{w^u(u)+1} - \frac{1}{w^u(u_k)+1} -\delta- \beta(k) = \phi|Q_u| - C_1^{u} - \delta  - \beta(k)
\end{align*}
Where the second equality follows from  Lemma~\ref{lem:increase}.(a) and Observation~\ref{lem:child}. And the claim is proven.

\paragraph{ Case: (ii).(b).ii: v is not a final node.}
{\color{black}
In order to complete the proof in this case  we make use of the following lemma.
\begin{lemma}\label{lem:inequality3},
    Let \(2\le x, y \in \mathbb{Z}_{> 0}\).  Let \(\delta = \frac{97}{420}\), and \(\phi = 1.86 - \frac{1}{2100}\). 
    
    If \(\frac{1}{x} + \frac{1}{x+y-2} < \phi - H_2 - \delta\). Then \(x + y \ge 18\). 
\end{lemma}
\begin{proof}
    Assume that \(x+y < 18\). Since \(x,y \ge 2\), then 
    \[
        \frac{1}{x}+\frac{1}{x+y-2}\ge  \frac{2}{x+y-2}\ge \frac{2}{15} = 0.1\bar{3}>0.1286>  \phi - H_2 - \delta 
    \]
    which is a contradiction.  
\end{proof}
}

Since \(v\) is not final, we know \(C_j^{u_k} \ge \frac{1}{w^{u_k}(u_k) + j} + \frac{1}{w^{u_k}(v) + j}\).  Therefore, \(\phi - H_2 - \delta > C_1^{u_k} \ge \frac{1}{d_{u_k}} + \frac{1}{d_{u_k} + d_{v} - 2}\). 
Applying Lemma~\ref{lem:inequality3} we see  \( d_{u_k} + d_{v} \ge 18\), and by Observation~\ref{lem:child} we see \(w^{u_k}(\ell) \ge d_{u_k} + d_{v} - 2 \ge 16\).
Therefore, \(H_{w^u(\ell)} - H_{w^{u_k}(\ell)} = \sum_{j=1}^{k-1} \frac{1}{w^{u_k}(\ell)+j} \le \sum_{j=1}^{k-1} \frac{1}{16+j}\).
\begin{align*}
    h_{W^u}(Q_u) \le  \sum_{j>k_1} h_{W^{u_j}}(Q_{u_j}) + k_1H_2  + \sum_{j=1}^{k-1} C^{u_k}_j + H_k + \sum_{j=1}^{k-1}\frac{1}{16+j}
\end{align*}

We apply the inductive hypothesis to \(Q_{u_{k_1+1}}, \dots, Q_{u_k}\), and that \(\beta(j) \ge 0\) for all \(j \ge 0\):
\begin{align*}
    \le \sum_{j>k_1}\left(\phi|Q_{u_j}| - C^{u_j}_1 - \delta\right) + k_1H_2 + \sum_{j=1}^{k-1} C^{u_k}_j + H_k 
    + \sum_{j=1}^{k-1}\frac{1}{16+j}
\end{align*}
We apply the assumption \(C_{1}^{u_k} = \min_{j\in\{k_1+1, \dots,k\}} C_1^{u_j}\):
\begin{align*}
    \le &\sum_{j>k_1}\left(\phi|Q_{u_j}| - \delta\right) - (k-k_1)C^{u_k}_1+ k_1H_2 + \sum_{j=1}^{k-1} C^{u_k}_j + H_k 
    + \sum_{j=1}^{k-1}\frac{1}{16+j}\\ 
    =&\sum_{j>k_1}(\phi|Q_{u_j}| - \delta) + C^{u_k}_1(k_1 - 1)   + k_1H_2  + \sum_{j=1}^{k-1} \left( C^{u_k}_j - C^{u_k}_1 \right) + H_k + \sum_{j=1}^{k-1}\frac{1}{16+j} 
\end{align*}
Using \(C_1^{u_k} < \phi -H_2 - \delta\):
\begin{align*}
    <& \sum_{j>k_1}(\phi|Q_{u_j}| - \delta)  + k_1(\phi - H_2 - \delta)  - C^{u_k}_1 + k_1H_2 \\
    &+ \sum_{j=1}^{k-1} \left( C^{u_k}_j -  C^{u_k}_1 \right) + H_k + \sum_{j=1}^{k-1}\frac{1}{16+j} \\
    =& \sum_{j>k_1}\phi|Q_{u_j}| +k_1\phi - k\delta - C_1^{u_k} + \sum_{j=1}^{k-1} \left( C^{u_k}_j -  C^{u_k}_1 \right) + H_k + \sum_{j=1}^{k-1} \frac{1}{16+j}
\end{align*}
Therefore, we can apply Lemma~\ref{lem:case2}.(c) to see the following
\begin{align*}
    & \le \phi(k_1 + 1 +  \sum_{j>k_1}|Q_{u_j}|) -\delta  - C^{u_k}_1 + \sum_{j=1}^{k-1}\left(C_{j}^{u_k} - C_{1}^{u_k}\right)  - \frac{1}{k+1} -\beta(k)
    \\&= \phi\vert Q_u \vert-\delta- C^{u_k}_1 + \sum_{j=1}^{k-1}\left(C_{j}^{u_k} - C_{1}^{u_k}\right)  - \frac{1}{k+1} -\beta(k)
\end{align*}
Where the equality above follows since $\sum_{j>k_1}|Q_{u_j}|=|Q_u|-k_1-1$.
We can apply \(C_j^{u_k} \le C_1^{u_k}\) to see the claim 
\begin{align*}
    &\leq \phi\vert Q_u \vert -\delta - C^{u_k}_k  + \sum_{j=1}^{k}\left(C_{j}^{u_k} - C_{1}^{u_k}\right) - \frac{1}{k+1} -\beta(k)\\
    &\leq \phi\vert Q_u \vert -\delta  - C^{u_k}_k - \frac{1}{k+1} -\beta(k) =  \phi\vert Q_u \vert -\delta - C_1^{u} -\beta(k).
\end{align*}

\section{   Proof of Lemma~\ref{lem:CentersOflengthOne}}
\label{apx:nodelower}
{\color{black}In this section we discuss the proof of Lemma~\ref{lem:CentersOflengthOne} to complete the arguments of Section~\ref{sec:nodelower}.
We will need the following useful lemma.
\begin{lemma}\label{claim:useful}
    Let $x\ge 7 $ be a positive integer. Then  $H_x+H_{2x+3}+H_{2x+2}-H_{2x}-H_{2x-1}>H_{10}$.
\end{lemma} 
\begin{proof}
    We have:
    \begin{align*}
        &H_x + H_{2x+3} - H_{2x} + H_{2x+2} - H_{2x-1}\\
        =&H_x + \left( \frac{1}{2x+3}  +\frac{1}{2x+2} \right) + \left( \frac{1}{2x+1} + \frac{1}{2x+2} \right) + \left( \frac{1}{2x+1} + \frac{1}{2x} \right)\\
        >&H_x+\frac{1}{x+3}+\frac{1}{x+2}+\frac{1}{x+1}=H_{x+3}\ge H_{10}.
    \end{align*}
    And the claim is proven.
\end{proof}
}

\label{sec:CentersOflengthOne}
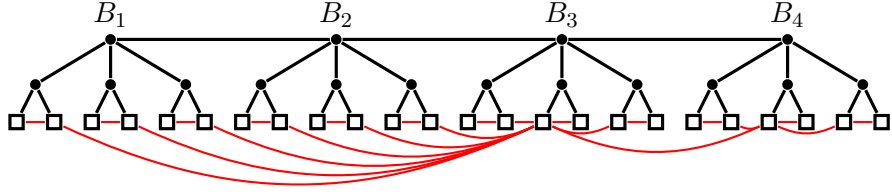
\begin{figure}[t]
    \begin{center}
        \begin{tikzpicture}[scale=0.5]
        
            \tikzset{terminal/.style={draw=black, very thick,minimum size=0pt, inner sep=2.5pt, outer sep=1pt}}
            
            \begin{scope}[every node/.style={terminal}]
                \node (t111) at (0, 0){};
                \node (t112) at (1, 0){};
                
                \node (t121) at (2, 0){};
                \node (t122) at (3, 0){};
                
                \node (t131) at (4, 0){};
                \node (t132) at (5, 0){};

                \node (t211) at (6, 0){};
                \node (t212) at (7, 0){};
                
                \node (t221) at (8, 0){};
                \node (t222) at (9, 0){};
                
                \node (t231) at (10, 0){};
                \node (t232) at (11, 0){};

                \node (t311) at (12, 0){};
                \node (t312) at (13, 0){};

                \node (t321) at (14, 0){};
                \node (t322) at (15, 0){};
        
                \node (t331) at (16, 0){};
                \node (t332) at (17, 0){};

                \node (t411) at (18, 0){};
                \node (t412) at (19, 0){};
                
                \node (t421) at (20, 0){};
                \node (t422) at (21, 0){};
                
                \node (t431) at (22, 0){};
                \node (t432) at (23, 0){};
            \end{scope}
            
            \tikzset{vertex/.style={draw=black, very thick, circle,minimum size=0pt, inner sep=1pt, outer sep=1pt,fill=black}}
            \begin{scope}[every node/.style={vertex}]
                \node (r11) at (0.5, 1){};
                \node (r12) at (2.5, 1){};
                \node (r13) at (4.5, 1){};
                
                \node (r21) at (6.5, 1){};
                \node (r22) at (8.5, 1){};
                \node (r23) at (10.5,1){};
                
                \node (r31) at (12.5,1){};
                \node (r32) at (14.5,1){};
                \node (r33) at (16.5,1){};
                
                \node (r41) at (18.5,1){};
                \node (r42) at (20.5,1){};
                \node (r43) at (22.5,1){};
            \end{scope}

            
            \begin{scope}[every node/.style={vertex}]
                
                \node (s1) at (2.5, 2.2){};
                \node (s2) at (8.5, 2.2){};
                \node (s3) at (14.5,2.2){};
                \node (s4) at (20.5,2.2){};

                \draw[very thick] (s1) to (s2) to (s3) to (s4);

                \draw[very thick] (s1) to (r13);
                \draw[very thick] (r12) to (s1) to (r11);
                \draw[very thick] (t111) to (r11) to (t112);
                \draw[very thick] (t121) to (r12) to (t122);
                \draw[very thick] (t131) to (r13) to (t132);
                
                \draw[very thick] (r22) to (s2) to (r21);
                \draw[very thick] (s2) to (r23);
                \draw[very thick] (t211) to (r21) to (t212);
                \draw[very thick] (t221) to (r22) to (t222);
                \draw[very thick] (t231) to (r23) to (t232);

                \draw[very thick] (r32) to (s3) to (r31);
                \draw[very thick] (s3) to (r33);
                \draw[very thick] (t311) to (r31) to (t312);
                \draw[very thick] (t321) to (r32) to (t322);
                \draw[very thick] (t331) to (r33) to (t332);                
                
                \draw[very thick] (r42) to (s4) to (r41);
                \draw[very thick] (s4) to (r43);
                \draw[very thick] (t411) to (r41) to (t412);
                \draw[very thick] (t421) to (r42) to (t422);
                \draw[very thick] (t431) to (r43) to (t432);

                \draw [red, thick] (t111) to (t112);
                \draw [red, thick] (t121) to (t122);
                \draw [red, thick] (t131) to (t132);

                \draw [red, thick] (t211) to (t212);
                \draw [red, thick] (t221) to (t222);
                \draw [red, thick] (t231) to (t232);

                \draw [red, thick] (t311) to (t312);
                \draw [red, thick] (t321) to (t322);
                \draw [red, thick] (t331) to (t332);

                \draw [red, thick] (t411) to (t412);
                \draw [red, thick] (t421) to (t422);
                \draw [red, thick] (t431) to (t432);
                
                \draw [red, thick, bend right=25] (t112) to (t321);
                \draw [red, thick, bend right=25] (t122) to (t321);
                \draw [red, thick, bend right=25] (t132) to (t321);
                \draw [red, thick, bend right=25] (t212) to (t321);
                \draw [red, thick, bend right=25] (t222) to (t321);
                \draw [red, thick, bend right=25] (t232) to (t321);
                \draw [red, thick, bend right=25] (t412) to (t421);
                \draw [red, thick, bend right=25] (t421) to (t431);

                \draw [red, thick, bend right=0] (t312) to (t321);
                \draw [red, thick, bend right=25] (t321) to (t331);
                
                \draw[red, thick, bend right=25] (t321) to (t421);
            \end{scope}

            \node[above=1pt] () at (s1) {$B_1$};
            \node[above=1pt] () at (s2) {$B_2$};
            \node[above=1pt] () at (s3) {$B_3$};
            \node[above=1pt] () at (s4) {$B_4$};
        \end{tikzpicture}
    \end{center}
    \caption{
    Red edges form a laminar witness tree $W$ that is not optimal. In this case we have centers \(B_3\) and \(B_4\). Where \(B_3\) has  \(x_L^3=2\), \(x_R^3=0\), \(L_3 = 0\), and  \(R_3=1\), and \(B_4\) has  \(x_L^4=0\), \(x_R^4=0\), \(L_4=1\), and \(R_4=0\).}
    \label{fig:centerInstance}
\end{figure}
\begin{proof}[Proof of Lemma~\ref{lem:CentersOflengthOne}]
    We introduce some notation. For center \(B_i\) in \(\W\), we denote by \(x_L^i\) (resp. \(x_R^i\)) the number of subtrees, \(B_k\) for \(k<i\) (resp. \(k>i\)), such that the subgraph on \(W\) induced by \(r_{k1}^1, r_{k2}^1, r_{k3}^1\) is three singletons and the unique terminal adjacent to these is in \(B_i\). 
    Furthermore, we let \(L_i = 1\) (resp. \(R_i = 1\))  if  there is a center \(B_j\) in \(\W\) with \(j < i\) (resp. \( j> i\)) and equal to \(0\) if not. 
    
    Therefore, any center \(B_i\) in \(\W\) has exactly \(3x^i_L\) edges to subtrees \(B_{i-j}\) for \(j=1, \dots, x^i_L\), and exactly \(3x^i_R\) edges to subtrees \(B_{i+j}\) for \( j=1,\dots, x^i_R\), plus a single edge to a center with index less than \(i\) if \(L_i =1\), and a single edge to a center with index greater than \(i\) if \(R_i = 1\). So, by the laminarity of \(\W\), we can see that there are exactly \(3x^i_L + 3x^i_R + L_i  + R_i\) edges incident to \(B_i\) in \(\W\). 

    Let \(w\) be the vector imposed on the nodes of \(T\) by \(\W\). 
    Observe, for every \(1\le k \le x_R^i\), \(w(s_{i+k}) = 3(x_R^i - k + 1) + R_i\) and \(w(t_{(i+k)j}) = 2\) for \(j\in \{1,2,3\}\). Similarly, for every \(1\le k \le x_L^i\), \(w(s_{i-k}) = 3(x_L^i - k + 1) + R_i\) and \(w(t_{(i-k)j}) = 2\) for \(j\in \{1,2,3\}\). Finally, let \(r_{i2}^1\) be the unique terminal that these subtrees are adjacent to, then \(w(s_i) = 3x^i_L + 3x^i_R + L_i  + R_i+2\), \(w(t_{i2}) = 3x^i_L + 3x^i_R + L_i  + R_i + 3\), and \(w(t_{i1}) = w(t_{i3}) = 2\). (see Figure~\ref{fig:centerInstance} for an example)

    Consider a center $B_i$ in \(\W\) and let be $x^i_L$, $x^i_R$, $R_i$ and $L_i$ defined as above. 
    We will show that \(x^i_L + x^i_R = 0\). That is, we will show that for every \(i\in [q]\), \(B_i\) must be  a center. Assume that \(B_i\) is a center with \(x_L^i + x_R^i \ge 1\). 
    We can see that 
    \begin{align*}
        &\sum_{j = i-x_L^i}^{i + x_R^i} \sum_{v\in B_j} H_{w(v)} = \sum_{j=i-x_L^i}^{i-1} \sum_{v\in B_j} H_{w(v)} + \sum_{j = i+1}^{i+x_R^i}\sum_{v\in B_j} H_{w(v)} + \sum_{v\in B_i} H_{w(v)}\\
        =& \sum_{j=1}^{x_L^i}(3H_2 + H_{3j + L_i}) + \sum_{j=1}^{x_R^i}(3H_2 + H_{3j + R_i}) + 2H_2 + H_{3x^i_L + 3x^i_R +R_i +L_i + 2} \\ & + H_{3x^i_L + 3x^i_R +R_i +L_i + 3}
    \end{align*}

    Consider laminar witness tree \(W'\) that is equal to \(\W\) except for edges with endpoints in \(B_{i+j}\), for \(j=-1,\dots, -x^i_l\), and \(j=1,\dots,x^i_r\). We instead let these \(B_{i+j}\) be centers in \(W'\), with \(x_L^{i+j}=x_R^{i+j}=0\). Clearly, \(L_{i+j}=R_{i+j}=1\) for \(j\neq -x_L^i, x_R^i\), and it is clear \(L_{i-x_L^i} = L_i\), and \(R_{i+x_R^i} = R_i\). Let \(w'\) be the vector imposed on the nodes of \(T\) by \(W'\). Clearly the difference between \(\sum_{v\in T} H_{w(v)}\) and \(\sum_{v\in T} H_{w'(v)}\) is \(\sum_{j= i - x_L^i}^{i+x_R^i} \sum_{v\in B_j} H_{w(v)} - \sum_{v\in B_j}H_{w'(v)}\) which is equal to
    \begin{align*}
        &\sum_{j=1}^{x_L^i}(3H_2 + H_{3j + L_i}) + \sum_{j=1}^{x_R^i}(3H_2 + H_{3j + R_i}) + 2H_2+ H_{3x^i_L + 3x^i_R +R_i +L_i + 2} \\ &  + H_{3x^i_L + 3x^i_R +R_i +L_i + 3} - \bigg( (2x^i_L + 2x^i_R+2)H_2 +  (x^i_L + x^i_R -1)(H_4 + H_5)
        \\& + H_{3+L_i} + H_{4+L_i} + H_{3+R_i} + H_{4+R_i}  \bigg) 
    \end{align*}
    We let \(P(x_R^i, x_L^i, L_i, R_i)\) denote this difference. We will show that \(P(x_R^i, x_L^i, L_i, R_i) >0 \), for every \((x^i_L,x^i_R,L_i,R_i) \in \mathbb{Z}^4 \) such that \(x^i_L,x^i_R \ge 0\), \(x^i_L + x^i_R \ge 1\) and \(L_i, R_i\in \{0,1\}\), contradicting the assumption that \(\nu_T(\W) = \min_{W\in \mathcal{W}} \nu_T (W)\). We proceed by induction on \(x^i_R + x^i_L\).

    For our base case, we assume \(x_R^i = 1 \ge x_L^i\). We have the following cases for the values of \(x_L^i\):
    \begin{enumerate}
        \item Case: \(x_L^i = 0\). Then \(P(0,1, L_i, R_i)\) is equal to 
        \begin{align*}
            &5H_2 + H_{3+R_i} + H_{5+ R_i +L_i} + H_{6+ R_i +L_i} - 4H_2 -  H_{3+L_i} - H_{4+L_i} - H_{3+R_i} - H_{4+R_i}\\
            =& H_2 +  H_{5+ R_i +L_i} + H_{6 + R_i +L_i}-  H_{3+L_i} - H_{4+L_i} - H_{4+R_i}\\
            \ge & H_2 + H_{6+R_i} + H_{7+R_i} - H_{4} - H_{5} - H_{4+R_i}\\
            \ge & H_2 + H_6 + H_7 - H_4 - H_5 - H_4 = 13/140 > 0
        \end{align*}
        Where the first inequality follows since it is not hard to see that \(H_{5+R_i} + H_{6+R_i} - H_3 - H_4 > H_{6+R_i} + H_{7+R_i} - H_4 - H_5 > 0 \). The second inequality follows for a similar reason.

        \item Case: \(x_L^i = 1\). Then \(P(1,1, L_i, R_i)\) is equal to 
        \begin{align*}
            &8H_2 + H_{3+L_i} + H_{3+R_i} + H_{8+L_i+R_i} + H_{9+L_i+R_i} - 6H_2 - H_4 - H_5 -H_{3+L_i}\\& - H_{4+L_i} - H_{3+R_i} - H_{4+R_i}\\
            =& 2H_2 + H_{8+L_i+R_i} + H_{9+L_i+R_i} - H_4 - H_5- H_{4+L_i}- H_{4+R_i} \\
            \ge & 2H_2 + H_8 + H_9 - 3H_4 - H_5 = 17/2160 > 0
        \end{align*}
        Where the first inequality above follow easily by checking the values of \(L_i, R_i\in\{0,1\}\).
    \end{enumerate}

    Our inductive hypothesis is to assume the inequality holds for \(x^i_L + x^i_R= k \ge 1\). We will show the claim holds when \(x^i_L + x^i_R = k+1\). Since we showed the base case for \(x_R^i=1\) and \(x_L^i \in \{0,1\}\), we can assume \(\max\{x_R^i, x_L^i\} \ge 2\). Furthermore, we can assume without loss of generality that \(3x_R^i + R_i \ge 3x_L^i + L_i\), which implies \(x_R^i \ge x_L^i\). We will show that \(P(x_L^i, x_R^i, L_i, R_i) > P(x_L^i, x_R^i - 1, L_i, R_i)> 0\), by applying the inductive hypothesis to \(x_L^i + x_R^i - 1 = k\). We can see 
    \begin{align*}
        &P(x_L^i, x_R^i, L_i, R_i) - P(x_L^i, x_R^i - 1, L_i, R_i)\\
        =&H_2 + H_{3x_R^i +R_i} - H_4 - H_5 + H_{3x^i_L + 3x^i_R+L_i+R_i+2} - H_{3x^i_L + 3x^i_R+L_i + R_i-1} \\&+ H_{3x^i_L + 3x^i_R + L_i + R_i +3} - H_{3x^i_L + 3x^i_R + L_i + R_i} \\
        \ge & H_2 + H_{3x_R^i +R_i} - H_4 - H_5 + H_{6x_R^i+2R_i+2} - H_{6x_R^i+2R_i-1} + H_{6x_R^i + 2 R_i +3} - H_{6x_R^i + 2 R_i} 
    \end{align*}
    Where the inequality above follows since we can see that the following inequalities hold since \(3x_R^i + R_i \ge 3x_L^i + L_i\)
    \begin{align*}
        H_{3x^i_L + 3x^i_R+L_i+R_i+2} - H_{3x^i_L + 3x^i_R+L_i + R_i-1} &\ge H_{6x_R^i+2R_i+2} - H_{6x_R^i+2R_i-1}\\
        H_{3x^i_L + 3x^i_R + L_i + R_i +3} - H_{3x^i_L + 3x^i_R + L_i + R_i} &\ge H_{6x_R^i + 2 R_i +3} - H_{6x_R^i + 2 R_i}
    \end{align*}
    
    Similarly, since \(R_i \le 1\), we have 
    \begin{align*}
        & H_2 + H_{3x_R^i +R_i} - H_4 - H_5 + H_{6x_R^i+2R_i+2} - H_{6x_R^i+2R_i-1} + H_{6x_R^i + 2 R_i +3} - H_{6x_R^i + 2 R_i} \\
        \ge & H_2 + H_{3x_R^i +1} - H_4 - H_5 + H_{6x_R^i+4} - H_{6x_R^i+1} + H_{6x_R^i + 5} - H_{6x_R^i + 2} 
    \end{align*}
    Finally, by applying Lemma~\ref{claim:useful} (by setting \(x= 3x_R^i + 1\)) we have the following
    \begin{align*}
        &P(x_L^i, x_R^i, L_i, R_i) - P(x_L^i, x_R^i - 1, L_i, R_i)\\
        \ge & H_2 + H_{3x_R^i +1} - H_4 - H_5 + H_{6x_R^i+4} - H_{6x_R^i+1} + H_{6x_R^i + 5} - H_{6x_R^i + 2} \\
        > & H_2 - H_4- H_5 + H_{10}    = 157/2520 > 0
    \end{align*}
    which completes the proof. 
\end{proof}

    \section{     Proofs for Section~\ref{sec:clawupper}}
\label{apx:clawupper}
The goal of this section is to provide the complete proofs of Section~\ref{sec:clawupper}. In Lemma~\ref{lem:smallValuesofQ}, we show that our approximation factor holds for small values of $q$, and then we provide the proof of Lemma~\ref{lem:upperbound}, giving us the necessary ingredients to prove Theorem~\ref{thm:clawlowerbound}.
\subsection{     Upperbound for small Steiner-Claw Free instances }
\label{sec:smallValuesOfQ}
\begin{lemma}
\label{lem:smallValuesofQ}
    If $q<5$, then \(\gamma_{(G,R,c)} < \frac{991}{732}\).
\end{lemma}
\begin{proof}
    We denote by \(L\subseteq E^*\) the edges of $T$ incident to a terminal, and by \(O = E^*\setminus L\) the edges of the path \(s_1, \dots, s_q\). Let \(\alpha \coloneqq c(O) / c(L)\). Note that \(c(E^*) = (1+\alpha)c(L) = \frac{1+\alpha}{\alpha}c(O)\).
    We distinguish two cases for the values of \(\alpha\):
    \begin{itemize}
        \item First assume $\alpha\ge \frac{1}{2}$. In this case we define $E_W$ as \(\{r_{i}r_{i+1} | 1\le i < q \}\). Observe that $w(e)$ is $1$ if $e=s_is_{i+1}$ for $1\le i < q$ and is at most $2$ if $e=s_ir_{i}$ $1\le i\le q$. Therefore:
        \begin{align*}
        \sum_{e\in E^*}c(e)H_{\bar w(e)}
            &\leq \frac{c(E^*)}{1+\alpha}H_2 + \frac{\alpha c(E^*)}{1+\alpha}H_1 
            =\frac{H_2 + \alpha}{1+ \alpha}c(E^*) \le \frac{4}{3}c(E^*).
        \end{align*}
    
        Therefore ${\bnu_T(W)}\le \frac{4}{3}< \frac{991}{732}$.
        \item Now assume $\alpha< \frac{1}{2}$. We uniformly at random select $1\le\sigma\le q$ and then we define \(E_W\) as \(\{r_\sigma  r_{i} | 1\le i \le q, i\neq \sigma \}\). 
        If $e=s_is_{i+1}$ for $1\le i<q$ then it's not hard to see $\mathbb{E}[H_{w(e)}]\le H_2$ since \(q<5\). 
        For $e=s_ir_{i}$, $\mathbb{E}[H_{w(e)}]=\frac{1}{q} H_{q-1} + \frac{q-1}{q} H_1\le \frac{H_3+3}{4}=29/24.$ Therefore:
        \begin{align*}
            \sum_{e\in E^*}c(e)H_{\bar w(e)}
            &\leq \frac{\frac{29}{24}c(E^*)}{1+\alpha} + \frac{\alpha c(E^*)}{1+\alpha}H_2
            =\frac{\frac{29}{24}+\alpha H_2}{1+\alpha} c(E^*)< \frac{47}{36}c(E^*).
        \end{align*}
        \end{itemize}
        Thus $\mathbb{E}[{\bnu_T(W)}]\le \frac{47}{36} < \frac{991}{732}$, which implies \(\gamma_{(G,R,c)} < \frac{991}{732}\). 
\end{proof}

\subsection{     Proof of Lemma~\ref{lem:upperbound}}
\label{sec:smallt}
\begin{proof}
    We denote
    \begin{align*}
         f(\alpha):=\frac{1}{\alpha + 1} \left(\frac{1}{t_\alpha} H_{t_\alpha+1} + 
         \frac{t_\alpha-1}{t_\alpha} + \alpha \left(\frac{1}{t_\alpha}  + 
         \frac{2}{t_\alpha} \sum_{i=2}^{\lceil \frac{t_\alpha}{2} 
         \rceil}H_{i} \right) \right).
    \end{align*}
    Suppose \(\alpha\in [0, 0.3\Bar{5}]\). Then by definition $t_\alpha=5$ and therefore we have
    \begin{align*}
        f(\alpha)=1.5\Bar{3}-\frac{1.5\Bar{3}- 1.29}{\alpha + 1} \le \frac{991}{732}.
    \end{align*}
    Suppose \(\alpha \in (0.3\Bar{5}, 1)\). In this case $t_\alpha=3$, thus
    \begin{align*}
        f(\alpha)= 1.\Bar{3}+\frac{1.36\Bar{1}-1.\Bar{3}}{\alpha+1} < \frac{991}{732}.
    \end{align*}
    Furthermore for \(\alpha \ge 1\), $t_\alpha=1$; so
    \begin{align*}
        f(\alpha) = 1+\frac{0.5}{\alpha+1}\le1.25.
    \end{align*}
\end{proof}
    \section{ Steiner-claw Free Lower Bound}
\label{sec:clawlower}
The goal of this section is to prove Theorem~\ref{thm:clawlowerbound}. {\color{black} We will need the following useful lemma.
\begin{lemma}
    \label{lem:useful6}
        For \(x\in \mathbb{Z}\) \(x\ge 3\), \(\alpha = 32/90\)
        \[
            \alpha\left( \frac{1}{x+1} + \frac{1}{x}\right)+ \frac{1}{2x+1} - \frac{1}{2x-2} > 0
        \]
    \end{lemma}
\begin{proof}
    \begin{align*}
        &\alpha\left( \frac{1}{x+1} + \frac{1}{x}\right)+ \frac{1}{2x+1} - \frac{1}{2x-2} \ge \frac{2\alpha}{x+1} - \frac{3}{(2x+1)(2x-2)} \\&\ge \frac{2\alpha}{x+1}-\frac{3}{4(2x+1)}\ge \frac{2\alpha}{x+1}-\frac{3}{7(x+1)}> 0
    \end{align*}
    Where the second and the third inequality above follows since \(x\ge 3\). 
\end{proof}}
Consider a Steiner-Claw Free instance \((G= (R\cup S, E),c)\), where the Steiner nodes \(S\) consist of a path \(s_0,s_1, \dots, s_{q+1}\), and each \(s_i\in S\) is adjacent to exactly one terminal \(r_i\in R\). We let \(L\subseteq E\) denote the terminal incident edges and \(O = E\backslash L\) denote the edges between Steiner nodes. For \(e\in O\), let \(c(e) \coloneqq \frac{32}{90}\), and for \(e\in L\), let \(c(e) = 1\). Clearly, the optimal Steiner of such an instance is \(T = G\). 
Let $W^*$ be a witness tree that minimizes $\bnu_T(\W)$. Recall that we can assume $W^*$ to be laminar by Theorem~\ref{thm:edgelaminar}, with \(w\) the vector imposed on \(E\) by \(\W\).

Consider an arbitrary laminar witness tree \(W=(R,E_W)\). For terminal \(r\in R\), let \(d_r^W\) denote  the degree of \(r\) in \(W\). We know by laminarity that either, \(d^W_r > 1\) and \(r\) is adjacent to at least \(d^W_r-2\) terminals of degree \(1\), or \(r\) has degree \(1\). For \(i\in [q]\), if \(d^W_{r_i} >1\) we call \(r_i\) a \emph{center}, and we always call \(r_0\) and \(r_{q+1}\) centers. Note that, by laminarity of \(W\), the centers form a path in \(W\) in increasing order of their index
(note that this corresponds to the notion of center \(B_i\) subtrees found in Section~\ref{sec:nodelower}). Let \(\I(W)\subseteq \{0,\dots, q+1\}\) be the index set of the centers of \(W\), then we denote by \(\mathcal{P}(\I(W))\) the path of the centers in increasing order of index. For \(i\in \{0,\dots, q+1\}\), and  \(x_L^i, x_R^i \in \mathbb{Z}_{\ge 0}\), we define a \emph{section} \(W(r_i, x_L^i, x_R^i)\) as the star graph centered at \(r_i\) with leaves  \(r_{i+j}\) for \(j= -x_L^i, \dots, x_R^i\). Clearly, for a laminar witness tree \(W\) with center index set \(\I(W)\), there exist sections \(\{W(r_i,x_L^i,x_R^i)\}_{i\in \I(W)}\) such that \(W = \bigcup_{i\in \I(W)} W(r_i,x_L^i,x_R^i) \bigcup \mathcal{P}(\I(W))\), we say that these sections are \emph{maximal} sections such that \(W(r_i,x_L^i,x_R^i)\subseteq W\).

Given a section \(W(r_i, x_L^i, x_R^i)\)    , we define a corresponding subtree \(S(r_i, x_L^i, x_R^i) \subseteq T\) as the induced subtree on the nodes \(s_{i+j}\) for \(j= -x_L^I, \dots, x_R^i+1\) and terminals \(r_{i_j}\) for \(j=-x_L^i, \dots, x_R^i\) (if \(i=q+1\), then \(S(r_{q+1}, x_L^{q+1}, x_R^{q+1})\) obviously does not include node \(s_{q+2}\)). 
(see Figure~\ref{fig:clawlowerboundsections} for an example of a section where \(q=5\)).
Let the centers of \(W\) be indexed by \(\I(W)\), and \(W\) let contain the maximal sections \(\{W(r_i,x_L^i,x_R^i)\}_{i\in \I(W)}\). Then it is clear that we have \(T = \cup_{i\in \I(W)}S(r_i,x_L^i,x_R^i)\) and \(S(r_i,x_L^i,x_R^i)\cap S(r_j,x_L^j,x_R^j) = \emptyset\) for all \(i\neq j\in\I(W)\). 

\begin{figure}[t]
    \centering
        \begin{tikzpicture}[scale=0.99]
    \coordinate (s0) at (0,1);
    \coordinate (s1) at (2,1);
    \coordinate (s2) at (4,1);
    \coordinate (s3) at (6,1);
    \coordinate (s4) at (8,1);
    \coordinate (s5) at (10,1);
    \coordinate (s6) at (12,1);
    
    \coordinate (r0) at (0,0);
    \coordinate (r1) at (2,0);
    \coordinate (r2) at (4,0);
    \coordinate (r3) at (6,0);
    \coordinate (r4) at (8,0);
    \coordinate (r5) at (10,0);
    \coordinate (r6) at (12,0);
    
    \filldraw[color=black, very thick]
        (s0)circle(1.5pt)
        (s1)circle(1.5pt)
        (s2)circle(1.5pt)
        (s3)circle(1.5pt)
        (s4)circle(1.5pt)
        (s5)circle(1.5pt)
        (s6)circle(1.5pt);
    
    \draw[color=black,very thick]
        (s0) -- (s1) node[midway,above] {\(\alpha\)} 
        (s1) -- (s2) node[midway,above] {\(\alpha\)} 
        (s2) -- (s3) node[midway,above] {\(\alpha\)} 
        (s3) -- (s4) node[midway,above] {\(\alpha\)} 
        (s4) -- (s5) node[midway,above] {\(\alpha\)} 
        (s5) -- (s6) node[midway,above] {\(\alpha\)} 
        
        (r0) -- (s0) node[midway,right] {1} 
        (r1) -- (s1) node[midway,right] {1} 
        (r2) -- (s2) node[midway,right] {1} 
        (r3) -- (s3) node[midway,right] {1} 
        (r4) -- (s4) node[midway,right] {1} 
        (r5) -- (s5) node[midway,right] {1}
        (s6) -- (r6) node[midway,left] {1}
        ;
        
    \draw[color=blue, very thick]
        (s3) -- (s4) 
        (s4) -- (s5) 
        (s5) -- (s6) 
        (r3) -- (s3) 
        (r4) -- (s4) 
        (r5) -- (s5) 
        ;
    
    \draw[color=red, very thick]
        (r0) edge[bend right=15] (r2)
        (r1) edge (r2)
        ;
    \draw[color=red, very thick, dashed]
        (r2) edge[bend right=15](r4)
        (r3) edge (r4)
        (r4) edge (r5)
        (r4) edge[bend right=15] (r6);
        
    \draw[color=black,very thick]
        (r0) node[draw=black, fill=black!0]{}
        (r1) node[draw=black, fill=black!0]{}
        (r2) node[draw=black, fill=black!0]{}
        (r3) node[draw=black, fill=black!0]{}
        (r4) node[draw=black, fill=black!0]{}
        (r5) node[draw=black, fill=black!0]{}
        (r6) node[draw=black, fill=black!0]{};
        
\end{tikzpicture}
        \caption{
        Depiction of the lower bound instance with sections for witness tree \(W\) marked in red edges. \(q=5\). Centers \(r_0\), \(r_2\), \(r_4\) and \(r_6\). There are sections \(W(r_0,0,0)\), \(W(r_2,1,0)\), \(W(r_4,1,1)\), and \(W(r_6,0,0)\). 
        The section \(W(r_4,1,1)\) is the red dashed edges. The subtree \(S(r_4,1,1)\) is the blue edges.
        }
    \label{fig:clawlowerboundsections}
\end{figure}
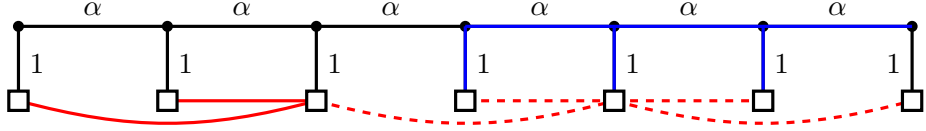

The following lemma will be useful in to allow us to replace sections of a witness tree and guarantee connectivity is maintained.

\begin{lemma}
\label{lem:sections}
    Consider sections \(W(r_i, x_L^i, x_R^i)\), \(W(r_j, x_L^j, x_R^j)\), and \(W(r_k, x_L^k, x_R^k)\), furthermore, let \(W(r_i, x_L^i, x_R^i) \subseteq W\) be a maximal section. Let the centers of \(W\) be indexed by \(\I(W)\). 
    Then 
    \[  
        W' = \bigcup_{\iota\in \I(W)\backslash\{i\}} W(r_\iota, x_L^\iota, x_R\iota) \bigcup W(r_j, x_L^j, x_R^j)\bigcup W(r_k, x_L^k, x_R^k)\bigcup \mathcal{P}(\I(W'))
    \]
    where \(\I(W') = \I(W) \cup\{j,k\}\backslash\{i\}\),  is a feasible witness tree if:\\
    1) \(r_j,r_k \in \{r_{i - x_L^i}, \dots, r_{i + x_R^i}\}\); 2) \(j+x_R^j +1 = k-x_L^k\); 3) \(i - x_L^i = j -x_L^j\), and; 4) \(i + x_R^i = k + x_R^k\)
\end{lemma}
\begin{proof}
    To see this claim, we need to show that \(W(r_j, x_L^j, x_R^j) \cup W(r_k, x_L^k, x_R^k)\cup \{r_jr_k\}\) is a tree over the same nodes as \(W(r_i, x_L^i, x_R^i)\). Clearly, \(r_j\) and \(r_k\) together are adjacent to every \(r\in \{r_{j-x_L^j}, \dots, r_{j+x_R^j}\} \cup \{r_{k-x_L^k}, \dots, r_{k+x_R^k}\} = \{r_{i - x_L^i}, \dots, r_{i + x_R^i}\}\),  and, we can see that \(\{r_{j-x_L^j}, \dots, r_{j+x_R^j}\} \cap \{r_{k-x_L^k}, \dots, r_{k+x_R^k}\} = \emptyset\). 
\end{proof}

Recall that the edges between Steiner nodes are denoted \(O\) and have cost \(\alpha=\frac{32}{90}\), and the terminal incident edges are denoted \(L\) and have cost \(1\). To prove Theorem~\ref{thm:clawlowerbound} we will first prove some useful facts about the maximal sections \(W(r_i, x_L^i, x_R^i) \subseteq \W\).  We show the following useful lemma about the maximal sections of \(\W\) with centers \(r_0\) and \(r_{q+1}\).
\begin{lemma}   
    \(W(r_0,0,0)\) and \(W(r_{q+1},0,0)\) are maximal sections of \(\W\).
\end{lemma}
\begin{proof}    
    Assume that maximal section \(W(r_0,0,x_R^0)\subseteq \W\) has \(x_R^0 > 0\). So \(r_0\) is adjacent to non-center terminals \(r_1,\dots, r_{x_R^0}\) in \(\W\). We apply Lemma~\ref{lem:sections} and consider witness tree 
    \begin{align*}
        W' \coloneqq \bigcup_{i\in \I(W)\backslash\{0\}} W(r_i,x_L^i,x_R^i) \bigcup W(r_0,0,0) \bigcup W(r_1,0,x_R^0-1)\bigcup \mathcal{P}(\I(W)\cup\{1\})
    \end{align*}
    Let \(w'\) be the vector imposed on \(E\) by \(W'\). It is clear that \(w(s_0r_0) = w(s_0s_1) = w'(s_1r_1) = x_R^0+1\), and \(w'(s_0r_0) = w'(s_0s_1) = w(s_1r_1) = 1\), and for all other \(e \in E\backslash\{s_0r_0, s_0s_1, s_1r_1\}\), \(w(e) = w'(e)\).
    Thus the difference between \(h_{\W}(T)\) and \(h_{W'}(T)\) is:
    \begin{align*}
        &H_{w(s_0r_0)} - H_{w'(s_0r_0)} + H_{w(s_1r_1)}- H_{w'(s_1r_1)} + \alpha(H_{w(s_0s_1)} - H_{w'(s_0s_1)}) \\
        =& H_{x_R^0+1} - 1 + 1 - H_{x_R^0+1} + \alpha(H_{x_R^0+1} - 1)  =  \alpha( H_{x_R^0+1} - 1) > 0
    \end{align*}
    Thus, \(\bnu_T(\W)\) can be reduced if \(x_R^0 > 0\), contradicting the assumption that \(\bnu_T(\W)\) is minimum. Demonstrating that \(x^{q+1}_L = 0\) in \(\W\) can be shown symmetrically. 
\end{proof}
For all future witness trees we consider we will assume that \(r_0\) and \(r_{q+1}\) are centers. Thus, we can see \(h_{\W}(S(r_0,0,0)) = c(s_0r_0)H_{w(s_0r_0)}+c(s_0s_1)H_{w(s_0s_1)} = 1 +\alpha\) and \(h_{\W}(S(r_{q+1},0,0)) = c(s_{q+1}r_{q+1}) H_{w(s_{q+1}r_{q+1})} = 1\).

We can now state a lemma that provides a general formula for \(h_W(S(r_i,x^i_L,x^i_R))\), \(i\in [q]\), where \(W(r_i,x_L^i,x_R^i)\subseteq W\) is a maximal section.
\begin{lemma}  
\label{lem:formula}
    Let \(W\) be a laminar witness tree. For \(i\in [q]\), let  \(W(r_i,x^i_L,x^i_R) \subseteq W\) be a maximal section. We have
    \begin{align*}
        h_W (S(r_i, x_L^i, x_R^i)) =\alpha \left( \sum_{j=2}^{x^i_L+1} H_{j} + \sum_{j=1}^{x^i_R+1} H_{j} \right) + x^i_L + x^i_R + H_{x^i_L + x^i_R +2} 
    \end{align*}
\end{lemma}
\begin{proof}
    Let \(w\) be the vector imposed on \(E\) by \(W\). We first consider edge \(e\in L\cap S(r_i,x_L^i, x_R^i)\). Clearly, \(w(e) = x_L^i + x_R^i +2\) if \(e\) is incident to a center terminal, and \(w(e)=1\) otherwise.

    Now consider edge \(e= s_{i+j}s_{i+j+1}\in O\cap S(r_i,x_L^i,x_R^i)\). For \(j = 0,\dots,x_R^i\), we know that \(e\) is on the \(r_i\textrm{-} r_{i+k}\) path in \(T\) for \(k=j+1,\dots, x_R^i\), and the path between the endpoints of an edge in \(\mathcal{P}(\I(W))\). Since \(W\) is laminar, we know that that these are the only edges of \(W\) with \(e\) on the path between their endpoints in \(T\). Therefore, \(w(e) = x_R^i +1 - j\). Similarly, for \(j=-1,\dots, -x_L^i\), we can see that \(w(e) = x_L^i +2 + j\). Therefore,
    \begin{align*}
        \sum_{j=-1}^{-x_L^i} H_{w(s_{i+j}s_{i+j+1})} + \sum_{j=0}^{x_R^i} H_{w(s_{i+j}s_{i+j+1})} = \sum_{j=2}^{x_L^i+1} H_j + \sum_{j=1}^{x_R^i +1} H_j 
    \end{align*}    
\end{proof}
We now show that the center of every maximal section of \(\W\) is, in some sense, in the ``middle'' of its terminals. 
\begin{lemma}
    \label{lem:sectionmid}
    Let \(W(r_i,x^i_L,x^i_R) \subseteq \W\) be a maximal section. Then \(|x^i_L-x^i_R| \leq 1\). 
\end{lemma}
\begin{proof}
    Assume there is a maximal section \(W(r_I, x_L^i, x_R^i) \subseteq \W\) such that \(|x_L^i - x_R^i| >1\). Without loss of generality we assume that \(x_R^i > x_L^i+1\), the other case can be handled similarly. Consider witness tree that removes the section \(W(r_i,x_L^i,x_R^i)\) from \(\W\) and adds the section \(W(r_{i+1}, x_L^i+1, x_R^i-1) \) in its place, \(W' \coloneqq \bigcup_{\iota\in \I(W)\backslash\{i\}}  W(r_\iota, x_L^\iota, x_R^\iota) \bigcup W(r_{i+1}, x_L^i+1, x_R^i-1) \bigcup \mathcal{P}(\I(W)\cup\{i+1\}\backslash{i}) \). By Lemma~\ref{lem:formula} we have 
    \begin{align*}
        h_{W'}(S(r_{i+1},x_L^i+1, x_R^i-1)) = \alpha \left( \sum_{j=2}^{x_L^i+2} H_j + \sum_{j=1}^{x_R^i}H_j\right) + x_L^i + x_R^i + H_{x_L^i + x_R^i +2}
    \end{align*}
    Therefore, the difference between \(h_{\W} (S_{r_i},x_L^i, x_R^i)\) and \(h_{W'}(S(r_{i+1},x_L^i+1, x_R^i-1))\) is
    \begin{align*}
        h_{\W}(S(r_i,x_L^i, &x_R^i))- h_{W'}(S(r_{i+1},x_L^i+1, x_R^i-1))\\
        &= \alpha\left( \sum_{j=2}^{x^i_L+1} H_{j} + \sum_{j=1}^{x^i_R+1} H_{j} -\sum_{j=2}^{x_L^i+2} H_j - \sum_{j=1}^{x_R^i}H_j\right)\\
        &= \alpha(H_{x_R^i +1} - H_{x_L^i+2} ) > \alpha (H_{x_L^i+2} - H_{x_L^i + 2}) = 0
    \end{align*}
    Therefore, \(\bnu_T(W') < \bnu_T(\W)\), contradicting our assumption on \(\W\). 
\end{proof}
For every maximal section \(W(r_i, x_L^i,x_R^i) \subseteq \W \) we can assume without loss of generality that \(x_R^i \ge x_L^i\). To see this, suppose \(x_R^i < x_L^i\), by Lemma~\ref{lem:sectionmid}, \(x_R^i +1 = x_L^i\). Consider the witness tree \(W'\coloneqq \bigcup_{\iota \in \I(W)\backslash\{i\}} W(r_\iota, x_L^\iota, x_R^\iota) \cup \W(r_{i+1}, x_L^i-1, x_R^i+1) \bigcup \mathcal{P}(\I(W)\cup \{i+1\}\backslash\{i\})\). By Lemma~\ref{lem:formula}, we see that \(\bnu_T(\W) = \bnu_T(W')\), so can consider \(W'\) instead of \(\W\).

\begin{lemma}
\label{lem:t6}
    Let \(\W(r_i, x^i_L, x^i_R) \subseteq \W\) be a maximal section. Then \(x^i_L + x^i_R + 1 \leq 5\).
\end{lemma}
\begin{proof}
    We assume for the sake of contradiction that \(x_L^i + x_R^i + 1 \ge 6\). By Lemma~\ref{lem:sections}, consider witness tree \(W'\) that removes the section \(W(r_i, x_L^i, x_R^i)\) from \(\W\) and replaces it with sections \(W(r_{i-1},x_L^i-1, x_R^i-2)\) and \(W(r_{i+x_R^i-1}, 1,1)\). Let \(\I(W') = \I(W) \cup  \{i-1,i+x_R^i-1\}\backslash\{i\}\). That is, \(W'\) is equal to 
    \[
        \bigcup_{\iota \in \I(W)\backslash\{i\}} W(r_\iota, x_L^\iota, x_R^\iota) \cup W(r_{i-1},x_L^i-1, x_R^i-2) \cup W(r_{i+x_R^i-1}, 1,1) \bigcup \mathcal{P}(\I(W'))
    \] 
    By Lemma~\ref{lem:formula}, we can see that
    \begin{align*}
        h_{W'} (S(r_{i-1},& x_L^i-1, x_R^i-2)) + h_{W'}(S(r_{i+x_R^i-1}, 1, 1))\\
        & = \alpha\left( 2H_2 + H_1 + \sum_{j=2}^{x_L^i}H_j + \sum_{j=1}^{x_R^i-1}H_j \right) + x_L^i + x_R^i -1 + H_{x_L^i + x_R^i-1} + H_4
    \end{align*}
    Therefore, the difference between \(h_{\W}(T)\) and \(h_{W'}(T)\) is
    \begin{align*}
        &\alpha \left( \sum_{j=2}^{x^i_L+1} H_{j} + \sum_{j=1}^{x^i_R+1} H_{j} -\left( 2H_2 + H_1 + \sum_{j=2}^{x_L^i}H_j + \sum_{j=1}^{x_R^i-1}H_j \right)\right)\\
        & + x_L^i + x_R^i + H_{x_L^i + x_R^i +2} - (x_L^i + x_R^i -1 + H_{x_L^i + x_R^i-1} + H_4)\\
        =& \alpha (H_{x_L^i+1} + H_{x_R^i} + H_{x_R^i +1} - 1 - 2H_2) + 1+ H_{x_L^i + x_R^i+2} - H_{x_L^i + x_R^i-1} -H_4
    \end{align*}
    We denote this above difference by \(P(x_L^i, x_R^i)\). We will show that \(P(x_L^i, x_R^i) > 0\) for all \(x_L^i + x_R^i> 5\), contradicting the assumption that \(\W\) minimizes \(\bnu_T(\W)\). We proceed by induction on \(x_L^i + x_R^i\). Recall that we assume \(|x_L^i - x_R^i| \le 1\).

    For our base case, we assume \(x_R^i = 3 \ge x_L^i\). Consider the following cases for the value of \(x_L^i\).
    \begin{enumerate}
        \item Case: \(x_L^i = 2\).
        \begin{align*}
            P(2,3) = \alpha\left( 2H_{3} + H_4 - 1 - 2H_2 \right) + 1 + H_{7} - 2H_{4} = 61/1260 > 0
        \end{align*}
        \item Case: \(x_L^i =3\).
        \begin{align*}
            P(3,3) = \alpha\left( 2H_{4} + H_3 - 1 - 2H_2 \right) + 1 + H_{8} - H_{5} - H_{4} = 157/2520 > 0
        \end{align*}
    \end{enumerate}
    So our base case holds. 

    Our inductive hypothesis is to assume the inequality holds for \(x_L^i + x_R^i = k \ge 6\). We will show the claim holds when \(x_L^i + x_R^i = k+1\). Since we showed the base case for \(x_R^i = 3\) and \(x_L^i \in \{2,3\}\), we can assume that \(\max\{x_R^i, x_L^i\}\ge 4\). We will show that \(P(x_L^i, x_R^i) > P(x_L^i, x_R^i-1)\), and by the inductive hypothesis, this will show that \(P(x_L^i, x_R^i) > 0\) and the claim will be proven.    

    The difference between \(P(x_L^i, x_R^i)\) and \(P(x_L^i, x_R^i-1)\) is
    \begin{align*}
        &\alpha(H_{x_R^i+1} - H_{x_R^i-1}) + H_{x_L^i + x_R^i+2} - H_{x_L^i + x_R^i-1} - H_{x_L^i + x_R^i+1} + H_{x_L^i + x_R^i-2}\\
        =& \alpha \left( \frac{1}{x_R^i +1} + \frac{1}{x_R^i} \right) + \frac{1}{x_L^i+x_R^i +2}- \frac{1}{x_L^i+x_R^i -1}\\
        \ge &\alpha \left( \frac{1}{x_R^i +1} + \frac{1}{x_R^i} \right) +  \frac{1}{2x_R^i +1}- \frac{1}{2x_R^i -2}
    \end{align*}
    Where the last inequality follows since \(x_L^i \ge x_R^i -1\). Applying Lemma~\ref{lem:useful6} we see the difference is strictly positive, and thus the claim holds 
\end{proof}

With Lemma~\ref{lem:t6}, we can see that for any maximal section \(W(r_i,x_L^i,x_R^i)\subseteq \W\), we have \(x_L^i+x_R^i+1\in \{1,2,3,4,5\}\). We consider the value of \(\sum_{e\in S(r_i,x_L^i,x_R^i)}\frac{c(e)H_{w(e)}}{c(S(r_i,x_L^i,x_R^i))} \), for each case of \(x_L^i+x_R^i+1\), where \(c(S(r_i,x_L^i,x_R^i)) = \sum_{e\in S(r_i,x_L^i,x_R^i)} c(e)\)
\begin{enumerate}
    \item  \(\frac{h_{\W}(S(r_i,0,0))}{\alpha+1} = \frac{H_2 + \alpha}{\alpha+1} = \frac{167}{122} \approx 1.3688\)
    \item  \(\frac{h_{\W}(S(r_i,0,1))}{2(\alpha+1)} = \frac{\alpha(H(2)+1) + H(3)+1}{2(\alpha+1)} = \frac{335}{244} \approx 1.373\)
    \item  \(\frac{h_{\W}(S(r_i,1,1))}{3(\alpha+1)} = \frac{\alpha(2H(2)+1) + H(4)+2}{3(\alpha+1)} = \frac{991}{732} \approx 1.3538\)
    \item  \(\frac{h_{\W}(S(r_i,2,1))}{4(\alpha+1)} = \frac{\alpha(H(3)+2H(2)+1)+H(5)+3}{4(\alpha+1)} = \frac{3793}{2928} \approx 1.3568\)
    \item  \(\frac{h_{\W}(S(r_i,2,2))}{5(\alpha+1)} = \frac{\alpha(2H(3)+2H(2)+1)+H(6)+4}{5(\alpha+1)} = \frac{991}{732}\approx 1.3538\)
\end{enumerate}
Let \(E_1 = E \setminus ( S(r_0,0,0)\cup S(r_{q+1},0,0))\), it is clear that \( \frac{\sum_{e\in E_1}c(e)H_{w(e)}}{\sum_{e\in E_1} c(e)} \ge \frac{991}{732}\). Thus, since \(\frac{\sum_{e\in E_1} c(e)}{\sum_{e\in E}c(e)} = \frac{q(1+\alpha)}{(q+2)(1+\alpha) - \alpha}> \frac{q}{q+2}\), we can see, for \(q> \frac{2}{\varepsilon}\), that \(\bnu_T(\W) \ge \frac{991}{732} (1- \varepsilon)\).
\end{appendix}
\end{document}